\documentclass[aps,prd,twocolumn,showpacs]{revtex4}
\usepackage{amsmath,amssymb,amsfonts,amsthm, bbm}
\usepackage{latexsym}
\usepackage{graphicx}
\usepackage{longtable}
\usepackage{floatflt}
\usepackage{pstricks}
\usepackage{pst-plot}
\usepackage[all]{xy}
\usepackage{subfigure}

\begin{document}


\newcommand{\Dr}{\Delta_r}
\newcommand{\Dt}{\Delta_\theta}
\newcommand{\D}{\mathcal{D}}
\newcommand{\al}{\bar a}
\newcommand{\la}{\bar \Lambda}
\newcommand{\ka}{\kappa}
\newcommand{\Dx}{\Delta_{\bar r}}
\renewcommand{\r}{\bar r}

\newcommand{\comment}[1]{}
\newtheorem{Thm}{Theorem}

\title{Analytical solution of the geodesic equation in Kerr-(anti) de Sitter space-times}

\author{Eva Hackmann}
\email{hackmann@zarm.uni-bremen.de}
\affiliation{ZARM, University of Bremen, Am Fallturm, 28359 Bremen, Germany}

\author{Valeria Kagramanova}
\email{kavageo@theorie.physik.uni-oldenburg.de}
\affiliation{Institut f{\"u}r Physik, University of Oldenburg, 26111 Oldenburg, Germany}

\author{Jutta Kunz}
\email{kunz@theorie.physik.uni-oldenburg.de}
\affiliation{Institut f{\"u}r Physik, University of Oldenburg, 26111 Oldenburg, Germany}

\author{Claus L{\"a}mmerzahl}
\email{laemmerzahl@zarm.uni-bremen.de}
\affiliation{ZARM, University of Bremen, Am Fallturm, 28359 Bremen, Germany}

\date\today


\begin{abstract}
The complete analytical solutions of the geodesic equations in Kerr-de Sitter and Kerr-anti-de Sitter space-times are presented. They are expressed in terms of Weierstrass elliptic $\wp$, $\zeta$, and $\sigma$ functions as well as hyperelliptic Kleinian $\sigma$ functions restricted to the one-dimensional $\theta$-divisor. We analyse the dependency of timelike geodesics on the parameters of the space-time metric and the test-particle and compare the results with the situation in Kerr space-time with vanishing cosmological constant. Furthermore, we systematically can find all last stable spherical and circular orbits and derive the expressions of the deflection angle of flyby orbits, the orbital frequencies of bound orbits, the periastron shift, and the Lense-Thirring effect. 
\end{abstract}


\maketitle


\section{Introduction and motivation}
All observations in the gravitational domain can be explained by means of Einstein's General Relativity. While for small scale gravitational effects (e.g. in the solar system) the standard Einstein field equations are sufficient, a consistent description of large scale obervations like the accelerated expansion of the universe can be achieved by the introduction of a cosmological term into the Einstein field equation
\begin{equation}
R_{\mu\nu} - \frac{1}{2} R g_{\mu\nu} + \Lambda g_{\mu\nu} = \kappa T_{\mu\nu} \,, \label{EinsteinEquation}
\end{equation}
where $\Lambda$ is the cosmological constant which has a value of $|\Lambda| \leq 10^{-46} \, \rm km^{-2}$ \cite{Dunkleyetal09}. 

Despite the smallness of the cosmological constant the question whether there might be measureable effects on solar system scales has attracted some attention. Within approximation schemes it was shown that all effects on this scale are too small to be detectable at present \cite{JetzerSereno06,KerrHauckMashhoon03,KagramanovaKunzLaemmerzahl06}. Nevertheless, there has been some discussion on whether the Pioneer anomaly, the unexplained acceleration of the Pioneer 10 and 11 spacecraft toward the inner solar system of $a_{\rm Pioneer} = (8.47 \pm 1.33) \times 10^{-10}\;{\rm m/s}^2$ \cite{Andersonetal02}, which is of the order of $c H$ where $H$ is the Hubble constant, may be related to the cosmological expansion and, thus, to the cosmological constant. The same order of acceleration is present also in the galactic rotation curves which astonishingly successfully can be modeled using a modified Newtonian dynamics involving an acceleration parameter $a_{\rm MOND}$ which again is of the order of $10^{-10}\;{\rm m/s^2}$ \cite{Sanders96}. Because of this mysterious coincidence of characteristic accelerations appearing at different scales and due to the fact that all these phenomena appear in a weak gravity or weak acceleration regime, it might be not clear whether current approximation schemes hold. Therefore, it is desireable to obtain analytical solutions of the equations of motion for a definite answer to these questions. 

There has been also some discussion if the cosmological constant has a measureable effect on the physics of binary systems, which play an important role in testing General Relativity. Although such an effect would be very small, it could influence the creation of gravitational waves \cite{NaefJetzerSereno09, BarabezHogan07}. In particular, the observation of gravitational waves originating from extreme mass ratio inspirals (EMRIs) is a main goal of the Laser Interferometer Space Antenna (LISA). The calculation of such gravitational waves benefits from analytical solutions of geodesic equations not only by improved accuracy, which is, in principle, arbitrary high, but also by the prospect of developing fast semi-analytically computation methods \cite{DexterAgol09}. Also, analytical solutions offer a systematic approach to determine the last stable spherical and circular orbits, which are starting points for inspirals and, thus, important for the calculation of gravitational wave templates.

Finally, for a thorough understanding of the physical properties of solutions of the gravitational field equations it is essential to study the orbits of test-particles and light rays in these space-times. On the one hand, this is important from an observational point of view, since only matter and light are observed and, thus, can give insight into the physics of a given gravitational field. On the other hand, this study is also important from a fundamental point of view, since the motion of matter and light can be used to classify a given space-time, to decode its structure and to highlight its characteristics. Furthermore, analytical solutions give a possibility to systematically study limiting cases like post-Newton, post-Schwarzschild, or post-Kerr expansions of geodesics and observables, which is also needed for a clear interpretation of the space-time.

Analytical solutions are especially useful for the analysis of the properties of a space-time not only from an academic point of view. In fact, they offer a frame for tests of the accuracy and reliability of numerical integrations due to their, in pinciple, unlimited accuracy. In addition, they can be used to sytematically calculate all observables in the given space-time with the very high accuracy needed for the understanding of some observations. In 1931 Hagihara \cite{Hagihara31} first analytically integrated the geodesic equation of test-particle motion in a Schwarzschild gravitational field. This solution is given in terms of the elliptic Weierstrass $\wp$ function. The geodesic equations in Reissner-Nordstr\"om, Kerr, and Kerr-Newman space-times have the same mathematical structure \cite{Chandrasekhar83} and can be solved analogously. For bound orbits in a Kerr space-time this has been elaborated recently \cite{Kraniotis04,FujitaHikada09}. The equations of geodesic motion in space-times with non-vanishing cosmological constant exhibit a more complicated structure. Recently two of us found the complete analytical solution of the geodesic equation in Schwarzschild-(anti-)de Sitter space-times based on the inversion problem of hyperelliptic integrals \cite{HackmannLaemmerzahl08, HackmannLaemmerzahl08b}. The equations of motion could be explicitly solved by restricting the problem to the $\theta$-divisor, an approach which was suggested by Enolskii, Pronine, and Richter who applied this method to the problem of the double pendulum \cite{EnolskiiPronineRichter03}. The mathematical tool developed in these papers was also applied to geodesic motion in higher dimensional spherically symmetric and static space-times \cite{Hackmannetal08} as well as to NUT-de Sitter and Pleba\'{n}ski-Demia\'{n}ski space-times without acceleration \cite{Hackmannetal09}.

In this paper we extend the approach developed in \cite{HackmannLaemmerzahl08, HackmannLaemmerzahl08b} to the case of the stationary and axially symmetric Kerr-(anti-)de Sitter space-times, thus generalizing the results of both \cite{HackmannLaemmerzahl08b} and \cite{FujitaHikada09}. We start with the derivation of the equation of motion for each coordinate dependent on proper time and decouple the equations for the $r$ and $\theta$ motion following an idea of Mino \cite{Mino03}. Then we discuss possible types of test-particle orbits with a focus on the influence of the cosmological constant $\Lambda$. In section \ref{section:exact solutions} we explicitly solve the equations derived before and present for the first time the complete analytical solution of the geodesic equation in Kerr-de Sitter space-time. After showing some chosen geodesics we derive the expressions of observables of particle and light trajectories. For bound orbits, the periastron advance and the Lense-Thirring effect are given in terms of the fundamental orbital frequencies.

\section{The geodesic equation}

We consider the geodesic equation
\begin{equation}
0 = \frac{d^2 x^\mu}{d\tau^2} + \left\{\begin{smallmatrix} \mu \\ \rho\sigma \end{smallmatrix}\right\} \frac{dx^\rho}{d\tau} \frac{dx^\sigma}{d\tau} \,,
\end{equation}
where $d\tau^2 = g_{\mu\nu} dx^\mu dx^\nu$ is the proper time along the geodesics and 
\begin{equation}
\left\{\begin{smallmatrix} \mu \\ \rho\sigma \end{smallmatrix}\right\} = \frac{1}{2} g^{\mu\nu} \left(\partial_\rho g_{\sigma\nu} + \partial_\sigma g_{\rho\nu} - \partial_\nu g_{\rho\sigma}\right)
\end{equation}
the Christoffel symbol, in a space-time given by the metric
\begin{align}
d\tau^2 = & \, \frac{\Dr}{\chi^2 \rho^2} \left(dt - a \sin^2\theta d\phi\right)^2 - \frac{\rho^2}{\Dr} dr^2 \nonumber\\
& \, - \frac{\Dt \sin^2\theta }{\chi^2 \rho^2} (a dt - (r^2+a^2) d\phi)^2 - \frac{\rho^2}{\Dt} d\theta^2 \,,
\end{align}
where
\begin{align}
\Dr & = \left(1-\frac{\Lambda}{3}r^2\right) (r^2+a^2) -2Mr\,, \\
\Dt & = 1 + \frac{a^2 \Lambda}{3} \cos^2\theta\,, \\
\chi & = 1 + \frac{a^2 \Lambda}{3}\,, \\
\rho^2 & = r^2+a^2\cos^2\theta
\end{align}
(in units where $c = G = 1$). This Boyer-Lindquist form of the Kerr-(anti)-de Sitter metric describes an axially symmetric and stationary vacuum solution of the Einstein equation and is characterized by $\bar M=2M$ related to the mass $M$ of the gravitating body, the angular momentum per mass $a = J/M$, and the cosmological constant $\Lambda$ . Note that this metric has coordinate singularities on the axes $\theta=0,\pi$ and on the horizons $\Dr=0$. The only real singularity is located at $\rho^2=0$, i.e. at simultaneously $r=0$ and $\theta = \frac{\pi}{2}$ assuming $a \neq 0$.

Analogously to the situation in Kerr space-time, we classify this form of the metric according to the number of (disconnected) regions where $\Dr > 0$, which depends on the parameters $\bar M$, $a$ and $\Lambda$. We speak of `slow` Kerr-de Sitter if there are two regions and of `fast` Kerr-de Sitter if there is one region where $\Dr > 0$. The limiting case where two regions are connected by a zero $\Dr$ is called `extreme` Kerr-de Sitter. Other cases are not possible, what can be seen by a comparison of coefficients in $\Dr=-\frac{\Lambda}{3}r^4+(1-\frac{\Lambda}{3} a^2)r^2 - \bar M r + a^2 = - \frac{\Lambda}{3} \prod_{i=1}^4 (r-r_i)$ where $r_i$ denote the zeros of $\Dr$. Fig.~\ref{Fig:slowfast} shows the modification of regions of slow, fast, and extreme Kerr-de Sitter with varying $\Lambda$.  

\begin{figure}
\centering
\includegraphics[width=0.23\textwidth]{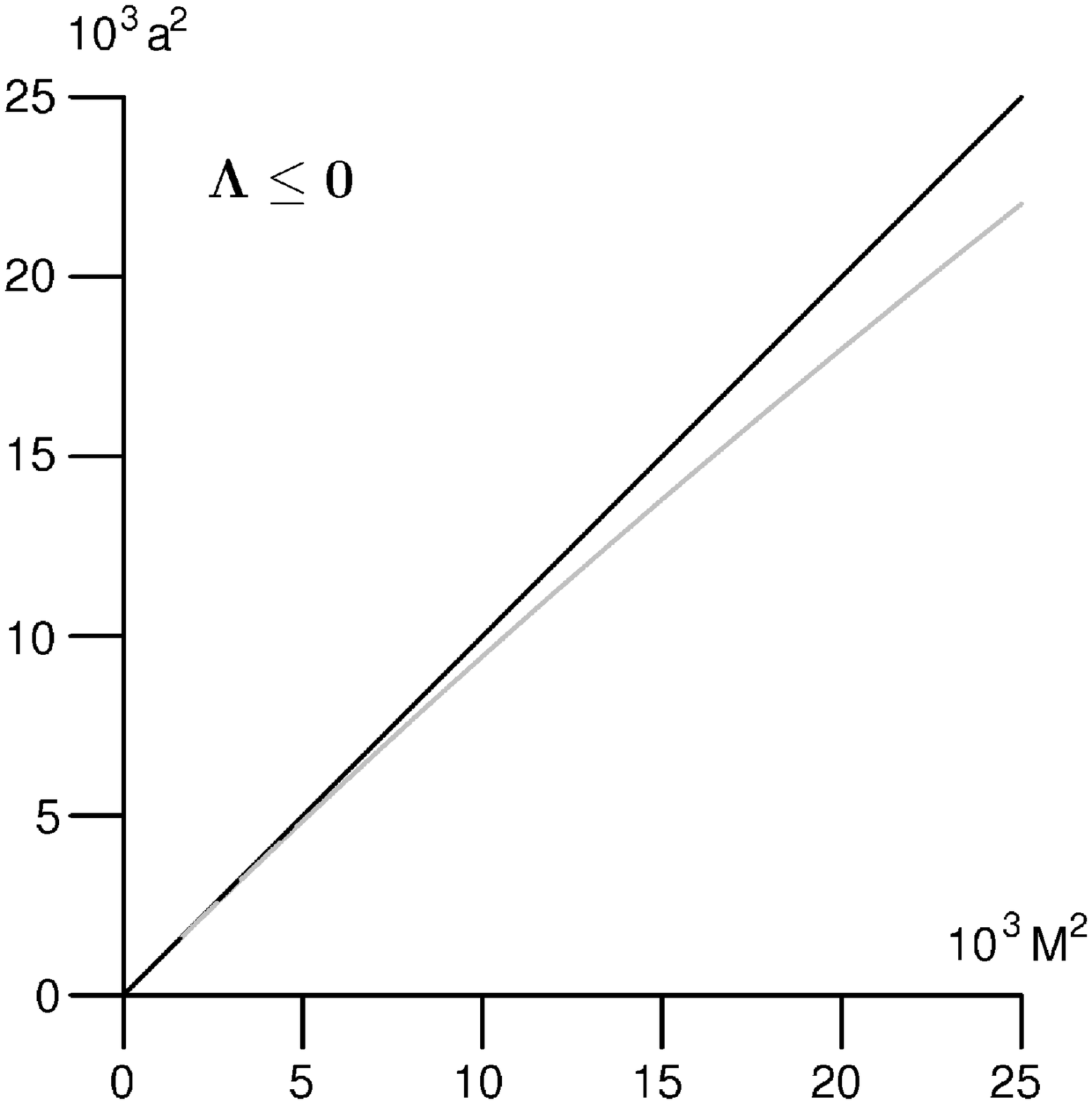}
\includegraphics[width=0.23\textwidth]{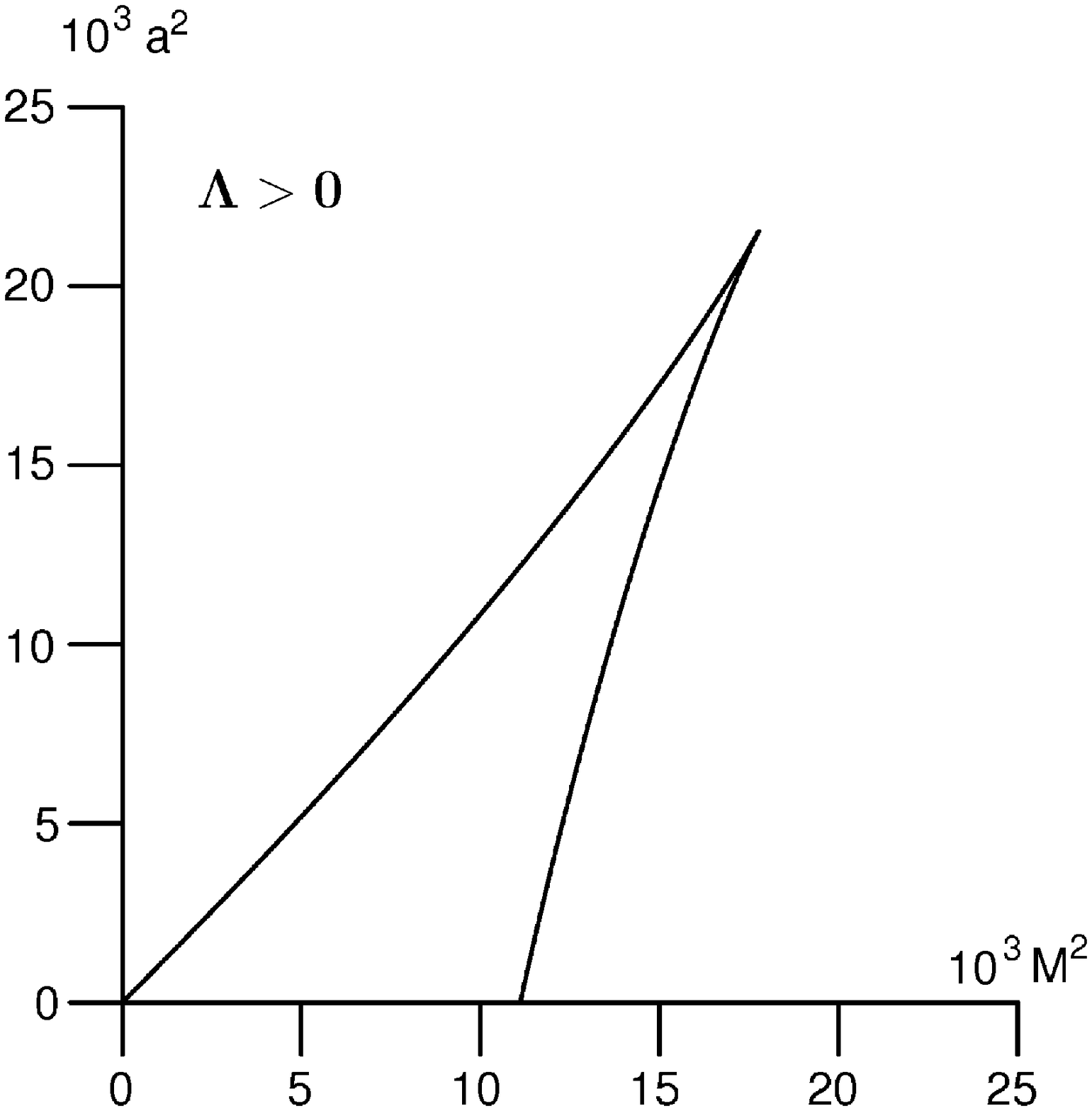}
\caption{Regions of slow and fast Kerr-de Sitter (KdS) for different values of $\Lambda$. Left: the black curve corresponds to $\Lambda=0$, the gray curve to $\Lambda=-10^{-5}$. Below the line we have slow KdS and fast above. Right: $\Lambda=10^{-5}$. The region bounded by the two curves corresponds to slow KdS, outside to fast.}
\label{Fig:slowfast}
\end{figure}

We can identify four constants of motion, two corresponding to the energy per unit mass $E$ and the angular momentum per unit mass in $z$ direction $L_z$ given by the generalized momenta $p_t$ and $p_\phi$
\begin{align}
p_t &  = g_{tt} \dot t + g_{t\phi} \dot \phi =: E \,, \\
- p_\phi &  = - g_{\phi \phi} \dot \phi - g_{t\phi} \dot t  =: L_z \,,
\end{align}
where the dot denotes a derivative with respect to the proper time $\tau$. In addition, a third constant of motion is given by the normalization condition $\delta= g_{\mu \nu} \dot x^\mu \dot x^\nu$ with $\delta=1$ for timelike and $\delta=0$ for lightlike geodesics. A fourth constant of motion can be obtained in the process of separation of the Hamilton-Jacobi equation 
\begin{equation}\label{HamiltonJacobi}
2 \frac{\partial S}{\partial \tau} = g^{ij} \frac{\partial S}{\partial x^i} \frac{\partial S}{\partial x^j}
\end{equation}
using the ansatz
\begin{equation} \label{GleichungS}
S = \frac{1}{2} \delta \tau - Et + L_z\phi + S_r(r) + S_\theta(\theta)\,.
\end{equation}
If we insert this into \eqref{HamiltonJacobi} we get
\begin{multline}
\delta a^2 \cos^2\theta + \Dt \left( \frac{\partial S_\theta}{\partial \theta} \right)^2 +  \frac{\chi^2}{\Dt \sin^2\theta} \left( a E \sin^2\theta - L_z \right)^2 = \\
- \delta r^2  - \Dr \left( \frac{\partial S_r}{\partial r} \right)^2 + \frac{\chi^2}{\Dr} \left( (r^2+a^2)E - a L_z \right)^2 ,
\end{multline}
where each side depends on $r$ or $\theta$ only. This means that each side is equal to a constant $K$, the famous Carter constant, \cite{Carter68}. 

From the separation ansatz \eqref{GleichungS} we derive the equations of motion
\begin{align}
\rho^4 \dot r^2 & = \chi^2 (E \bar M^2 \mathbbm{P})^2 - \Dr (\delta r^2 + K) \,,  \label{dot r} \\
\rho^4 \dot \theta^2 & = \Dt (K- a^2 \delta \cos^2\theta) - \frac{\chi^2 (E \bar M \mathbbm{T})^2}{\sin^2 \theta} \,, \label{dot theta} \\
\frac{\rho^2}{\chi^2} \dot \phi & = \frac{a}{\Dr} E \bar M^2 \mathbbm{P} - \frac{E \bar M}{\Dt \sin^2\theta} \mathbbm{T} \,, \label{dot phi} \\
\frac{\rho^2}{\chi^2} \dot t & = \frac{r^2+a^2}{\Dr} E \bar M^2 \mathbbm{P} - \frac{a E \bar M}{\Dt} \mathbbm{T} \,. \label{dot t}
\end{align}
where $E \bar M^2 \mathbbm{P}=(r^2+a^2) E - a L_z$ and $E \bar M \mathbbm{T}=a E \sin^2\theta - L_z$. The equations for $\dot r$ and $\dot \theta$ are coupled by $\rho^2=r^2+a^2\cos^2\theta$. This difficulty can be overcome by introducing the Mino time $\lambda$ \cite{Mino03} which is related to the proper time $\tau$ by $\frac{d\tau}{d\lambda} = \rho^2$. For simplicity, we rescale the parameters appearing in eqs.~\eqref{dot r}-\eqref{dot t} such that they are dimensionless. Thus, we introduce
\begin{align}
\r & = \frac{r}{\bar M}, \quad \al=\frac{a}{\bar M}, \quad \la =\frac{1}{3} \Lambda \bar M^2,  \nonumber \\
\bar L_z & = \frac{L_z}{\bar M}, \quad \bar K = \frac{K}{\bar M^2}, \, 
\end{align}
and accordingly
\begin{align}
\Dx & = \left( 1 - \la \r^2\right)(\r^2+\al^2) - \r \,, \quad (\Dr = \bar M^2 \Dx) \nonumber \\
\Dt & = 1 + \al^2 \la \cos^2 \theta \,,  \\
\chi & = 1+\al^2 \la \,. \nonumber
\end{align}
In addition, we can absorb $(E \bar M)$ in the definition of $\lambda$ by introducing
\begin{equation}
\D =\frac{\bar L_z}{E}\,, \quad \ka =\frac{\bar K}{E^2}\,, \quad \delta_2 = \frac{\delta}{E^2} \,, \quad \gamma = E \bar M \lambda \,.
\end{equation}
Then the equations \eqref{dot r}-\eqref{dot t} decouple and read
\begin{align}
\left( \frac{d\r}{d\gamma} \right)^2 & = R(\r) =  \chi^2 \mathbbm{P}^2 - \Dx (\delta_2 \r^2 + \ka)\,, \label{dot r_sn}\\
\left( \frac{d\theta}{d\gamma} \right)^2 & = \Theta(\theta) = \Dt (\ka - \delta_2 \al^2 \cos^2\theta) - \frac{\chi^2 \mathbbm{T}^2}{\sin^2\theta}\,, \label{dot theta_sn}\\
\frac{1}{\chi^2} \frac{d\phi}{d\gamma} & = \frac{\al}{\Dx} \mathbbm{P} - \frac{1}{\Dt \sin^2\theta} \mathbbm{T} \label{dot phi_sn} \,, \\
\frac{1}{\bar M \chi^2} \frac{dt}{d\gamma} & = \frac{\r^2+\al^2}{\Dx} \mathbbm{P}  - \frac{\al}{\Dt} \mathbbm{T} \label{dot t_sn} \,.
\end{align}
where, as before, 
\begin{align}
\mathbbm{P} & = \r^2 + \al^2 - \al \D \,,  \\
\mathbbm{T} & = \al \sin^2\theta - \D \,.
\end{align}
In section \ref{section:exact solutions} we will explicitly solve these equations.

\section{Types of timelike geodesics} \label{section:Types}

Before solving the equations of motion derived in the previous section, we analyse the structure of possible orbits dependent on the black hole parameters $\al$, $\la$ and the particle parameters $\delta, E, L_z, K$. The major point in this analysis is that \eqref{dot r_sn} and \eqref{dot theta_sn} imply $R(\r) \geq 0$ and $\Theta(\theta) \geq 0$ as a necessary condition for the existence of a geodesic. Although we concentrate here on timelike geodesics light can be treated in the same manner and is in general easier to deal with.

First of all, we state two short theorems about possible values of the Carter constant. The corresponding theorems for vanishing cosmological constant can be found in \cite{ONeill95}.
\begin{Thm}
If a geodesic lies entirely in the equatorial plane $\theta=\frac{\pi}{2}$ or if it hits the ring singularity $\rho^2=0$ then the modified Carter constant $Q = K - \chi^2 (aE-L)^2$ is zero. 
\end{Thm}
\begin{proof}
A geodesic lies entirely in the equatorial plane iff $\theta(\gamma)=\frac{\pi}{2}$ for all $\gamma$. This implies that $\Theta(\theta) = ( \frac{d\theta}{d\gamma} )^2=0$ and with
\begin{equation*}
\Theta\left(\theta=\frac{\pi}{2}\right) = \ka - \chi^2 (\al-\D)^2 = Q E^{-2} \bar M^{-2}
\end{equation*}
it follows $Q = 0$. If a geodesic hits the ring singularity, then there is a $\gamma$ such that $\r(\gamma)=0$ and $\theta(\gamma) = \frac{\pi}{2}$. As $R(\r) \geq 0$ and $\Theta(\theta) \geq 0$ for all $\r$ and $\theta$, and in particular for $\r=0$, $\theta=\frac{\pi}{2}$, it follows 
\begin{align*}
R(\r=0) & = \chi^2 \al^2 (\al-\D)^2 - \al^2 \ka = \frac{- \al^2 Q}{E^2 \bar M^2} \geq 0 \, \Rightarrow Q \leq 0
\end{align*}
and as above $\Theta\left(\theta=\frac{\pi}{2}\right) = \frac{Q}{E^{2} \bar M^{2}} \geq 0$.
\end{proof}
Note that this theorem implies that $Q=0$ is a necessary condition for equatorial orbits, which is an important class found in many astrophysical objects like accretion discs and planetary systems. Note that the modified Carter constant $Q$ depends on the cosmological constant, which also influences the next theorem.
\begin{Thm} \label{Thm2}
For $\Lambda> - \frac{3}{a^2}$ all timelike and null geodesics have $K\geq0$. In this case $K=0$ implies $Q=0$ and the geodesic lies entirely in the equatorial plane. 
\end{Thm}
\begin{proof}
A geodesic can only exist if there are values for $\r(\gamma)$ and $\theta(\gamma)$ with $R(\r) \geq 0$ and $\Theta(\theta) \geq 0$. From $\Lambda>\frac{-3}{a^2}$ it follows $\Dt = 1+ \al^2 \la \cos^2\theta>1-\cos^2\theta \geq0$. If $K<0$ then $(\ka - \delta_2 \al^2 \cos^2\theta)<0$ and
\begin{align*}
\Theta(\theta) & = \Dt ( \ka - \delta_2 \al^2 \cos^2\theta ) - \frac{\chi^2}{\sin^{2}\theta} \mathbbm{T}^2 < 0
\end{align*}
for all values of $\theta$. Assume now $K=0$. Consequently
\begin{align*}
\Theta(\theta) & = - \delta_2 \al^2 \cos^2\theta \Dt - \frac{\chi^2}{\sin^{2}\theta} \mathbbm{T}^2 \leq 0
\end{align*}
and $\Theta(\theta)=0$ only if $\cos^2\theta = 0$ and additionally $\mathbbm{T} = \al \sin^2\theta - \D = \al-\D = 0$.
\end{proof}
Since from observation the cosmological constant has a small positive value, the condition $\Lambda>\frac{-3}{a^2}$ is always fulfilled.

From these two theorems it is obvious that, while $K$ originates from the separation procedure, the modified Carter constant $Q$ has a geometric interpretation since it is related to possible $\theta$ values of the orbits. This relation will become more explicit in the following subsections.

In the remainder of the section we will study the consequences of the two conditions $R(\r) \geq 0$ and $\Theta(\theta) \geq 0$.

\subsection{Types of latitudinal motion}
Geodesics can take an angle $\theta$ if and only if $\Theta(\theta) \geq 0$. Thus, we want to determine which values of $\al$, $\la$, $E$, $\bar L_z$, $\bar K$, and $\theta \in [0,\pi]$ result in positive $\Theta(\theta)$. For simplicity, we substitute $\nu = \cos^2\theta$ giving
\begin{multline} \label{Theta(nu)}
\Theta(\nu) = (1+\al^2 \la \nu)  (\ka - \delta_2 \al^2 \nu) \\ - \chi^2 \left( \al^2 (1-\nu) - 2\al\D + \frac{\D^2}{1-\nu} \right)  \,.
\end{multline}
Assume now that for a given set of parameters there exists a certain number of zeros of $ \Theta(\nu)$ in $[0,1]$. If we vary the parameters, the position of zeros varies and the number of real zeros in $[0,1]$ can change only if (i) a zero crosses $0$ or $1$ or (ii) two zeros merge. Let us consider case (i). $0$ is a zero iff
\begin{align}
\Theta(\nu=0) = \ka - \chi^2 (\al-\D)^2 = 0
\end{align}
or
\begin{align}
\bar L_z = \al E \pm \frac{\sqrt{\bar K}}{\chi} \,.
\end{align}
As $\nu=1$ is in general a pole of $\Theta(\nu)$ it is a necessary condition for $1$ being a zero of $\Theta(\nu)$ that this pole becomes a removable singularity. From \eqref{Theta(nu)} it follows that this is the case for $\D=0$ or, equivalently, $\bar L_z=0$. Under this assumption we obtain
\begin{align}
\Theta(\nu=1) = (1+\al^2\la)(\ka-\delta_2\al^2) \qquad \text{for } \bar L_z=0\,.
\end{align}
If we additionaly assume that $\Lambda > \frac{-3}{a^2}$ (as in Thm.~\ref{Thm2}) we can conclude that $\Theta(\nu=1)=0$ iff $\bar L_z=0$ and $\ka = \delta_2 \al^2$. Summarized, $\bar L_z = \al E \pm \frac{\sqrt{\bar K}}{\chi}$ and simultaneously $\bar L_z=0$ and $\ka = \delta_2 \al^2$ (assuming $\Lambda > \frac{-3}{a^2}$) give us boundary cases of the $\theta$ motion.

Now let us consider case (ii). If we exclude the coordinate singularities $\theta=0,\pi$ or $\nu=1$ the zeros of $\Theta(\nu)$ are given by the zeros of 
\begin{align}
\Theta_\nu =  (1-\nu) (1+\al^2 \la \nu)  (\ka - \delta_2 \al^2 \nu) - \chi^2 \left(\al-\D-\al\nu\right)^2 \,, \label{def Theta_nu}
\end{align}
which is in general a polynomial of degree $3$. Then two zeros coincide at $x \in [0,1)$ iff
\begin{align}\label{doublezerotheta}
\Theta_\nu = (\nu-x)^2 (a_1\nu+a_0)
\end{align}
for some real constants $a_1, a_0$. By a comparison of coefficients we can solve this equation for $\bar L_z(x)$ and $E^2(x)$ dependent the remaining parameters $\al, \la$, and $\bar K$. This parametric representation of values of $\bar L_z$ and $E^2$ again correspond to boundary cases of the $\theta$ motion. Let us additionally consider the conditions for $\nu=1$ being a double zero. With $\bar L_z=0$ and $\ka=\delta_2\al^2$ (assuming $\Lambda > \frac{-3}{a^2}$) it follows
\begin{align}
\frac{d\Theta(\nu)}{d\nu}\Big\vert_{\nu=1} = \al^2 \chi (\chi - \delta_2) \,,
\end{align}
which is zero for $\chi=\delta_2$ or, equivalently, $E^2 = \chi^{-1}$. 

\begin{figure}
\includegraphics[width=0.23\textwidth]{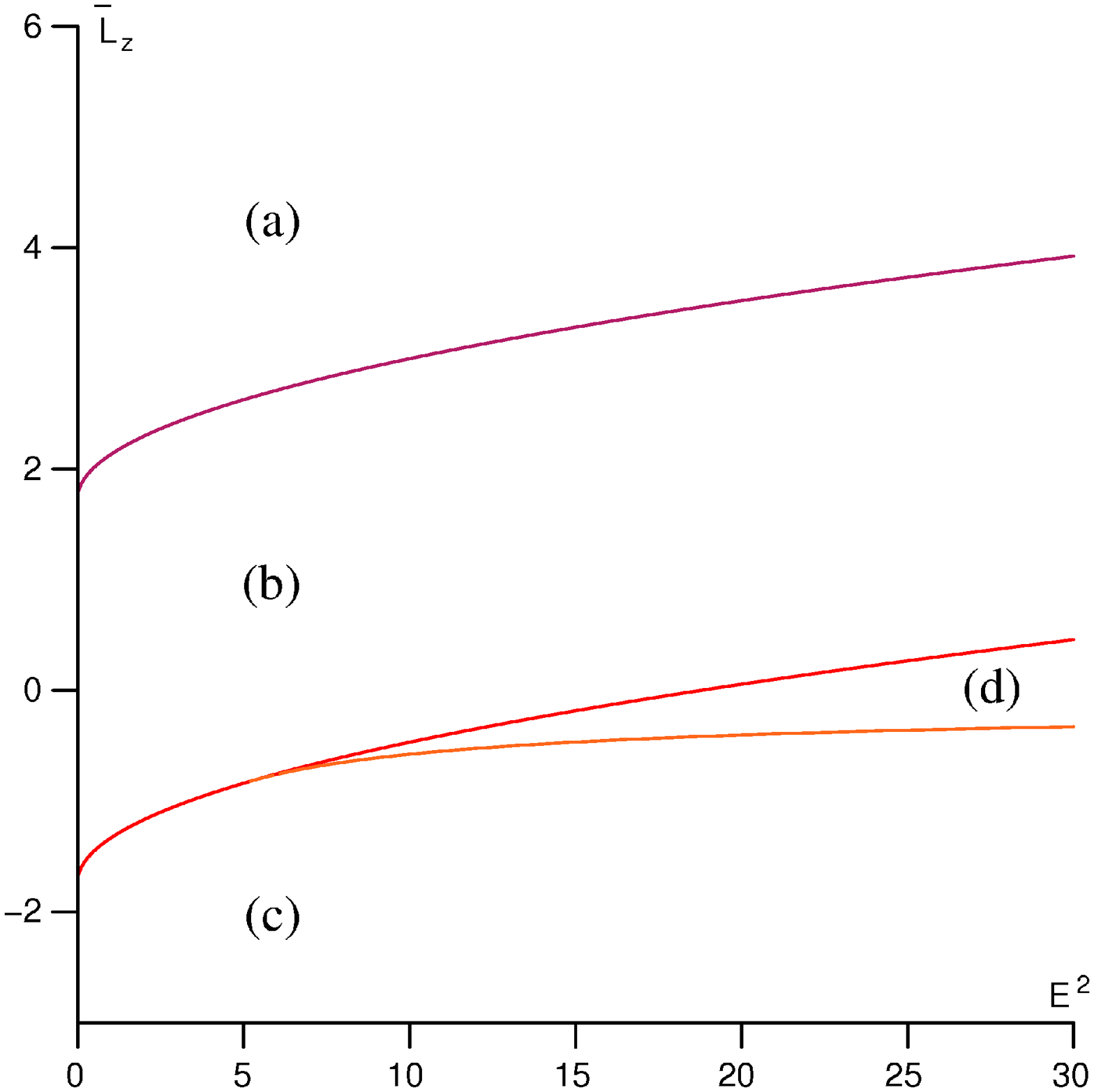}
\includegraphics[width=0.23\textwidth]{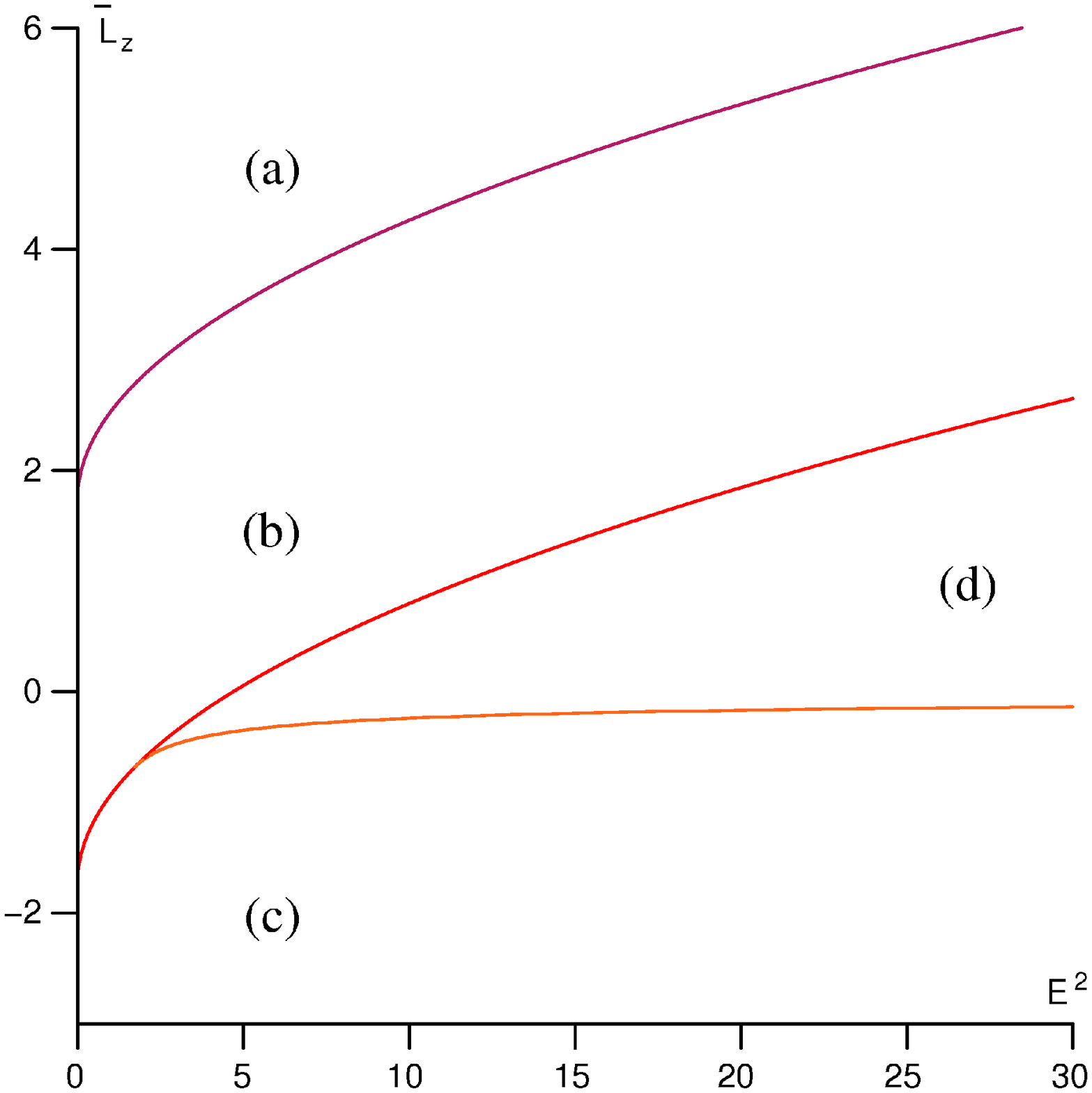}
\caption{Typical regions of different types of $\theta$-motion. Here $\bar M=2$, $\la=10^{-5}$, and $\bar K=3$. The transition from slow Kerr-de Sitter (left, $\al=0.4$) to fast Kerr-de Sitter (right, $\al=0.8$) is continuous.}
\label{Fig:theta_EL}
\end{figure} 

For given parameters of the black hole $\al$ and $\la$, we can use these informations to analyse the $\theta$ motion of all possible geodesics in this space-time. As a typical example for timelike geodesics consider Fig.~\ref{Fig:theta_EL}, where the curves divide the half plane into four regions (a)-(d) which correspond to different arrangement of zeros in $[0,1]$. A geodesic motion is only possible in regions (b) and (d) because in all other regions $\Theta_\nu$ is negative for all $\nu \in [0,1]$. Note that for the special case of $\ka=\delta_2\al^2$ (assuming $\Lambda > \frac{-3}{a^2}$) strictly speaking regions (b) and (d) are divided by $\bar L_z=0$. However, in each region we have for $\bar L_z>0$ and $\bar L_z<0$ the same number of zeros in $[0,1]$ and, thus, the same type of motion. (More precisely, near $\bar L_z=0$ a zero $\nu_0<1$ of $\Theta(\nu)$ approaches $1$, but does not cross it.) Therefore, in each region we put the parts above and below $\bar L_z=0$ together. The arrangement of zeros in the two regions (b) and (d) correspond to the following different types of motion in $\theta$ direction (cp.~Fig.~\ref{Fig:Theta}) 
\begin{itemize}
\item Region (b): $\Theta_\nu$ has one real zero $\nu_{\rm max}$ in $[0,1]$ with $\Theta_\nu \geq 0$ for $\nu \in [0,\nu_{\rm max}]$, i.e. $\theta$ oscillates around the equatorial plane $\theta=\frac{\pi}{2}$.
\item Region (d): $\Theta_\nu$ has two real zeros $\nu_{\rm min}, \nu_{\rm max}$ in $[0,1]$ with $\Theta_\nu \geq 0$ for $\nu \in [\nu_{\rm min},\nu_{\rm max}]$, i.e. $\theta$ oscillates between $\arccos(\pm\sqrt{\nu_{\rm min}})$ and $\arccos(\pm\sqrt{\nu_{\rm max}})$. 
\end{itemize}
The boundaries of region (b) are given by $\bar L_z = \al E \pm \frac{\sqrt{\bar K}}{\chi}$ and, therefore, the regions gets larger if $\frac{\sqrt{\bar K}}{\chi}$ grows, i.e. if $\bar K$ grows or $\al^2 \la$ gets smaller. A change of $\al$ in addition causes region (b) to shift up or down. The dependence of region (d) on the parameters $\bar K, \al$ and $\la$ is much more involved. The upper boundary of (d) is also the lower boundary of region (b). The lower boundary is given in a complicated parametric form which makes it apparently impossible to determine an explicit connection between the form of region (d) and the parameters. However, the point where the upper and lower boundaries of region (d) touch each other is where $0$ is a double zero of $\Theta_\nu$, which is given by $x=0$ in $E(x)$ and $\bar L_z(x)$ from \eqref{doublezerotheta},
\begin{align}
E(0) & = \frac{1}{2} \frac{\al^2 + \bar K - \al^2 \la \bar K}{\al \sqrt{\bar K} \chi}, \quad \bar L_z(0) = \frac{1}{2} \frac{\al^2-\bar K\chi}{\sqrt{\bar K} \chi}\,.
\end{align}
The regions (b) and (d) are characterized in a simple way in terms of the modified Carter constant $Q$. As in region (d) $\Theta_\nu(0) < 0$ it follows that this region corresponds to $Q<0$ because of $\Theta_\nu(0) = \frac{Q}{E^2 r_S^2}$. In the same way we can conclude that region (b), where $\Theta_\nu(0)>0$, corresponds to $Q>0$.

\begin{figure}
\includegraphics[width=0.23\textwidth]{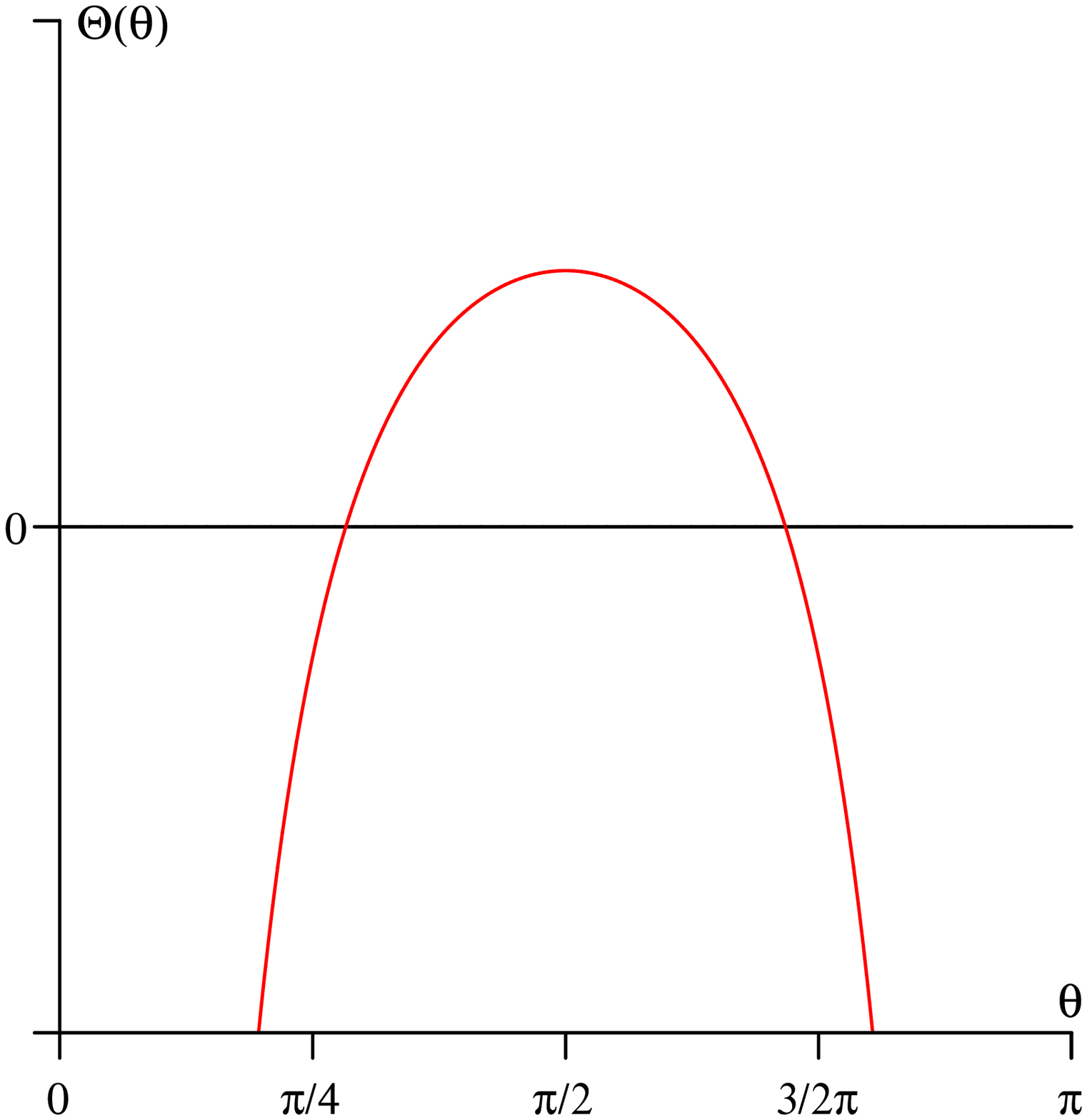}
\includegraphics[width=0.23\textwidth]{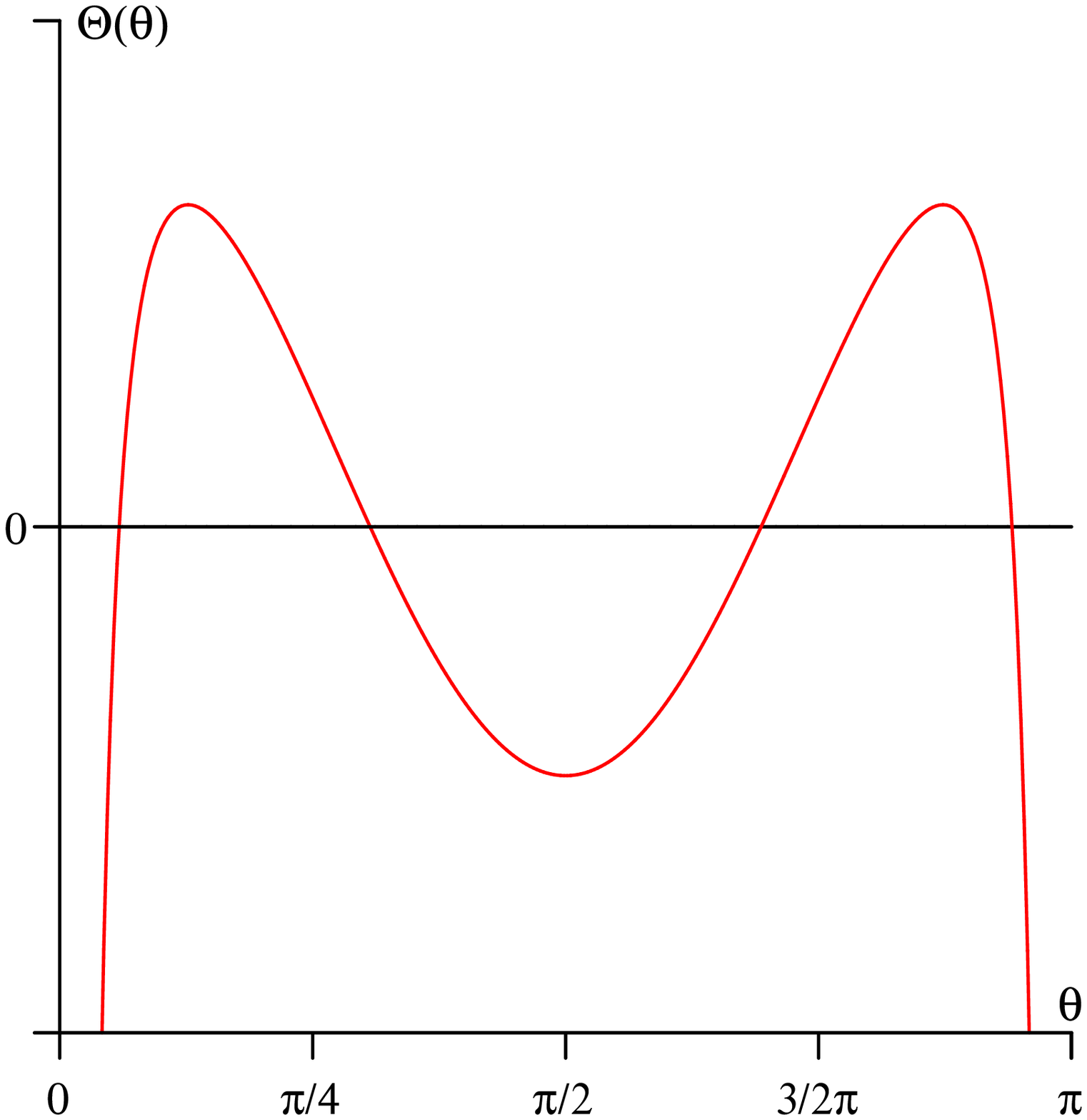}
\caption{Typical examples of $\Theta(\theta)$ in the two regions (b) (left) and (d) (right). Allowed values of $\theta$ are given by $\Theta(\theta)\geq0$.}
\label{Fig:Theta}
\end{figure}

\subsection{Types of radial motion}
A geodesic can take a radial coordinate $\r$ if and only if $R(\r)\geq 0$. The zeros of $R$ are extremal values of $\r(\gamma)$ and determine the type of geodesic. The polynomial $R$ is in general of degree six in $\r$ and, therefore, has six possibly complex zeros of which the real zeros are of interest for the type of motion. As a Kerr-de Sitter space-time has no singularity in $\r=0, \theta \neq \frac{\pi}{2}$, we can also consider negative $\r$ as valid. However, $\r=0$ is an allowed value of $\r(\gamma)$ iff 
\begin{align}
0 < R(0) & = \chi^2 ( \al^2 - \al \D )^2 - \al^2 \ka  \nonumber \\
& = -\al^2 \frac{Q}{E^2 r_S^2} = -\al^2 \Theta_\nu(0) \,.
\end{align}
It follows that $\r=0$ can only be crossed iff $Q < 0$, which corresponds to region (d) of the $\theta$ motion. In region (b) of the $\theta$ motion where $Q>0$ a transition from positive to negative $\r$ is not possible.

To clarify the discussion we introduce some types of orbits \cite{ONeill95}.
\begin{itemize}
\item Flyby orbit: $\r$ starts from $\pm \infty$, then approaches a periapsis $\r=r_0$ and back to $\pm \infty$.
\item Bound orbit: $\r$ oscillates between to extremal $r_1\leq \r \leq r_2$ with $- \infty < r_1 < r_2 < \infty$.
\item Transit orbit: $\r$ starts from $\pm \infty$ and goes to $ \mp \infty$ crossing $\r=0$.
\end{itemize}
All other types of orbits are exceptional and treated separately. They are either connected with the ring singularity $\rho^2=0$ or with the appearence of multiple zeros in $R$, which simplifies the structure of the differential equation \eqref{dot r_sn} considerably. Examples for the latter type are homoclinic orbits, cp.~\cite{LevinPerez-Giz09}. As large negative $\r$ correspond to negative mass of the black hole \cite{HawkingEllis73}, we will assign the attribute 'crossover' to flyby or bound orbits which pass from positive to negative $\r$ or vice versa. (By definition, a transit orbit is always a crossover orbit and, therefore, we will not explicitly state that.) Therefore, in region (d) of the $\theta$ motion exists a crossover orbit, whereas all orbits located in region (b) of the $\theta$ motion do not cross $\r=0$. 

For a given set of parameters we have a certain number of real zeros of $R$. If we vary the parameters this number can change only if two zeros merge to one. This happens at $\r=x$ iff
\begin{align}\label{doublezerosr}
R = (\r-x)^2 (a_4 \r^4 + a_3 \r^3 + a_2 \r^2 + a_1 \r + a_0)
\end{align}
for some real constants $a_i$. By a comparison of cofficients we can solve the resulting 7 equations for $E^2(x)$ and $\bar L_z(x)$ dependent on the remaining parameters $\al$, $\la$, and $\bar K$. A typical result in slow Kerr-de Sitter for small $\Lambda$ including the results of the foregoing subsection is shown in Figs.~\ref{Fig:r_EL} and \ref{Fig:regiond}. An analysis of the influence of each of the parameters $\al$, $\la$, and $\bar K$ is not done easily due to the complexity of the expressions for $E^2$ and $\bar L_z$. However, some typical examples for varying $\bar K$ are shown in Fig.~\ref{Fig:rtheta_EL}. Examples of $R(\r)$ for different numbers of real zeros are given in Fig.~\ref{Fig:R}.

\begin{figure}
\includegraphics[width=0.23\textwidth]{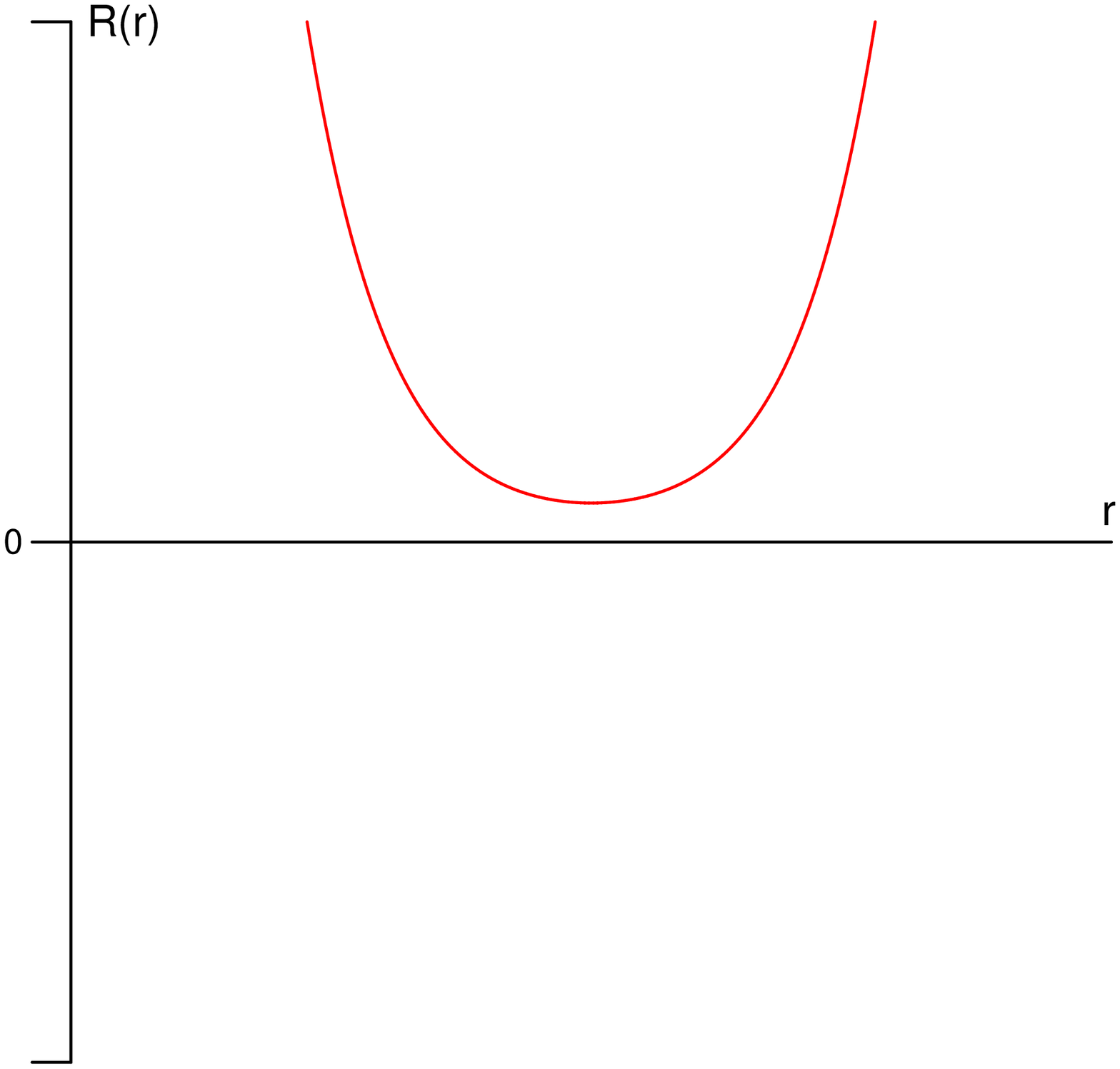}
\includegraphics[width=0.23\textwidth]{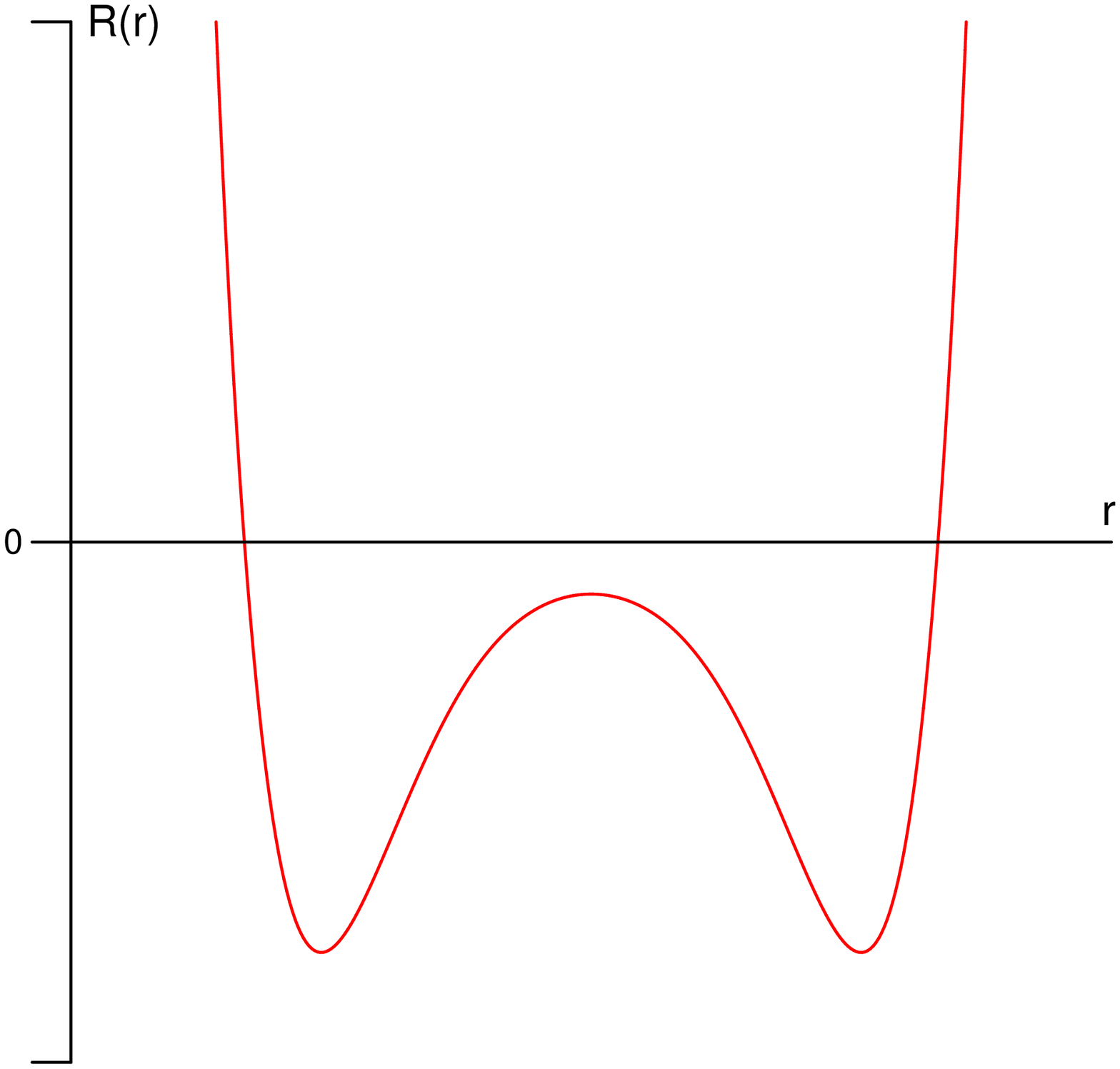}
\includegraphics[width=0.23\textwidth]{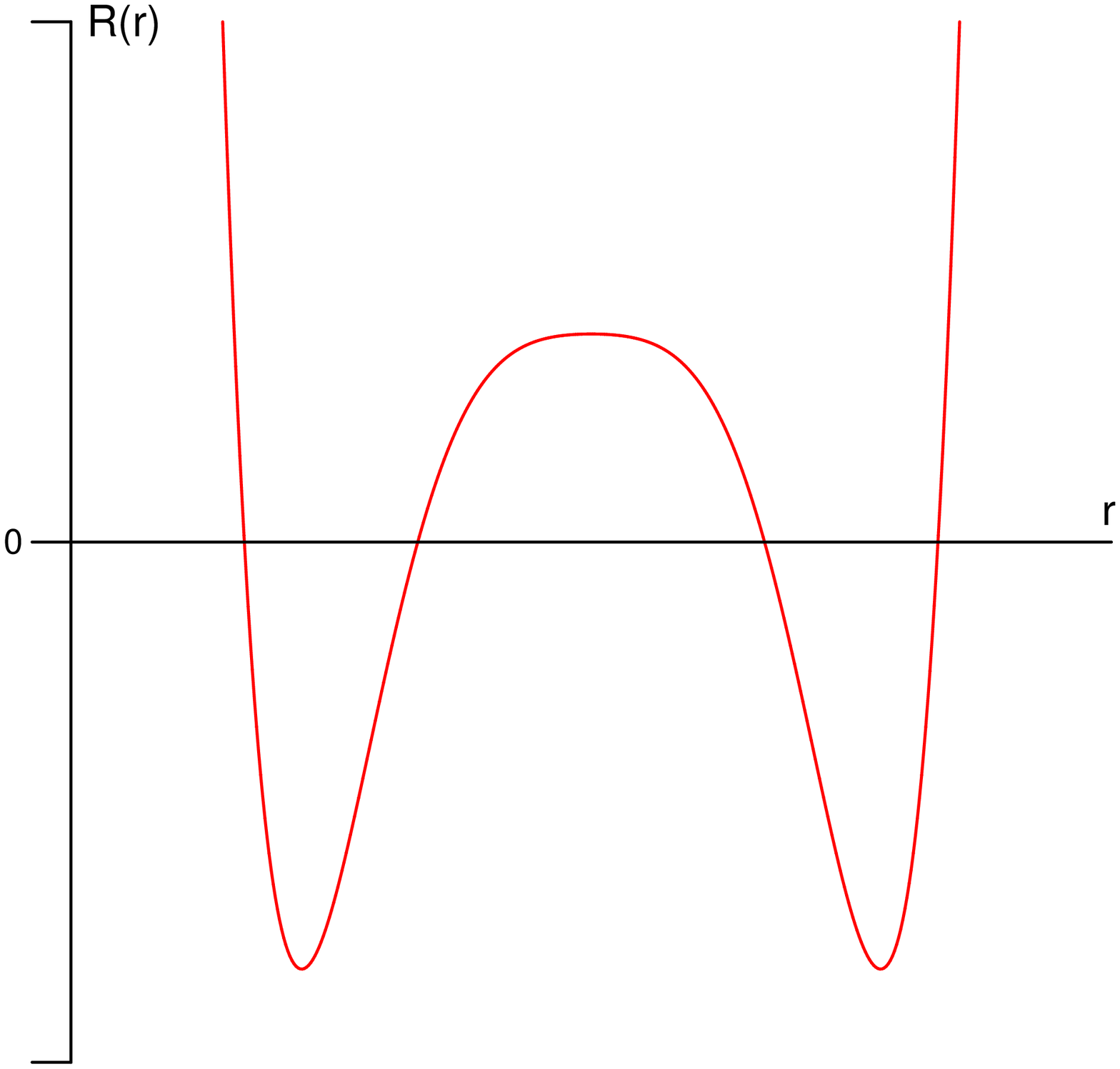}
\includegraphics[width=0.23\textwidth]{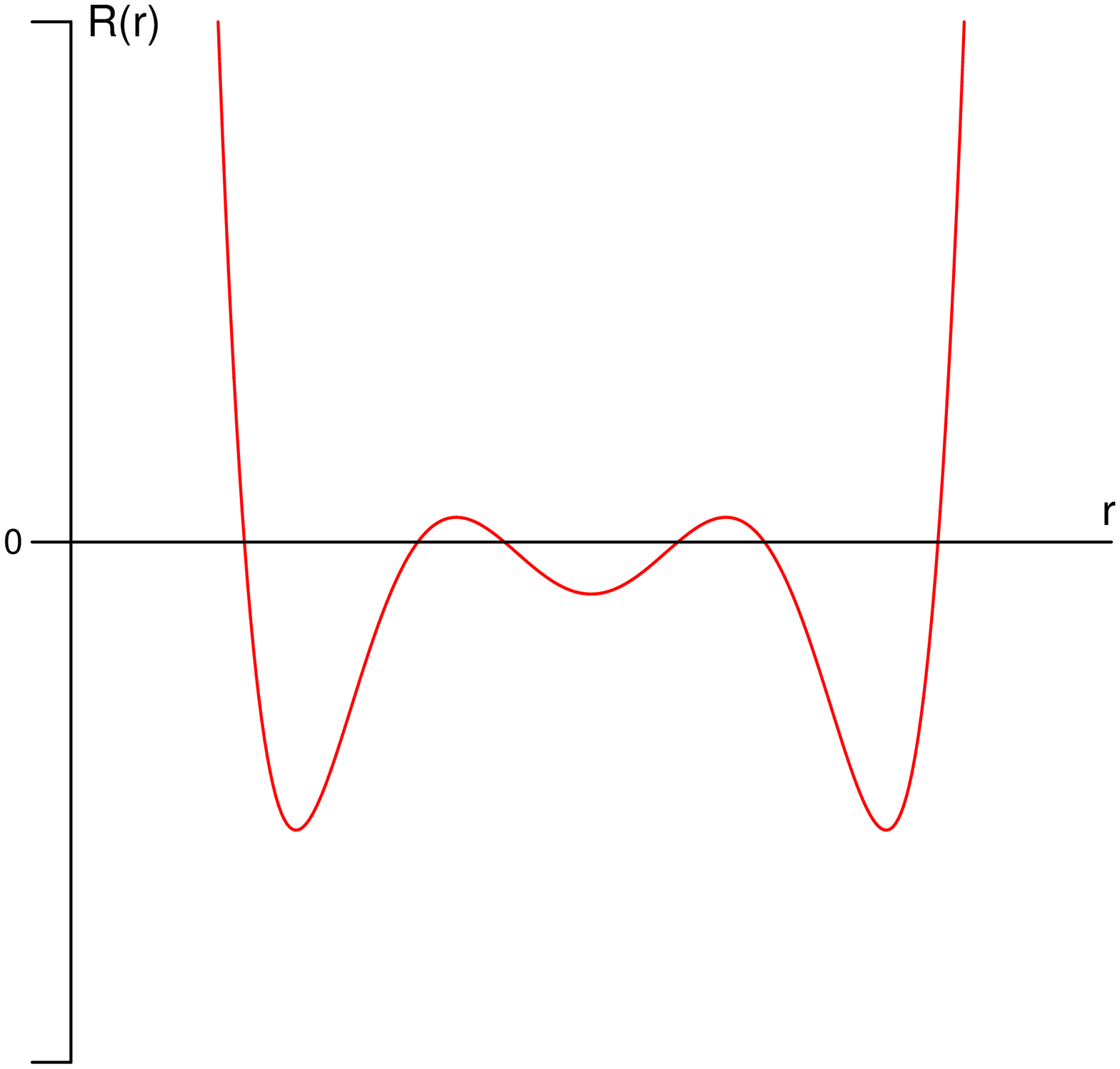}
\caption{Typical examples of $R(\r)$ for $\Lambda>0$ and for different numbers of real zeros. Allowed values of $\r$ are given by $R(\r) \geq 0$.}
\label{Fig:R}
\end{figure}

We discuss now the resulting types of orbits. For simplicity, we restrict ourselves here to the case of slow Kerr-de Sitter although fast and extreme Kerr-de Sitter can be discussed analogously. This will be explicitly carried through in a future publication. 

For comparison, let us first study the situation in slow Kerr with $\Lambda=0$.

\paragraph{Case $\Lambda=0$.}
We recognize five regions of different types of $r$ motion. (Here we always assume $r_i < r_{i+1}$.)
\begin{itemize}
\item Region (I): all zeros of $R$ are complex and $R(\r) \geq 0$ for all $\r$. Possible orbit types: transit orbit.
\item Region (II): $R$ has two real zeros $r_1,r_2$ and $R(\r) \geq 0$ for $\r \leq r_1$ and $r_2\leq \r$. Possible orbit types: two flyby orbits, one to $+\infty$ and one to $-\infty$.
\item Region (III): all four zeros $r_i$, $1\leq i \leq 4$, of $R$ are real and $R(\r) \geq 0$ for $r_{2k-1} \leq \r \leq r_{2k}$, $k=1,2$. Possible orbit types: two different bound orbits.
\item Region (IV): again all four zeros of $R$ are real but $R(\r) \geq 0$ for $\r \leq r_1$, $r_2 \leq \r \leq r_3$, and $r_4 \leq \r$. Possible orbit types: two flyby orbits, one to each of $\pm \infty$ and a bound orbit.
\item Region (V): $R$ has two real zeros $r_1,r_2$ and $R(\r) \geq 0$ for $r_1 \leq \r \leq r_2$. Possible orbit types: a bound orbit.  
\end{itemize}
Although there are the same number of real zeros the different orbit types in regions (III)/(IV) and (II)/(V) is due to the different behaviour of $R$ when $\r \to \pm \infty$. For $\la=0$ the expression $E^2 R = \sum_{i=1}^4 a_i \r^4$ is a polynomial of degree 4 with $a_4 = E^2-1$ which for $\r \to \pm \infty$ yields $R(\r) \to \infty$ if $E^2>1$ and $R(\r) \to - \infty$ if $E^2<1$. 

Let us also analyse where we have crossover orbits. Region (II) is the only one which intersects region (b) as well as region (d) of $\theta$ motion. As region (I) can only contain a transit orbit, which is by definition a crossover orbit, it can only intersect region (d). All other regions contain only region (b) of $\theta$ motion and, therefore do not have any crossover orbits. The results of this paragraph together with the numbers of positive and negative zeros for each region are summarized in Tab.~\ref{tab:lambda0}.

\begin{table}[t]
\begin{center}
\begin{tabular}{|c|c|c||c|p{2.5cm}|}
\hline
region & $-$ & + & range of $\r$ & types of orbits \\ 
\hline\hline
Id & 0 & 0 &
\begin{pspicture}(-2,-0.2)(2,0.2)
\psline[linewidth=0.5pt]{->}(-2,0)(2,0)
\psline[linewidth=0.5pt](0,-0.2)(0,0.2)
\psline[linewidth=1.2pt](-2,0)(2,0)
\end{pspicture}
& transit \\
\hline
IIb & 1 & 1 & 
\begin{pspicture}(-2,-0.2)(2,0.2)
\psline[linewidth=0.5pt]{->}(-2,0)(2,0)
\psline[linewidth=0.5pt](0,-0.2)(0,0.2)
\psline[linewidth=1.2pt]{-*}(-2,0)(-1,0)
\psline[linewidth=1.2pt]{*-}(0.5,0)(2,0)
\end{pspicture}
& 2x flyby\\
\hline
IId & 2 & 0 &   
\begin{pspicture}(-2,-0.2)(2,0.2)
\psline[linewidth=0.5pt]{->}(-2,0)(2,0)
\psline[linewidth=0.5pt](0,-0.2)(0,0.2)
\psline[linewidth=1.2pt]{-*}(-2,0)(-1.2,0)
\psline[linewidth=1.2pt]{*-}(-0.5,0)(2,0)
\end{pspicture} 
& flyby, \\
& & & & crossover flyby\\
\hline
IIIb & 0 & 4 & 
\begin{pspicture}(-2,-0.2)(2,0.2)
\psline[linewidth=0.5pt]{->}(-2,0)(2,0)  
\psline[linewidth=0.5pt](0,-0.2)(0,0.2)  
\psline[linewidth=1.2pt]{*-*}(0.3,0)(0.7,0)
\psline[linewidth=1.2pt]{*-*}(1.2,0)(1.6,0)
\end{pspicture}
& 2x bound\\
\hline
IVb & 1 & 3 & 
\begin{pspicture}(-2,-0.2)(2,0.2)
\psline[linewidth=0.5pt]{->}(-2,0)(2,0)
\psline[linewidth=0.5pt](0,-0.2)(0,0.2)
\psline[linewidth=1.2pt]{-*}(-2,0)(-0.8,0)
\psline[linewidth=1.2pt]{*-*}(0.5,0)(1,0)
\psline[linewidth=1.2pt]{*-}(1.5,0)(2,0)
\end{pspicture}
& 2x flyby, bound\\
\hline
Vb & 0 & 2 & 
\begin{pspicture}(-2,-0.2)(2,0.2)
\psline[linewidth=0.5pt]{->}(-2,0)(2,0)
\psline[linewidth=0.5pt](0,-0.2)(0,0.2)
\psline[linewidth=1.2pt]{*-*}(0.5,0)(1,0)
\end{pspicture}
& bound \\
\hline
\end{tabular}
\caption{Orbit types for $\Lambda=0$. The $+$ and $-$ columns give the number of positive and negative real zeros of the polynomial $R$. Here the thick lines represent the range of orbits. Turning points are shown by thick dots. The small vertical line denotes $\r=0$.}
\label{tab:lambda0}
\end{center}
\end{table}

\begin{figure*}
\begin{minipage}{0.3\textwidth}
\centering
\includegraphics[width=0.9\textwidth]{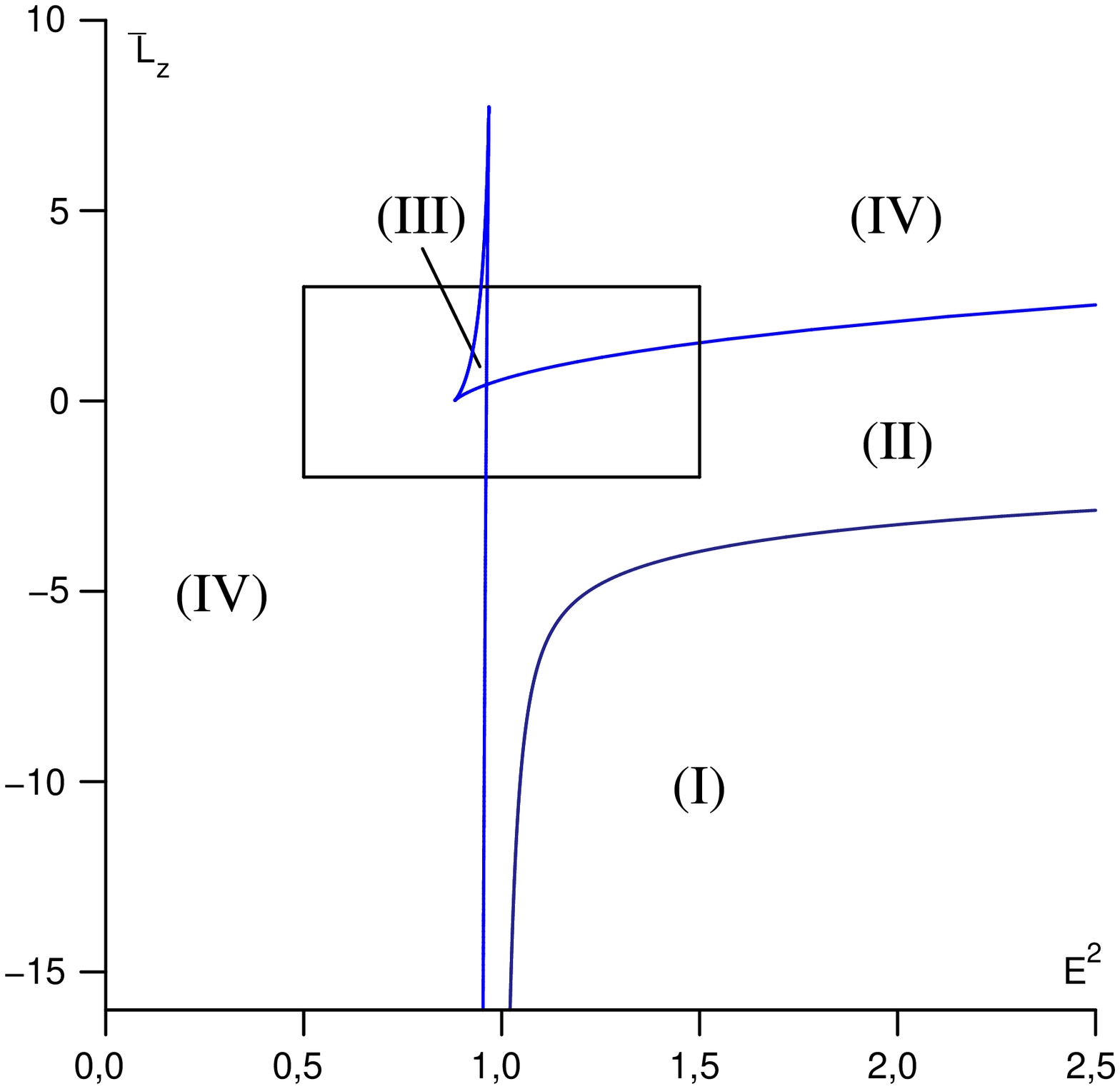}\\
\includegraphics[width=0.85\textwidth]{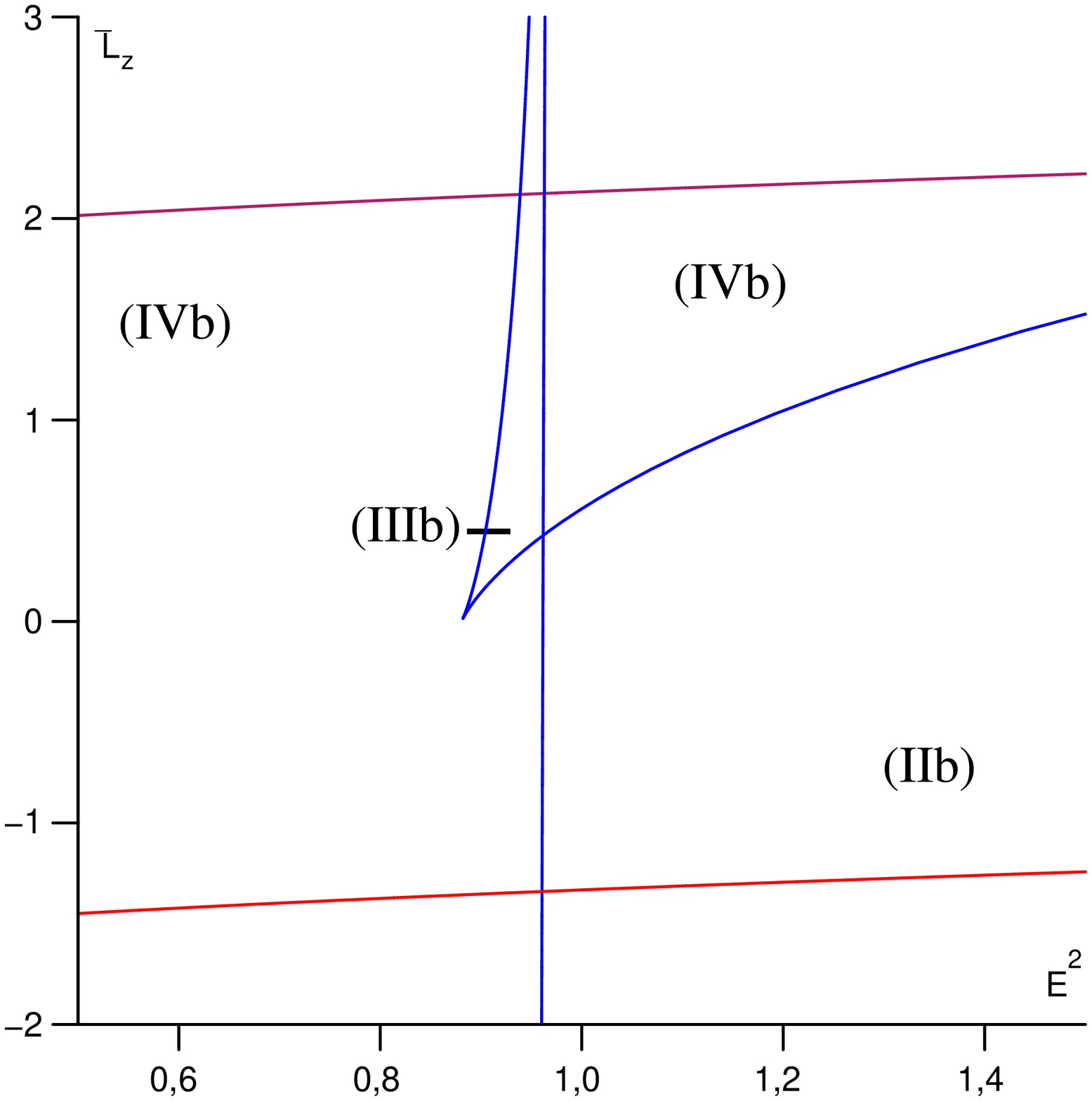}
\end{minipage}
\begin{minipage}{0.3\textwidth}
\centering
\includegraphics[width=0.9\textwidth]{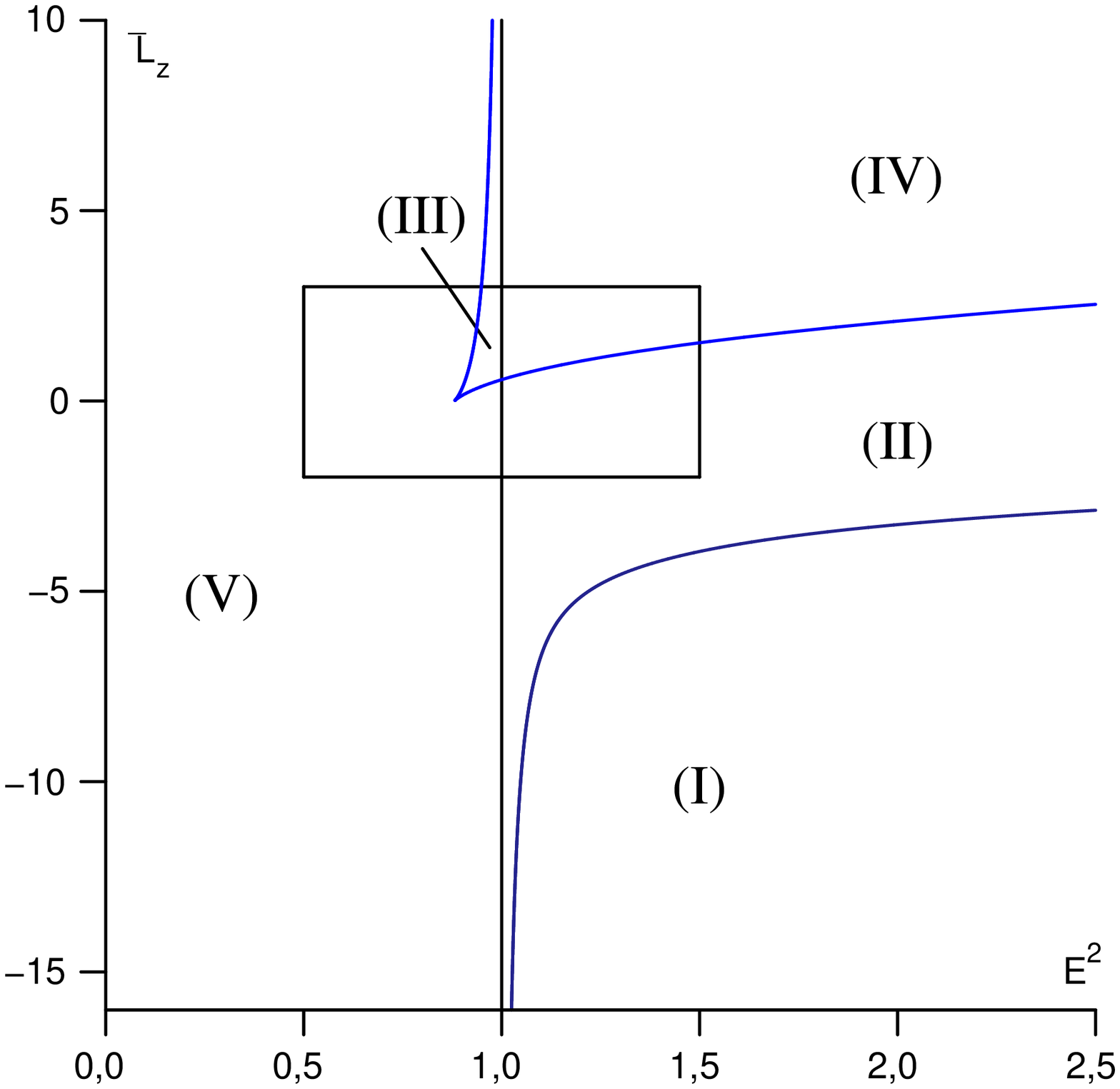}\\
\includegraphics[width=0.85\textwidth]{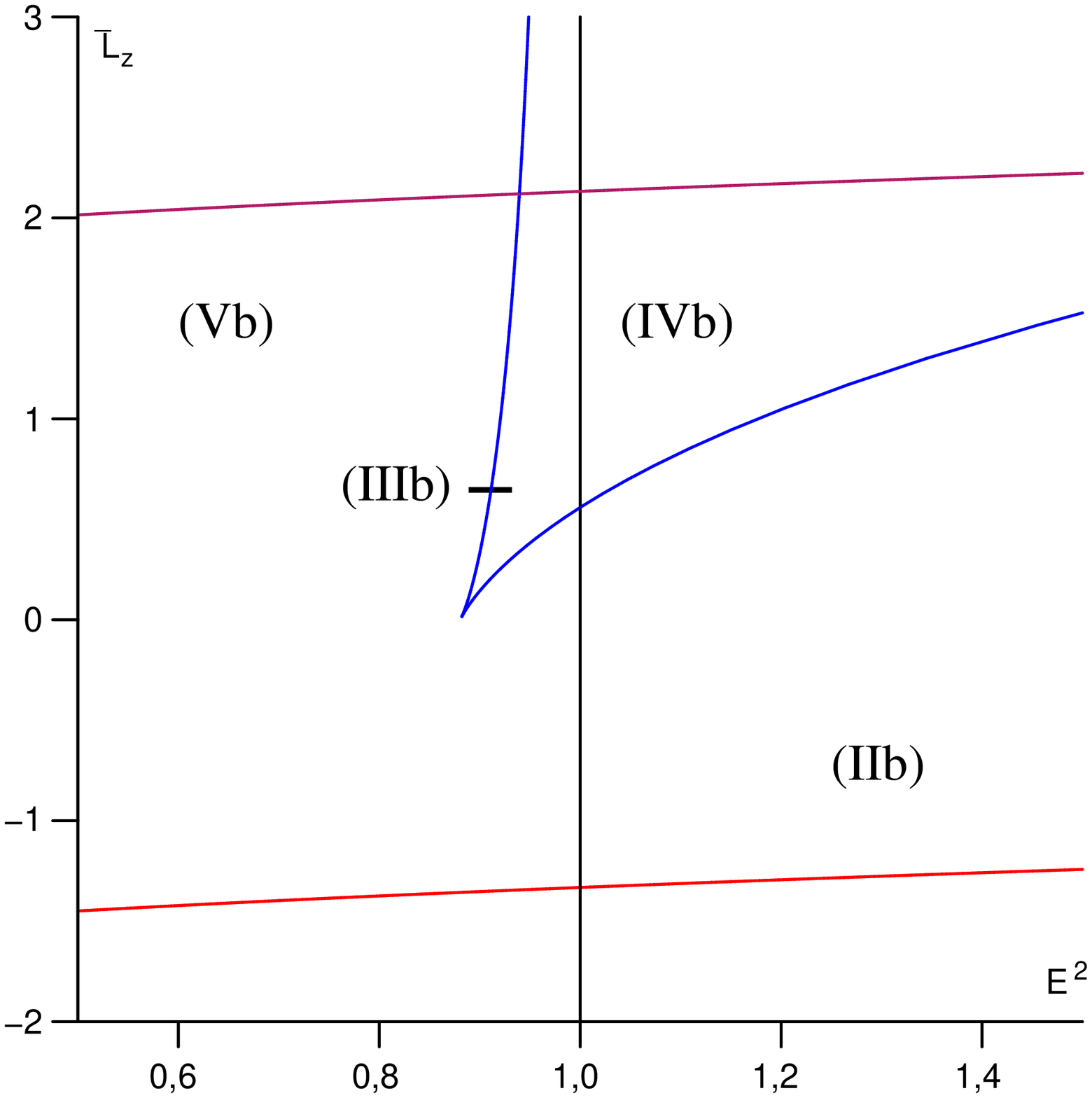}
\end{minipage}
\begin{minipage}{0.3\textwidth}
\centering
\includegraphics[width=0.9\textwidth]{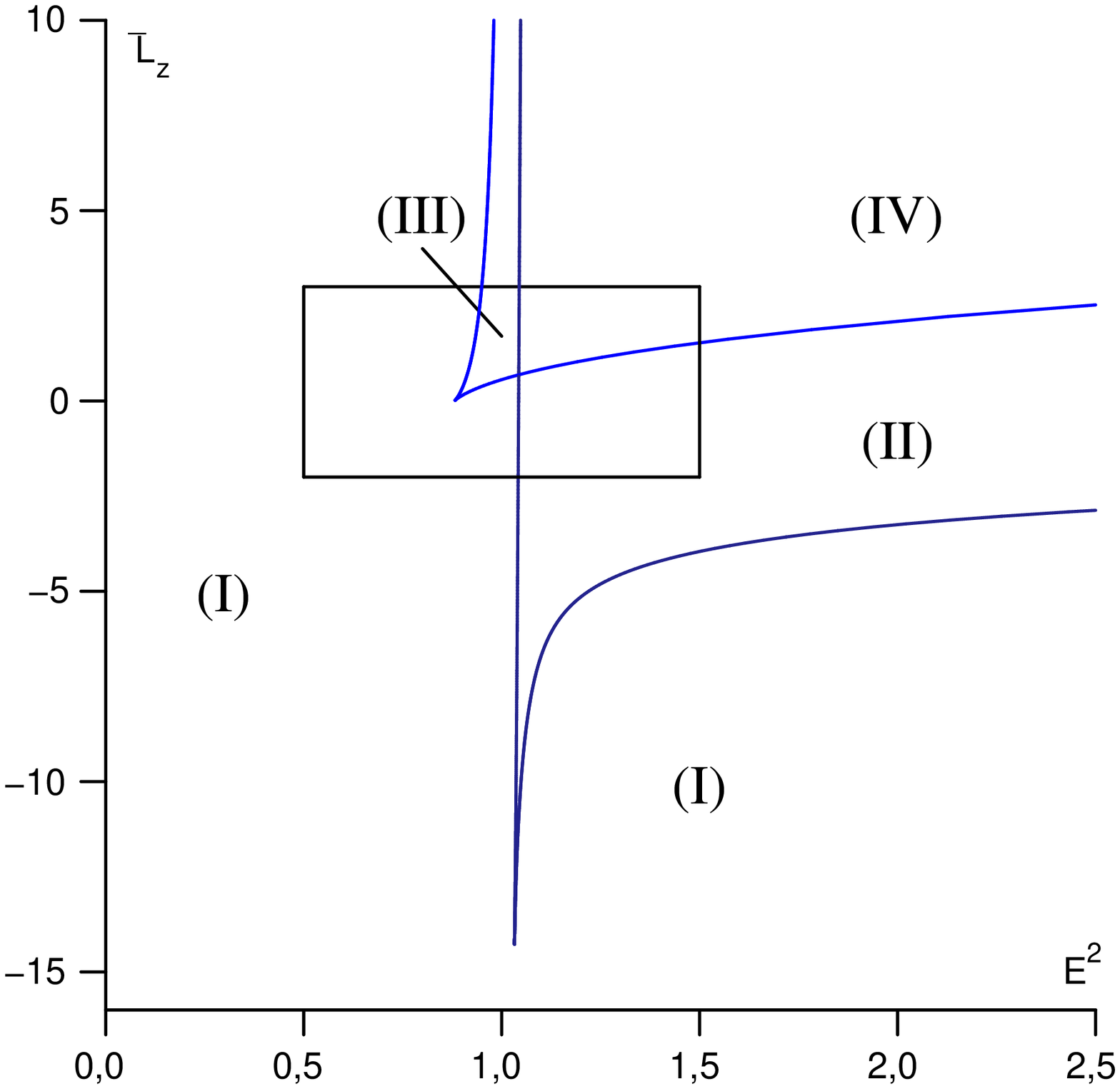}\\
\includegraphics[width=0.85\textwidth]{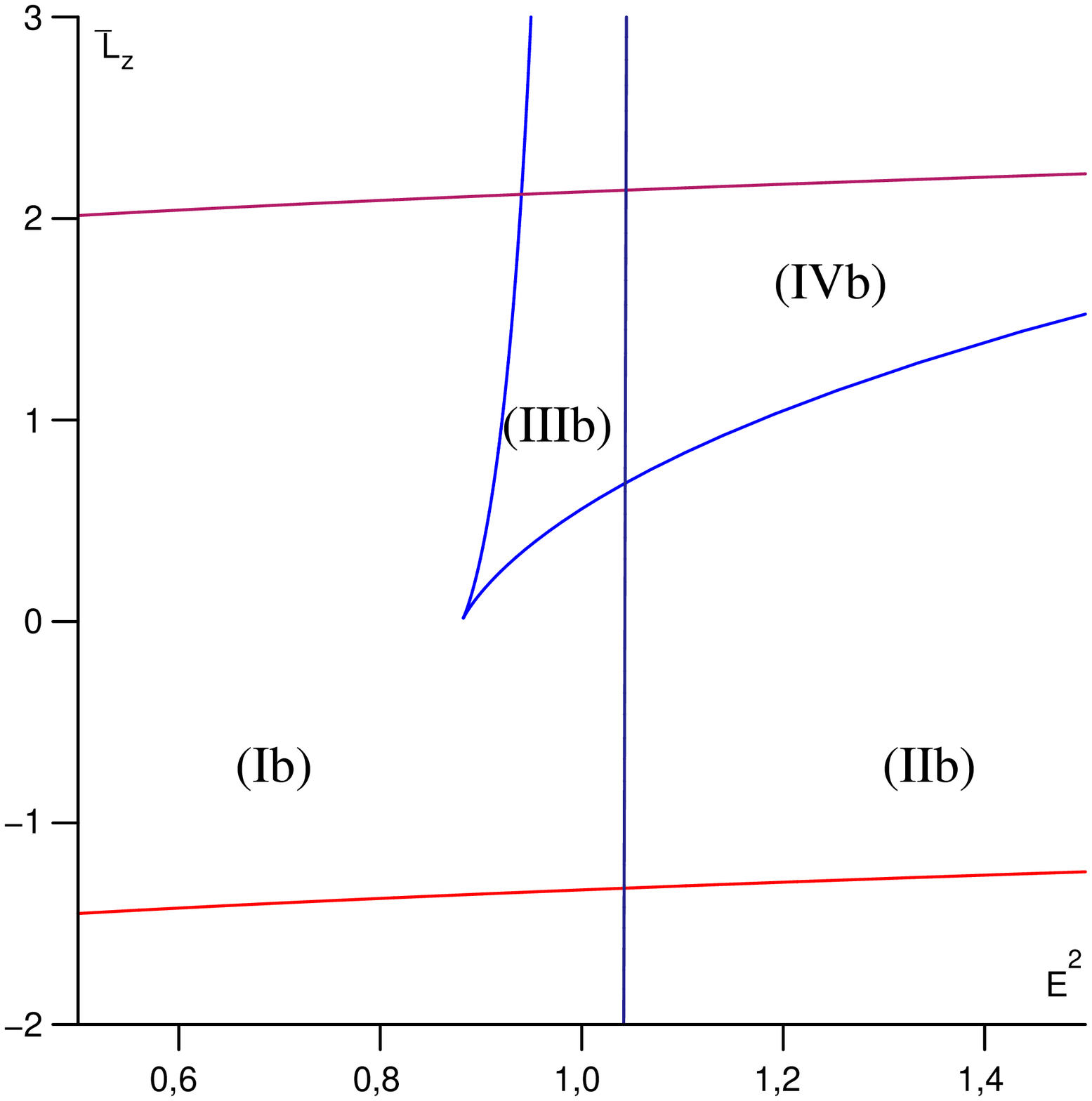}
\end{minipage}
\caption{Regions of different types of $\r$ motion. Here $\al=0.4$ and $\bar K=3$. Left column $\la=10^{-5}$, middle $\la=0$ and right column $\la=-10^{-5}$. Note that in the upper row for $\la>0$ the vertical line has a maximum and for $\la<0$ a minimum. For $\la=0$ the vertical line is not related to a change of the number of real zeros but a switch from $\lim_{\r \to \infty} R(\r) = \infty$ to $\lim_{\r \to \infty} R(\r) = - \infty$. Boundaries of the regions correspond to multiple zeros in $R(\r)$ and, therefore, to the occurrence of a stable or unstable spherical orbit. A detailed view including the results from the $\theta$ motion can be seen in the lower row (note the rescaled axes).}
\label{Fig:r_EL}
\end{figure*}

\begin{figure}
\includegraphics[width=0.25\textwidth]{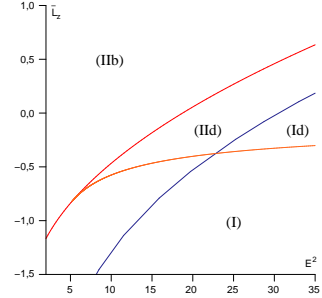}
\caption{Highlighted view on region (d) of $\theta$ motion combined with $r$ motion for $\al=0.4$, $\la=10^{-5}$, and $\bar K=3$. Regions (Id) and (IId) correspond to $Q<0$ and, thus, to crossover orbits whereas region (IIb) corresponds to $Q>0$. Note the differently scaled axes compared to Fig.~\ref{Fig:r_EL}.}
\label{Fig:regiond}
\end{figure}

\begin{figure*}
\includegraphics[width=0.23\textwidth]{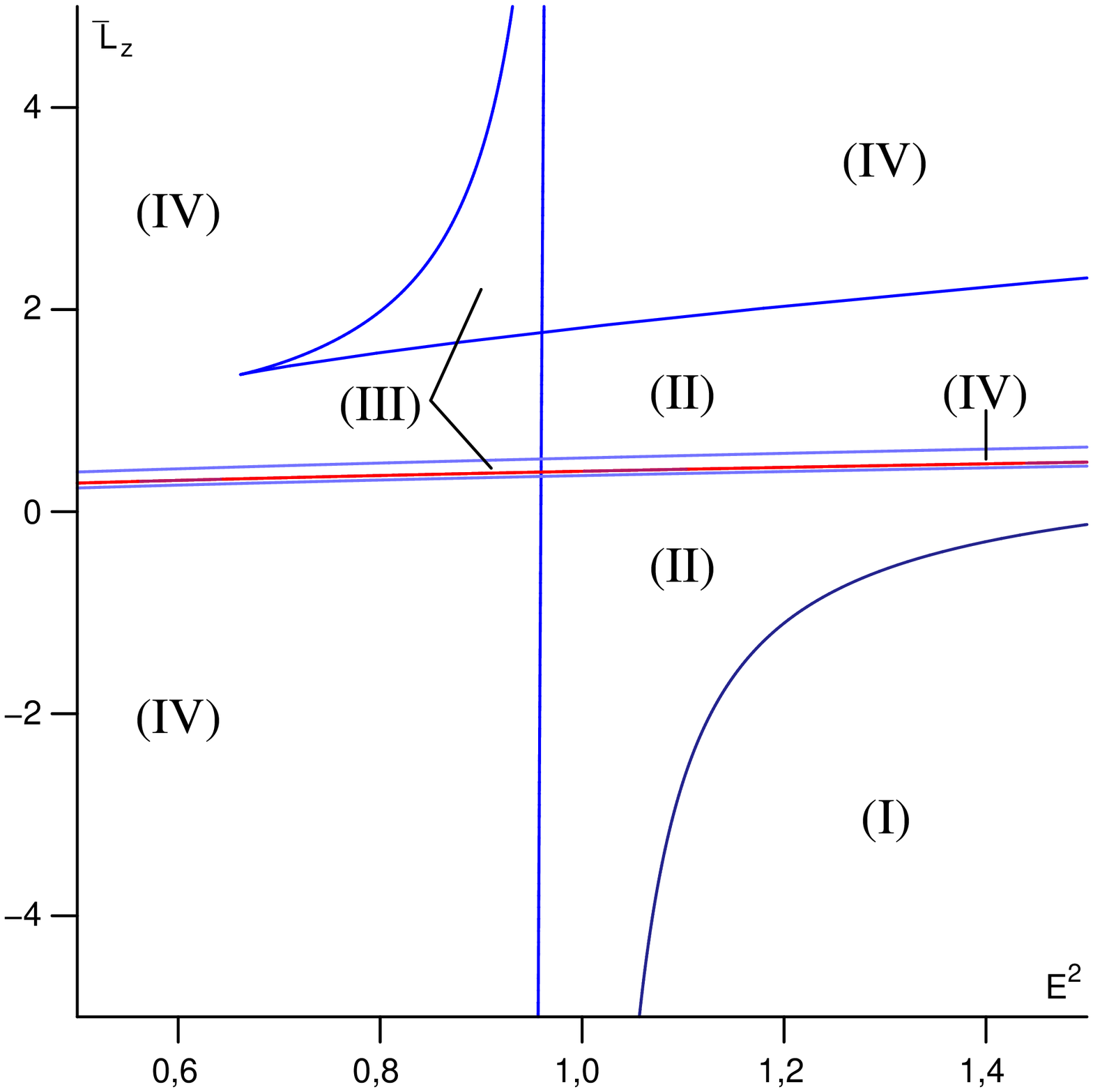}
\includegraphics[width=0.23\textwidth]{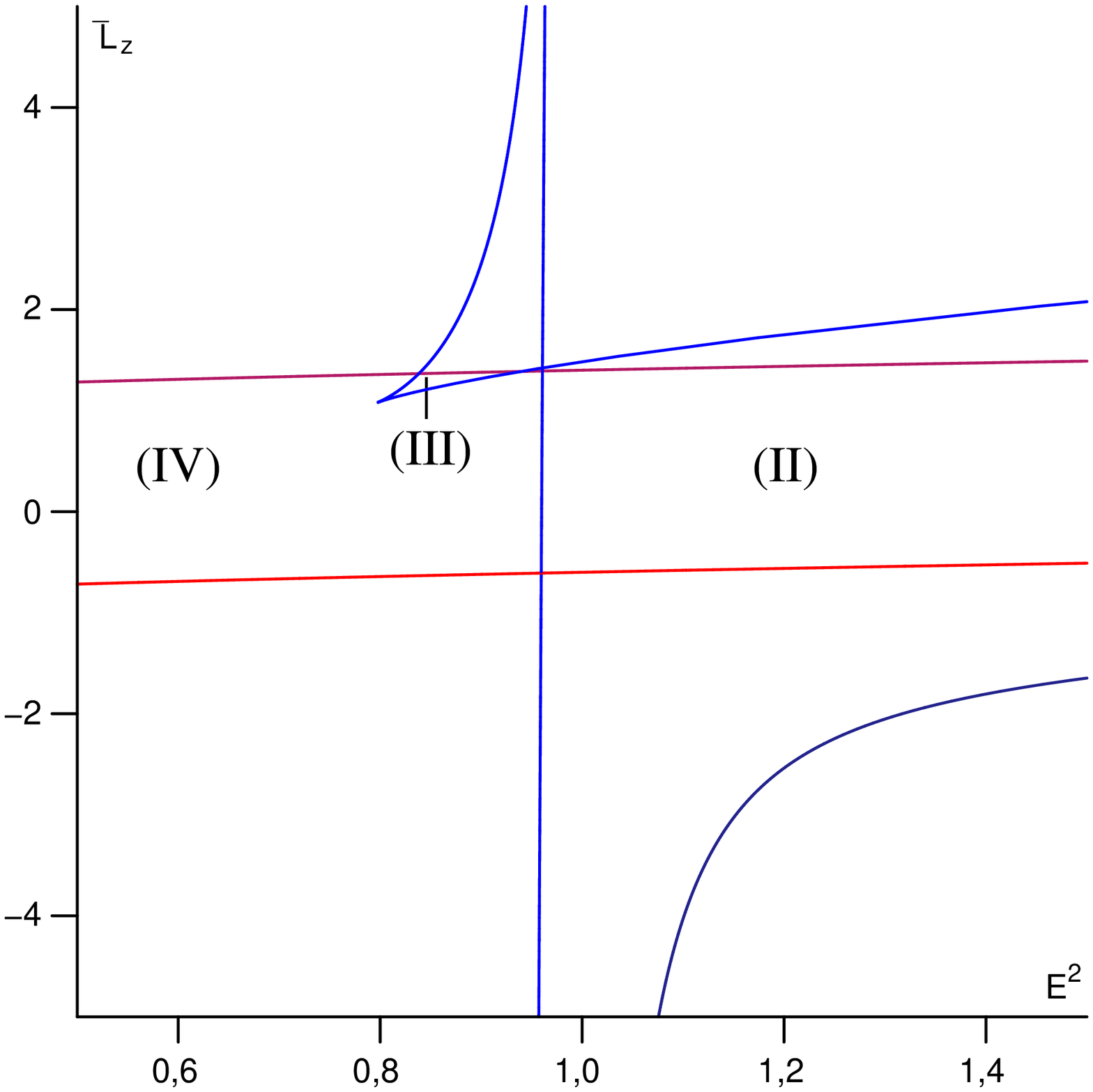}
\includegraphics[width=0.23\textwidth]{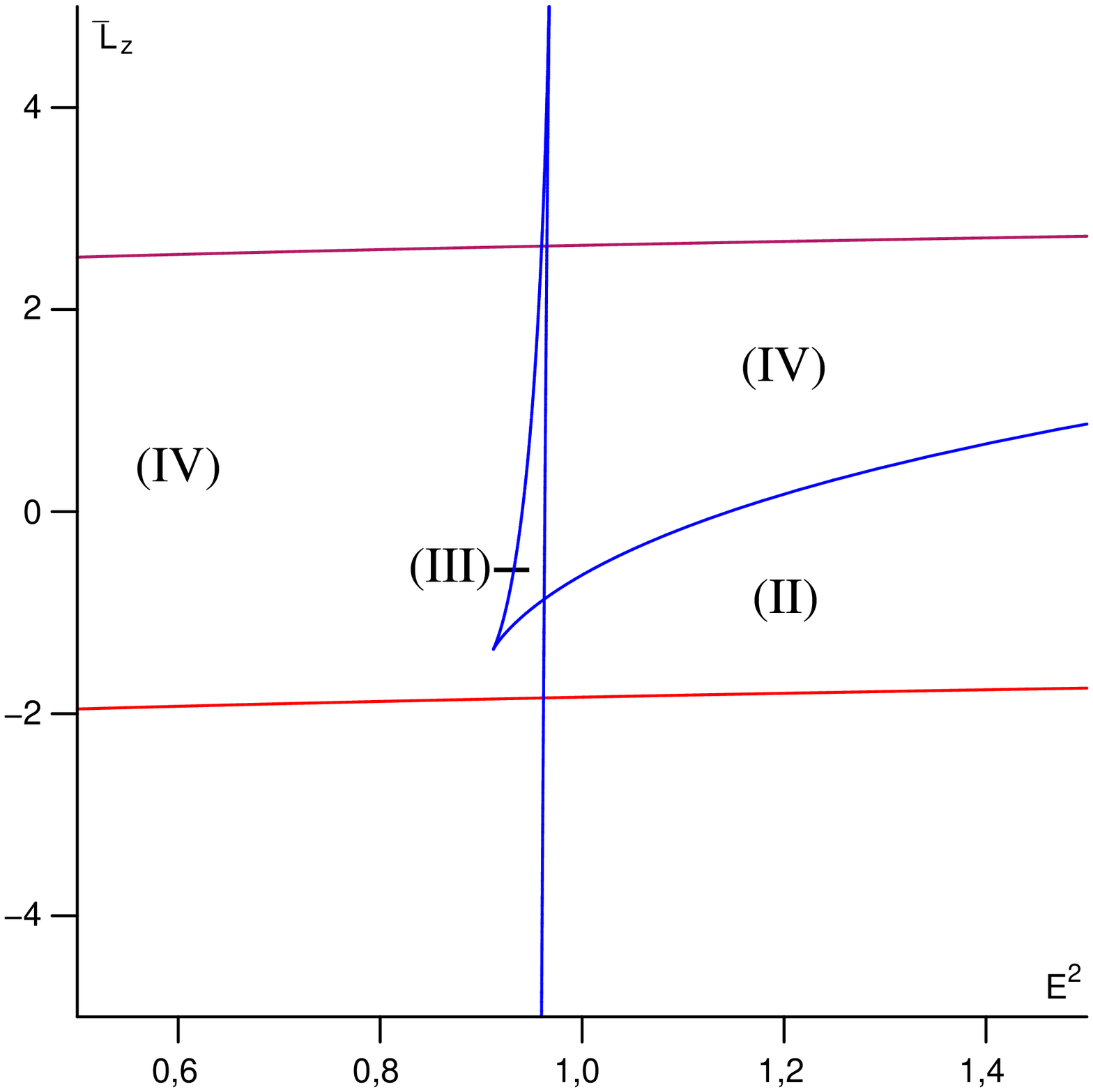}
\includegraphics[width=0.23\textwidth]{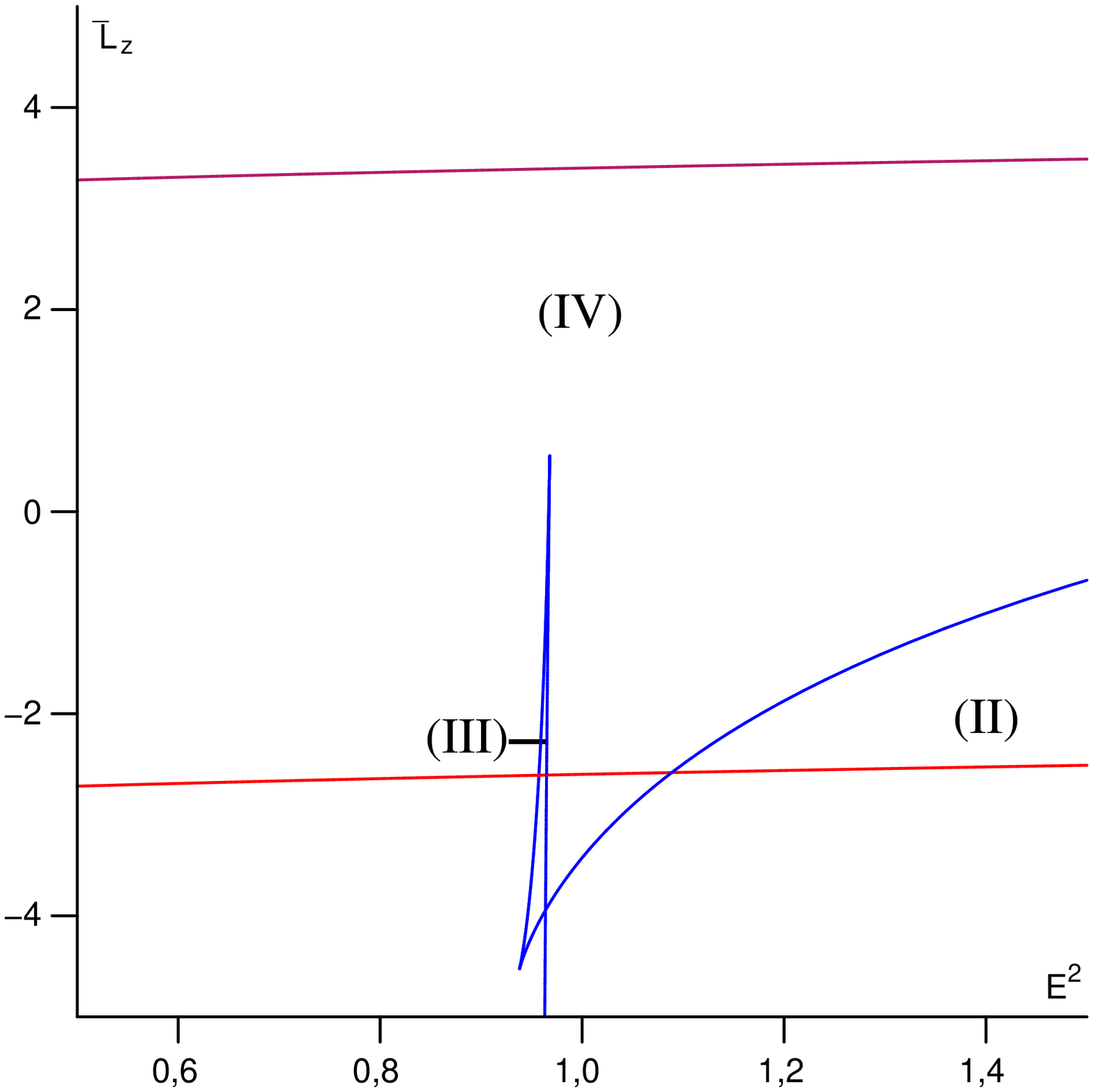}
\caption{Deformation of regions of $r$ motion for $\al=0.4$, $\la=10^{-5}$, and different $\bar K$. From left to right $\bar K=0$, $\bar K=1$, $\bar K=5$, and $\bar K=9$. For $\bar K=0$ region (d) of $\theta$ motion vanishes and region (b) reduces to the line $\bar L_z = \al E$, which, therefore, is the only allowed parameter region and corresponds to $Q=0$ and $\theta \equiv \frac{\pi}{2}$, see Thm.~\ref{Thm2}. Around this line $\bar L_z=\al E$ additional parts of regions (III) and (IV) appear not present for $\bar K>0$.}
\label{Fig:rtheta_EL}
\end{figure*}

\paragraph{Case $\Lambda>0$.} 
Let us analyse now which regions change compared to the case $\Lambda=0$. At first, we recognize that region (V) for $\Lambda=0$ merged with region (IV) and that region (III) becomes smaller for $\Lambda>0$ due to the shift of the separating $E^2=1$ line towards the left. A comparison of the possible orbit types for $\Lambda>0$ with the one for $\Lambda=0$ shows that in regions (I) and (II) there are no differences. However, these regions are slightly deformed (for small $\Lambda$) and a pair of parameters $(E^2,\bar L_z)$ located in region (I) or (II) for $\Lambda>0$ may be located in a different region for $\Lambda=0$. We consider regions (III) and (IV). Here again we assume $r_i<r_{i+1}$.
\begin{itemize}
\item Region (III): all six zeros $r_i$ of $R$ are real and $R(\r) \geq 0$ for $\r \leq r_1, r_6 \leq \r$ and $r_{2k} \leq \r \leq r_{2k+1}$ for $k=1,2$. Possible orbit types: two flyby orbits, one to each of $\pm \infty$, and two different bound orbits.
\item Region (IV): $R$ has four real zeros and $R(\r) \geq 0$ for $\r \leq r_1, r_2 \leq \r \leq r_3, r_4 \leq \r$. Possible orbit types: two flyby orbits, one to each of $\pm \infty$ and a bound orbit.
\end{itemize}
Analogously to $\Lambda=0$, regions (III) and (IV) only contain region (b) of $\theta$ motion implying that there are no crossover orbits. Region (I) can only intersect region (d) because only transit orbits are possible. The remaining region (II) is the only one which intersects regions (b) and (d). 

We conclude that for $E^2>1$ the types of orbits are not noticeably changed, whereas for $E^2 \leq 1$ there are significant changes. In the former region (V) (for $\Lambda=0$), which is now in region (IV), and in region (III) we have two additional flyby orbits which are not present for $\Lambda=0$. In a small vertical stripe left of $E^2=1$ there are even orbits which are bound for $\Lambda=0$ but reaching infinity for $\Lambda>0$. In particular, it is independent of the value of $E$ if a geodesic may reach infinity as expected from the repulsive cosmological force related to $\Lambda>0$. 

Note that for huge $\Lambda$ the separation in regions (I) to (IV) is no longer possible because the repulsive cosmological force becomes so strong that bound orbits are no longer possible. In this case we have only two regions, one with two real zeros corresponding to two flyby orbits and one with only complex zeros corresponding to a transit orbit. 

All orbit types for small $\Lambda>0$ are summarized in Tab.~\ref{tab:lambda>0}.

\begin{table}[t]
\begin{center}
\begin{tabular}{|c|c|c||c|p{2.55cm}|}
\hline
region & $-$ & + & range of $\r$ & types of orbits \\ 
\hline\hline
Id & 0 & 0 &
\begin{pspicture}(-2,-0.2)(2,0.2)
\psline[linewidth=0.5pt]{->}(-2,0)(2,0)
\psline[linewidth=0.5pt](0,-0.2)(0,0.2)
\psline[linewidth=1.2pt](-2,0)(2,0)
\end{pspicture}
& transit \\
\hline
IIb & 1 & 1 & 
\begin{pspicture}(-2,-0.2)(2,0.2)
\psline[linewidth=0.5pt]{->}(-2,0)(2,0)
\psline[linewidth=0.5pt](0,-0.2)(0,0.2)
\psline[linewidth=1.2pt]{-*}(-2,0)(-1,0)
\psline[linewidth=1.2pt]{*-}(0.5,0)(2,0)
\end{pspicture}
& 2x flyby\\
\hline
IId & 2 & 0 &   
\begin{pspicture}(-2,-0.2)(2,0.2)
\psline[linewidth=0.5pt]{->}(-2,0)(2,0)
\psline[linewidth=0.5pt](0,-0.2)(0,0.2)
\psline[linewidth=1.2pt]{-*}(-2,0)(-1.2,0)
\psline[linewidth=1.2pt]{*-}(-0.5,0)(2,0)
\end{pspicture} 
& flyby, \\
& & & & crossover flyby\\
\hline
IIIb & 1 & 5 & 
\begin{pspicture}(-2,-0.2)(2,0.2)
\psline[linewidth=0.5pt]{->}(-2,0)(2,0)  
\psline[linewidth=0.5pt](0,-0.2)(0,0.2)  
\psline[linewidth=1.2pt]{-*}(-2,0)(-1.2,0)
\psline[linewidth=1.2pt]{*-*}(0.2,0)(0.6,0)
\psline[linewidth=1.2pt]{*-*}(0.9,0)(1.3,0)
\psline[linewidth=1.2pt]{*-}(1.6,0)(2,0)
\end{pspicture}
&2x flyby, 2x bound\\
\hline
IVb & 1 & 3 & 
\begin{pspicture}(-2,-0.2)(2,0.2)
\psline[linewidth=0.5pt]{->}(-2,0)(2,0)
\psline[linewidth=0.5pt](0,-0.2)(0,0.2)
\psline[linewidth=1.2pt]{-*}(-2,0)(-0.8,0)
\psline[linewidth=1.2pt]{*-*}(0.5,0)(1,0)
\psline[linewidth=1.2pt]{*-}(1.5,0)(2,0)
\end{pspicture}
& 2x flyby, bound\\
\hline
\end{tabular}
\caption{Orbit types for small $\Lambda>0$. For the description of the $+$, $-$ and range of $\r$ columns see Tab.~\ref{tab:lambda0}.}
\label{tab:lambda>0}
\end{center}
\end{table}

\paragraph{Case $\Lambda<0$.}
Here region (V) from $\Lambda=0$ merges with region (I) and region (III) becomes larger for $\Lambda<0$ due to the shift of the $E^2=1$ line to the right. Compared to the situation for $\Lambda=0$ the possible orbit types in region (III) do not change, but a set of parameters located there may be located in a different region for $\Lambda=0$. Let us examine the remaining regions. 
\begin{itemize}
\item Region (I): $R$ has two real zeros $r_1<r_2$ and $R(\r) \geq 0$ for $r_1 \leq \r \leq r_2$. Possible orbit types: bound orbit.
\item Region (II) (and (III)): $R$ has four real zeros with $R(\r) \geq 0$ for $r_{2k-1} \leq \r \leq r_{2k}$, $k=1,2$. Possible orbit types: two different bound orbits.
\item Region (IV): all six zeros of $R$ are real and $R(\r) \geq 0$ for $r_{2k-1} \leq \r \leq r_{2k}$, $k=1,2,3$. Possible orbit types: three different bound orbits.
\end{itemize}
Concerning crossover orbits regions (III) and (IV) again only contain region (b) of $\theta$ motion. Also region (II) intersects both (b) and (d) whereas region (I) can only contain region (d) of $\theta$ motion.

Summarizing, the types of orbits significantly change if $E^2>1$. The transit orbit in region (I) for $\Lambda=0$ is transformed to a bound orbit for $\Lambda<0$ as well as the flyby orbits in regions (II) and (IV). Although region (V) for $\Lambda=0$ merges with region (I) for $\Lambda<0$, the types of orbits do not change there. In general, because of $R \to - \infty$ if $\r \to \pm \infty$ we can not have orbits reaching $\r= \pm \infty$ at all as expected due to the attractive cosmological force related to $\Lambda<0$. 

All orbit types for $\Lambda<0$ are summarized in Tab.~\ref{tab:lambda<0}

\begin{table}[t]
\begin{center}
\begin{tabular}{|c|c|c||c|p{2.5cm}|}
\hline
region & $-$ & + & range of $\r$ & types of orbits \\ 
\hline\hline
Id & 1 & 1 &
\begin{pspicture}(-2,-0.2)(2,0.2)
\psline[linewidth=0.5pt]{->}(-2,0)(2,0)
\psline[linewidth=0.5pt](0,-0.2)(0,0.2)
\psline[linewidth=1.2pt]{*-*}(-1,0)(1,0)
\end{pspicture}
& crossover bound \\
\hline
IIb & 2 & 2 & 
\begin{pspicture}(-2,-0.2)(2,0.2)
\psline[linewidth=0.5pt]{->}(-2,0)(2,0)
\psline[linewidth=0.5pt](0,-0.2)(0,0.2)
\psline[linewidth=1.2pt]{*-*}(-1.2,0)(-0.5,0)
\psline[linewidth=1.2pt]{*-*}(0.5,0)(1.2,0)
\end{pspicture}
& 2x bound\\
\hline
IId & 3 & 1 &   
\begin{pspicture}(-2,-0.2)(2,0.2)
\psline[linewidth=0.5pt]{->}(-2,0)(2,0)
\psline[linewidth=0.5pt](0,-0.2)(0,0.2)
\psline[linewidth=1.2pt]{*-*}(-1.5,0)(-1,0)
\psline[linewidth=1.2pt]{*-*}(-0.5,0)(0.5,0)
\end{pspicture} 
& bound\\
& & & & crossover bound \\
\hline
IIIb & 0 & 4 & 
\begin{pspicture}(-2,-0.2)(2,0.2)
\psline[linewidth=0.5pt]{->}(-2,0)(2,0)  
\psline[linewidth=0.5pt](0,-0.2)(0,0.2)  
\psline[linewidth=1.2pt]{*-*}(0.3,0)(0.8,0)
\psline[linewidth=1.2pt]{*-*}(1.2,0)(1.7,0)
\end{pspicture}
& 2x bound\\
\hline
IVb & 2 & 4 & 
\begin{pspicture}(-2,-0.2)(2,0.2)
\psline[linewidth=0.5pt]{->}(-2,0)(2,0)
\psline[linewidth=0.5pt](0,-0.2)(0,0.2)
\psline[linewidth=1.2pt]{*-*}(-1.5,0)(-0.8,0)
\psline[linewidth=1.2pt]{*-*}(0.3,0)(0.8,0)
\psline[linewidth=1.2pt]{*-*}(1.2,0)(1.7,0)
\end{pspicture}
& 3x bound\\
\hline
\end{tabular}
\caption{Orbit types for $\Lambda<0$. For the description of the $+$, $-$ and range of $\r$ columns see Tab.~\ref{tab:lambda0}.}
\label{tab:lambda<0}
\end{center}
\end{table}

\section{Analytic solutions of the equations of motion} \label{section:exact solutions}
We will now analytically solve the geodesic equation in Kerr-de Sitter space-time \eqref{dot r_sn} - \eqref{dot t_sn}. Each equation will be treated separately.

\subsection{$\theta$ motion} \label{thetamotion}
We begin with the differential equation \eqref{dot theta_sn}
\begin{align}
\left( \frac{d\theta}{d\gamma} \right)^2 & = \Theta = \Dt (\ka - \delta_2 \al^2 \cos^2\theta) - \frac{\chi^2 \mathbbm{T}^2}{\sin^2\theta} \,, \nonumber
\end{align}
which can be simplified by the substitution $\nu = \cos^2\theta$ yielding
\begin{align}
\frac{1}{4} \left( \frac{d\nu}{d\gamma} \right)^2 & = \nu \Theta_\nu \label{Theta_nu} \,,
\end{align}
where $\Theta_\nu$ is defined as in \eqref{def Theta_nu}. This differential equation can be solved easily if $\nu \Theta_\nu$ has a zero with multiplicity $2$ or more. In this case \eqref{Theta_nu} can be rewritten as
\begin{align}
4 (\gamma-\gamma_0) = \int_{\nu_0}^\nu \frac{d\nu'}{(\nu'-\nu_i)^j \sqrt{P_2(\nu')}} \,,
\end{align}
where $\gamma_0$ and $\nu_0$ are initial values, $P_2$ is a polynomial with maximum degree $2$, and $\nu_i$ is a zero of $\nu \Theta_\nu$ with multiplicity $2j$ or $2j+1$, $j=1,2$. The integral on the right hand side can then be solved by elementary functions \cite{GradshteynRyzhik83}. As in this case the explicit expression provides no further insight and some case distinctions would be necessary we skip the solution procedure.

If $\nu \Theta_\nu$ has only simple zeros the differential equation \eqref{Theta_nu} is of elliptic type and first kind and can be solved in terms of the Weierstrass elliptic function $\wp$. In contrast to the $r$ motion considered in the next subsection, this structure does not simplify if we consider only light with $\delta=0$. To obtain a solution we transform $\nu \Theta_\nu$ to the Weierstrass form $(4y^3-g_2y-g_3)$ for some constants $g_2$ and $g_3$: First, we substitute $\nu=\xi^{-1}$ giving
\begin{align}
\frac{1}{4} \left( \frac{d\xi}{d\gamma} \right)^2 & = \Theta_\xi \,,
\end{align}
where
\begin{align}
\Theta_\xi & := \xi^3 \left( \ka - \chi^2 (\al-\D)^2 \right) + \xi^2 ( \al^2(\ka \la-\delta_2) -\ka \nonumber \\
& \quad +2\chi^2\al(\al-\D) ) + \al^2 ( \delta_2(1-\la \al^2)-\chi^2-\la \ka ) \xi \nonumber \\
& \quad +  \delta_2 \al^4 \la \label{Theta_xi} \\
& =: \sum_{i=1}^3 a_{i} \xi^i \,. \nonumber
\end{align}   
Second, we substitute $\xi = \frac{1}{a_{3}} \left( 4y-\frac{a_{2}}{3} \right)$ yielding
\begin{align}
\frac{1}{4} \left( \frac{dy}{d\gamma} \right)^2 & = 4 y^3 - g_2 y - g_3 \,, \label{Theta_y}
\end{align}
where
\begin{align}
g_2 & = \frac{a_{2}^2}{12} - \frac{a_{1} a_{3}}{4} \,, \nonumber \\
g_3 & = \frac{1}{48} a_{1} a_{2} a_{3} - \frac{1}{16} a_{0} a_{3}^2 - \frac{1}{216} a_{2}^3 \,. \label{def_g2g3}
\end{align}
The differential equation \eqref{Theta_y} is elliptic of first kind, which can be solved by
\begin{align}
y(\gamma) & = \wp( 2 \gamma - \gamma_{\theta,\rm in};g_2,g_3) \label{sol_y} \,.
\end{align}
Accordingly, the solution of \eqref{dot theta_sn} is given by
\begin{align}
\theta(\gamma) & =  \arccos \left( \pm \sqrt{ \frac{a_{3}}{ 4 \wp( 2 \gamma - \gamma_{\theta,{\rm in}};g_2,g_3 ) - \frac{a_{2}}{3}}} \right) \,,
\end{align}
where $\gamma_{\theta,\rm in} = 2 \gamma_0 + \int_{y_0}^\infty \frac{dy'}{\sqrt{4y'^3-g_2y'-g_3}}$ with $y_0 = \frac{a_{3}}{4\cos^2(\theta_0)} + \frac{a_{2}}{12}$ depends on the initial values $\gamma_0$ and $\theta_0$ only. The sign of the square root depends on whether $\theta(\gamma)$ should be in $(0,\frac{\pi}{2})$ (positive sign) or in $(\frac{\pi}{2},\pi)$ (negative sign) and reflects the symmetry of the $\theta$ motion with respect to the equatorial plane $\theta=\frac{\pi}{2}$. If the motion is located in region (b) from the previous section this implies that the two solutions have to be glued together along $\theta(\gamma)=\frac{\pi}{2}$ if the whole $\theta$ motion should be considered.

\subsection{$r$ motion} \label{rmotion}
The differential equation that describes the dynamics of $r$ \eqref{dot r_sn}
\begin{align}
\left( \frac{d\r}{d\gamma}\right)^2 & = R = \chi^2 \mathbbm{P}^2 - \Dx (\delta_2 \r^2 + \ka) \nonumber
\end{align}
is more complicated because $R$ is a polynomial of a degree up to $6$. If $R$ has a zero of multiplicity $4$ or more, or if $R$ has two zeros of multiplicity $2$ or more, the differential equation \eqref{dot r_sn} can be written as
\begin{align}
\gamma-\gamma_0 = \int_{\r_0}^{\r} \frac{d\r'}{\prod_{i=1}^k (\r'-\r_i)^{j_i} \sqrt{P_2(\r')}} \,,
\end{align}
where $\gamma_0$ and $\r_0$ are initial values, $P_2$ is a polynomial with maximum degree 2, $\r_i$ are zeros of $R$ with multiplicity $2j_i$ or $2j_i+1$ where $j_i=1,2$, and $k=2$ if there are two zeros of multiplicity $2$ or more and $k=1$ else. The integral on the right hand side can then be solved by elementary functions \cite{GradshteynRyzhik83}. As the explicit expression provides no further insight and some case distinctions would be necessary we skip the solution procedure.

If we consider null geodesics, i.e. $\delta=0$, $R$ is in general of degree 4 and the differential equation \eqref{dot r_sn} is of elliptic type and first kind. Then it can be handled using the method presented in the foregoing subsection: With the substitutions $\r=\xi^{-1}+\r_R$, where $\r_R$ is a zero of $R$, and $\xi = \frac{1}{b_3} \left(4y-\frac{b_2}{3}\right)$, where $b_i = \frac{1}{(4-i)!} \frac{d^{(4-i)} R}{d\r^{(4-i)}} (\r_R)$, \comment{and $R_\xi = \left(\frac{d\xi}{d\r}\right)^2 R_{\r=\xi^{-1}+\r_R}$,} we arrive at a form \eqref{Theta_y}. This can then again be solved in terms of Weierstrass elliptic functions. The result is
\begin{align}\label{sol_r light}
\r(\gamma) = \frac{b_3}{4\wp(\gamma-\gamma_{\r,\rm{in}};g_{2},g_{3}) - \frac{b_2}{3}} + \r_R \,,
\end{align}
where $\gamma_{\r,\rm{in}} = \gamma_0 + \int_{y_0}^\infty \frac{dy'}{\sqrt{4y'^3-g_{2,r}y'-g_{3,r}}}$ with $y_0 = \frac{b_3}{4(\r_0-r_{\rm R})} + \frac{b_2}{12}$ depends only on the initial values $\gamma_0$ and $\r_0$ and $g_2, g_3$ are defined as in \eqref{def_g2g3} with $a_{i} = b_i$.

The differential equation \eqref{dot r_sn} is also of elliptic type but of third kind if $R$ has a double or triple zero $\r_1$. In this case \eqref{dot r_sn} reads
\begin{align}\label{requation_elliptic}
\gamma - \gamma_0 =\int_{\r_0}^{\r} \frac{d\r'}{(r-r_1) \sqrt{P_4(\r)}} \,,
\end{align}
where $P_4$ is a polynomial of degree 4. This equation can be solved for $\r(\gamma)$ analogous to the method which will be presented in subsection \ref{phimotion}.

If we consider particles, i.e. $\delta = 1$, and assume that $R$ has only simple zeros the differential equation \eqref{dot r_sn} is of hyperelliptic type. It can be solved in terms of derivatives of the Kleinian $\sigma$ function with the method developed in \cite{HackmannLaemmerzahl08}. For this, we have to cast \eqref{dot r_sn} in the standard form by a substitution $\r= \pm \frac{1}{u} + \r_R$ with a zero $\r_R$ of $R$. This yields
\begin{align}
\left( u \frac{du}{d\gamma} \right)^2 & = c_{5} R_u  \label{R_u} \,,
\end{align}
where
\begin{align}
R_u & = \sum_{i=0}^5 \frac{c_{i}}{c_{5}} u^i \,, \quad c_{i} = \frac{(\pm 1)^i}{(6-i)!} \frac{d^{(6-i)} R}{du^{(6-i)}} (\r_R)\,. \label{basic_R}
\end{align}
The sign in the substitution has be chosen such that the constant $c_{5}$ is positive and, therefore, depends on the choice of $\r_R$ and the sign of $\la$. The differential equation \eqref{R_u} is of first kind and can be solved by
\begin{align}
u(\gamma) & = - \frac{\sigma_1}{\sigma_2} \left( \begin{matrix} f( \sqrt{c_{5}} \gamma - \gamma_{\r,\rm in}) \\ \sqrt{c_{5}} \gamma - \gamma_{\r,\rm in} \end{matrix} \right) \,,
\end{align}
where $\gamma_{\r,\rm in} = \sqrt{c_{5}} \gamma_0 + \int_{u_0}^\infty \frac{u' du'}{\sqrt{\tilde R_{u'}}}$ and $u_0 = \pm \left(\r_0-\r_R \right)^{-1}$ depends only on the initial values $\gamma_0$ and $\r_0$. Here $f$ is the function that describes the $\theta$-divisor, i.e. $\sigma\left( (f(x),x)^t \right) = 0$, cp.~\cite{HackmannLaemmerzahl08}. The radial distance $\r$ is then given by
\begin{align}
\r(\gamma) & = \mp \frac{\sigma_2}{\sigma_1}\left( \begin{matrix} f( \sqrt{c_{5}} \gamma - \gamma_{\r,\rm in}) \\ \sqrt{c_{5}} \gamma - \gamma_{\r, \rm in} \end{matrix} \right)  + \r_R\,,
\end{align}
where the sign depends on the sign chosen in the substitution $\r = \pm \frac{1}{u} + \r_R$, i.e. is such that $c_5$ in \eqref{basic_R} is positive.

\subsection{$\phi$ motion} \label{phimotion}
We treat now the most complicated equation of motion in Kerr-de Sitter space-time, namely the equation for the azimuthal angle \eqref{dot phi_sn}
\begin{equation*}
\frac{d\phi}{d\gamma} = \chi^2 \left[ \frac{\al}{\Dx} \mathbbm{P} - \frac{\mathbbm{T}}{\Dt \sin^2\theta} \right] \,.
\end{equation*}
This equation can be splitted in a part dependent only on $\r$ and in a part only dependent on $\theta$. Integration yields
\begin{align}
\phi - \phi_0 & = \chi^2 \left[  \int_{\gamma_0}^\gamma \frac{\al \mathbbm{P}}{\Delta_{\r(\gamma)}} d\gamma - \int_{\gamma_0}^\gamma  \frac{\mathbbm{T} d\gamma}{\Delta_{\theta(\gamma)} \sin^2\theta(\gamma)} \right]  \nonumber\\
& = \chi^2 \left[ \int_{\r_0}^{\r} \frac{\al \mathbbm{P} d\r'}{\Delta_{\r'} \sqrt{R}} - \int_{\theta_0}^\theta  \frac{\mathbbm{T} d\theta'}{\Delta_{\theta'} \sin^2\theta' \sqrt{\Theta}} \right] \,, \label{phi_int}
\end{align}
where we substituted $\r=\r(\gamma)$, i.e. $\frac{d\r}{d\gamma} = \sqrt{R}$, in the first and $\theta=\theta(\gamma)$, i.e. $\frac{d\theta}{d\gamma} = \sqrt{\Theta}$, in the second integral.

We will solve now the two integrals in \eqref{phi_int} separately.

\subsubsection{The $\theta$ dependent integral} \label{phithetamotion}
Let us consider the integral
\begin{equation}
I_\theta := \int_{\theta_0}^\theta  \frac{\left( \sin^2\theta' \al - \D \right) d\theta'}{\Delta_{\theta'} \sin^2\theta' \sqrt{\Theta}} \,,
\end{equation}
which can be transformed to the simpler form 
\begin{align}
I_\theta & = \mp \frac{1}{2} \int_{\nu_0}^\nu \frac{\al-\D-\al \nu'}{\Delta_{\nu'} (1-\nu') \sqrt{\nu' \Theta_{\nu'}}} d\nu'  \label{I_theta}
\end{align}
by the substitution $\nu=\cos^2\theta$, where $\Theta_\nu$ is defined in \eqref{def Theta_nu} and $\Delta_\nu = 1+\al^2\la\nu$. Here we have to pay special attention to the integration path. If $\theta \in (0,\frac{\pi}{2}]$ we have $\cos \theta = +\sqrt{\nu}$ but for $\theta \in [\frac{\pi}{2}, \pi)$ it is $\cos \theta = -\sqrt{\nu}$. Accordingly, we first have to split the integration path from $\theta_0$ to $\theta$ such that every piece is fully contained in the interval $(0,\frac{\pi}{2}]$ or $[\frac{\pi}{2},\pi)$ and then to choose the appropiate sign of the square root of $\nu$. In the following we assume for simplicity that $\cos \theta = + \sqrt{\nu}$.

Analogous to subsection \ref{thetamotion} the integral $I_\theta$ can be solved by elementary functions if $\nu \Theta_\nu$ has at least a double zero \cite{GradshteynRyzhik83}. If $\nu \Theta_\nu$ has only simple zeros $I_\theta$ is of elliptic type and of third kind. If this is the case, the solution to $I_\theta$ is given by 
\begin{multline} \label{sol I_theta}
I_\theta = \frac{|a_{3}|}{2 a_{3}} \bigg\{ (\al-\D) (v-v_0) \\
- \sum_{i=1}^4 \frac{a_{3}}{4 \chi \wp'(v_i)} \left( \zeta(v_i) (v-v_0) + \log \frac{\sigma(v-v_i)}{\sigma(v_0-v_i)} + 2\pi i k_i \right) \\
\cdot \big( \al^3 \la (\chi-\al \la \D) (\delta_{i1}+\delta_{i2}) + \D (\delta_{i3}+\delta_{i4}) \big) \bigg\}
\end{multline}
where the constants $a_i$ are defined as in subsection \ref{thetamotion}, $\wp(v_1) = \frac{a_2}{12} - \frac{1}{4} \al^2\la a_3= \wp(v_2)$, $\wp(v_3) = \frac{a_{2}}{12} + \frac{a_{3}}{4} = \wp(v_4)$, $v = v(\gamma) = 2\gamma - \gamma_{\theta,\rm in}$ with $\gamma_{\theta,\rm in}$ as in \eqref{sol_y} and $v_0 = v(\gamma_0)$. The integers $k_i$ correspond to different branches of $\log$. The details of the computation can be found in appendix \ref{app:elliptic}.

\subsubsection{The $r$ dependent integral}
We solve now the first, $\r$ dependent integral in \eqref{phi_int}
\begin{equation}
I_r := \int_{\r_0}^{\r} \frac{\al \left( \r'^2+\al^2 - \al \D \right) d\r'}{\Delta_{\r'} \sqrt{R}} \,. \label{I_r}
\end{equation}
Analogous to subsection \ref{rmotion} this integral can be solved by elementary functions if $R$ has a zero with multiplicity $4$ or more or two zeros with multiplicity $2$ or more \cite{GradshteynRyzhik83}.

If we consider light, i.e. $\delta=0$, $R$ is in general of degree $4$ and $I_r$ is of elliptic type and third kind. In this case it can be solved analogously to $I_\theta$. The same substitutions $\r = \frac{1}{\xi}+\r_R$ and $\xi = \frac{1}{b_3}\left(4y-\frac{b_2}{3}\right)$ as in subsection \ref{rmotion} for the case $\delta=0$, a subsequent partial fraction decomposition, and the final substitution $y=\wp(v)$ result in 
\begin{align}
\frac{b_3}{|b_3|}I_r & = \sum_{i=1}^4 C_i \int_{v_0}^v \frac{dv}{\wp(v)-y_i} - \frac{\al \left( \r_R^2+\al^2 - \al \D \right)}{\Delta_{\r=\r_R}} \int_{v_0}^v dv \,,
\end{align}
where $y_i$ are the four zeros of $\Delta_{y(\r)}$, $b_3$ defined as in \eqref{sol_r light}, and $C_i$ are the coefficients of the partial fractions dependent on the parameters and $\r_R$. The four functions $f_i(v) = (\wp(v)-y_i)^{-1}$ have simple poles in $v_{i1}, v_{i2}$ with $\wp(v_{i1}) = y_i = \wp(v_{i2})$ and have to be integrated wih the method presented in appendix \ref{app:elliptic}. Then $I_r$ is given by
\begin{multline} \label{phirequation_elliptic}
\frac{b_3}{|b_3|} I_r = \sum_{i=1}^4 \sum_{j=1}^2 \frac{C_i}{\wp'(v_{ij})} \bigg[ \zeta(v_{ij})(v-v_0) + \log \sigma(v-v_{ij}) \\ 
- \log \sigma(v_0-v_{ij}) \bigg] - \frac{\al \left( \r_R^2+\al^2 - \al \D \right)}{\Delta_{\r=\r_R}} (v-v_0)\,,
\end{multline}
where $v=v(\gamma)=\gamma-\gamma_{\r,\rm in}$, $v_0=v(\gamma_0)$ with $\gamma_{\r,\rm in}$ as in \eqref{sol_r light}. In the same way $I_r$ can be solved if $R$ has a double or triple zero.

If we consider particles, i.e. $\delta=1$, and assume that $R$ has only simple zeros, $I_r$ is of hyperelliptic type and third kind. The details of the solution method can be found in appendix \ref{app:hyperelliptic} but we give an outline here: First, we transform analogously to section \ref{rmotion} to the standard form by $\r= \pm 1/u+\r_R$ with a zero $\r_R$ of $R$. Afterward we simplify the integrand by a partial fraction decomposition which allows us to express $I_r$ in terms of the canonical holomorphic differentials $d\vec z$ \eqref{def_holomorphic} and the canonical differential of third kind $dP(x_1,x_2)$ \eqref{def_dP}. These differentials can then be expressed in dependence of the normalized Mino time $\gamma$. If we define $w = w(\gamma) = \sqrt{c_{5}} \gamma - \gamma_{\r,\rm in}$ and $w_0 =w(\gamma_0)$ the result is
\begin{multline} \label{sol I_r}
I_r = - \frac{\al u_0}{\sqrt{c_{5}} |u_0|} \bigg\{ C_1 (w-w_0)  + C_0 ( f(w)-f(w_0) ) \\
+ \sum_{i=1}^4 \frac{C_{2,i}}{\sqrt{R_{u_i}}} \bigg[ \frac{1}{2} \log \frac{\sigma(W^+(w))}{\sigma( W^-(w) )} - \frac{1}{2} \log \frac{\sigma( W^+(w_0))}{\sigma( W^-(w_0) )} \\
- \big( f(w)-f(w_0), w-w_0 \big) \left( \int_{u_i^-}^{u_i^+} d\vec r\right) \bigg] \bigg\} \,.
\end{multline}
where the constants $C_i$ are the coefficients of the partial fractions, $u_i$ are the four zeros of $\Delta_{\r=\pm 1/u+\r_R}$, $u_0=\pm (\r-\r_R)^{-1}$, and $c_5$ and $R_u$ as in \eqref{R_u}. The functions $W^{\pm}$ are defined by $W^{\pm}(w):=(f(w),w)^t - 2 \int_\infty^{u_i^{\pm}} d\vec z$, where the points $u_i^{\pm} = (u_i, \pm \sqrt{R_{u_i}})$ on the Riemann surface of $y^2=R_u$ are the pole $u_i$ located on the positive and the negative branch of the square root.

\subsection{$t$ motion}
The equation for $t$ \eqref{dot t_sn}
\begin{equation*}
\frac{dt}{d\gamma} = \chi^2 \bar M \left[ \frac{\r^2+\al^2}{\Dx} \mathbbm{P}  - \frac{\al}{\Dt} \mathbbm{T} \right]
\end{equation*}
is as complicated as the equation for $\phi$ motion. An integration yields
\begin{align}
t-t_0 & = \chi^2 \bar M \left[ \int_{\gamma_0}^\gamma \frac{\r^2+\al^2}{\Dx} \mathbbm{P} d\gamma - \int_{\gamma_0}^\gamma \frac{\al}{\Dt} \mathbbm{T} d\gamma \right] \nonumber\\
& = \chi^2 \bar M \left[ \int_{\r_0}^{\r} \frac{(\r^2+\al^2)\mathbbm{P}}{\Dx \sqrt{R}} d\r - \int_{\theta_0}^\theta \frac{\al \mathbbm{T}}{\Dt \sqrt{\Theta}} d\theta \right] \nonumber \\
& = \chi^2 \bar M \left[ \tilde I_r - \tilde I_\theta \right] \,. \label{t_int}
\end{align}
Because we already demonstrated the solution procedure, we only give here the results for the most general cases. \comment{Details of the necessary solution steps can be found in the appendix.}

If $\nu \Theta_\nu$ in \eqref{Theta_nu} has only simple zeros the solution of the $\theta$ dependent part is given by
\begin{multline}
\tilde I_\theta = a_{3} (v-v_0) - \sum_{i=1}^2 \frac{a_{3} \al^2 \la}{4 \wp'(v_i)} \big[ \zeta(v_i) (v-v_0) \\ + \log \sigma(v-v_i) - \log \sigma(v_0-v_i) \big]
\end{multline}
where $a_{3}$ is defined as in \eqref{Theta_xi}, $\wp(v_1) = \frac{a_2}{12} - \frac{1}{4} \al^2\la a_3= \wp(v_2)$, and $v=v(\gamma)=2\gamma-\gamma_{\theta,\rm in}$, $v_0=v(\gamma_0)$ as in \eqref{sol I_theta}.

If we consider light, i.e. $\delta=0$, the solution for the $\r$ dependent part is given by
\begin{multline}
\frac{b_3}{|b_3|} \tilde I_r = \sum_{i=1}^4 \sum_{j=1}^2 \frac{\tilde C_i}{\wp'(v_{ij})} \bigg[ \zeta(v_{ij})(v-v_0) + \log \sigma(v-v_{ij}) \\ 
- \log \sigma(v_0-v_{ij}) \bigg] - \frac{(\r_R^2+\al^2) (\r_R^2+\al^2 - \al \D)}{\Delta_{\r=\r_R}} (v-v_0)\,,
\end{multline}
where $b_3$ is defined as in \eqref{sol_r light}, $\tilde C_i$ are the coefficients of the partial fractions, $\wp(v_{i1}) = y_i = \wp(v_{i2})$ with the four zeros $y_i$ of $\Delta_{y(\r)}$, $y$ and $v = v(\gamma) = \gamma - \gamma_{\r,\rm in}$, $v_0=v(\gamma_0)$ as in \eqref{sol_r light}.

If $R$ has only simple zeros and we consider timelike geodesics $\delta=1$ the solution of the $\r$ dependent part is given by
\begin{multline}
\tilde I_r = \frac{u_0}{|u_0| \sqrt{c_{5}}} \bigg\{ \tilde C_1 (w-w_0) + \tilde C_0 (f(w)-f(w_0)) \\
+ \sum_{i=1}^4 \frac{\tilde C_{2,i}}{\sqrt{R_{u_i}}} \bigg[ \frac{1}{2} \log \frac{\sigma(W^+(w))}{\sigma(W^-(w))} - \frac{1}{2} \log \frac{\sigma(W^+(w_0))}{\sigma(W^-(w_0))} \\
- (f(w)-f(w_0),w-w_0) \left( \int_{u_i^-}^{u_i^+} d\vec r \right) \bigg] \bigg\}
\end{multline}
where the notation is as in \eqref{sol I_r} and $\tilde C_0, \tilde C_1, \tilde C_{2,i}$ are the coefficients of the partial fractions. 

\section{Discussion of some geodesics}
In Sec.~\ref{section:Types} we discussed the general features of the different types of timelike geodesic motion in Kerr-de Sitter and Kerr-anti-de Sitter space-time. With the analytical solution derived in section \ref{section:exact solutions} at hand we want to discuss now some chosen geodesics. 

We start with orbits which highlight the influence of $\Lambda$ on the geodesics. From the results of section \ref{section:Types} we conclude that for $\Lambda>0$ there are four parameter regions where the changes compared to $\Lambda=0$ are most obvious. The first two are the regions (III) and (IV) with $E^2 < 1$, where we have additional flyby orbits not present for $\Lambda=0$. Third and fourth, the shift from region (V) of $\Lambda=0$ to region (II) of $\Lambda>0$ for $E^2= 1-\epsilon$, $\epsilon>0$ small, and the shift from region (III) of $\Lambda=0$ to region (IV) of $\Lambda>0$, again for $E^2=1-\epsilon$ are most interesting as the (outer) bound orbit becomes a flyby orbit. A plot of the corresponding orbits can be found in Fig. \ref{Fig:flybyorbits}.

\begin{figure}
\subfigure[$E^2=0.9$, $\bar L_z=-1$\newline(region (IVb))]{
\includegraphics[width=0.23\textwidth]{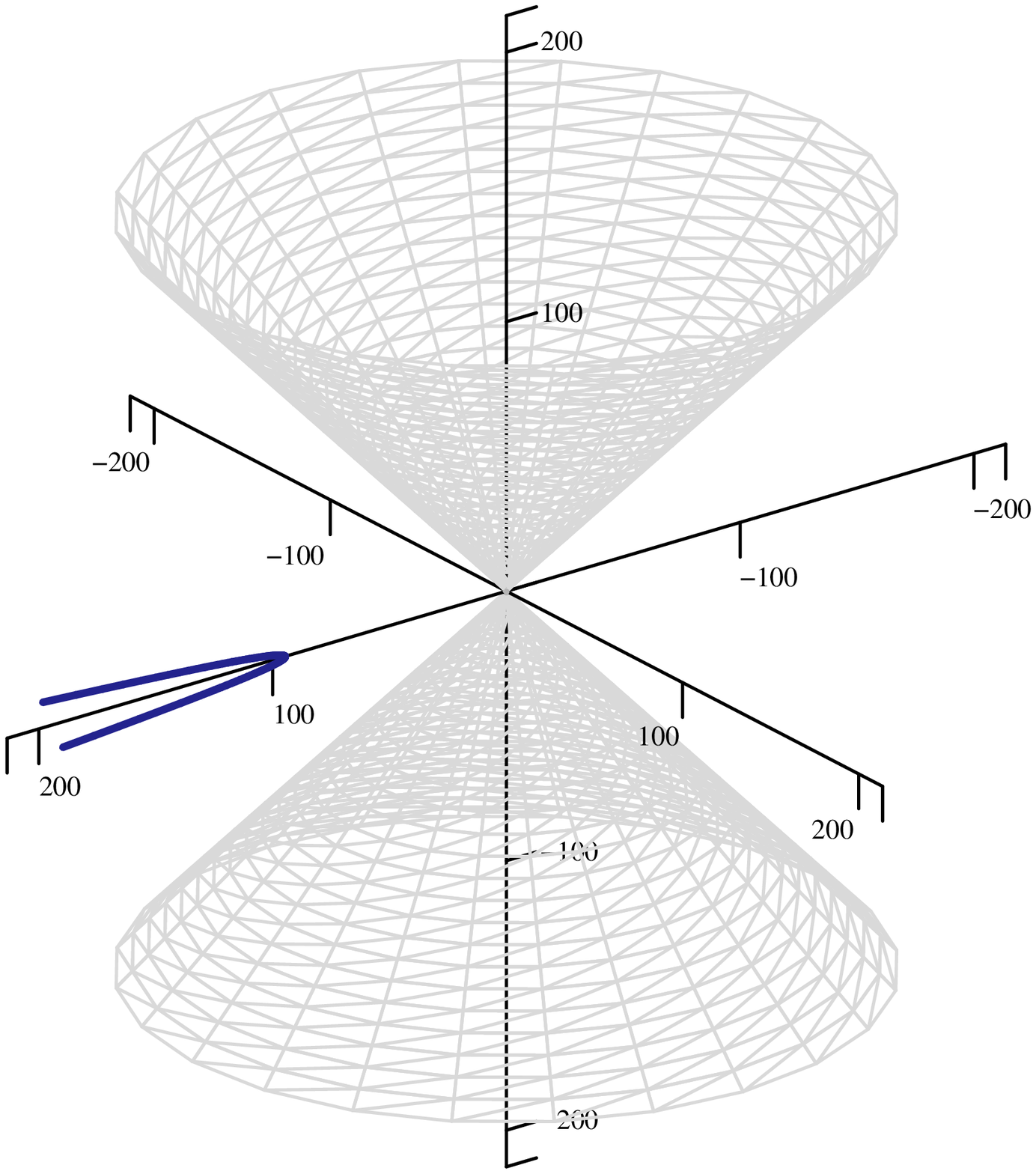}}
\subfigure[$E^2=0.94$, $\bar L_z=0.6$\newline(region(IIIb))]{
\includegraphics[width=0.23\textwidth]{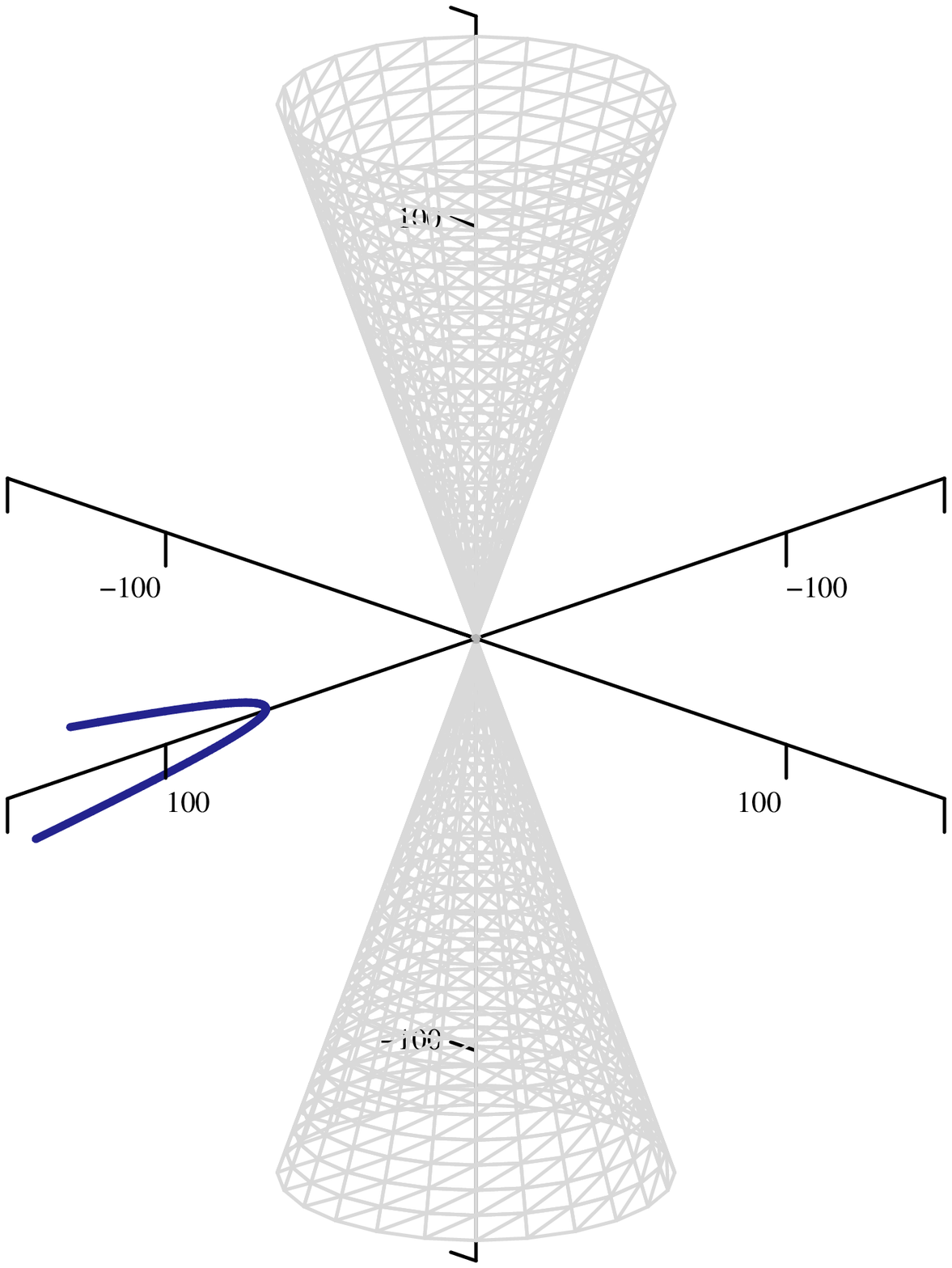}}
\subfigure[$E^2=0.97$, $\bar L_z=-1$\newline(region (IIb))]{
\includegraphics[width=0.23\textwidth]{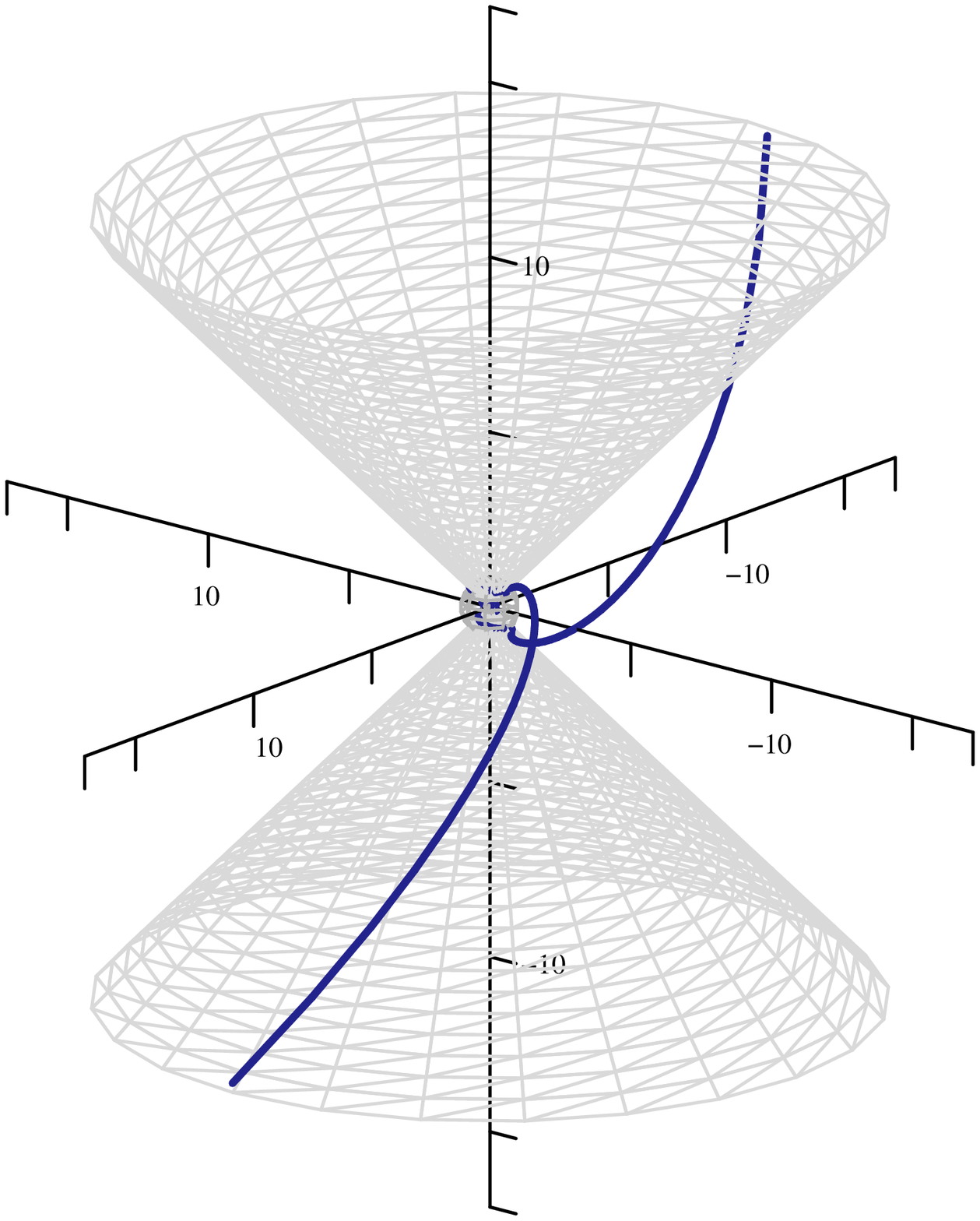}}
\subfigure[$E^2=0.97$, $\bar L_z=1.8$\newline(region (IVb)]{
\includegraphics[width=0.23\textwidth]{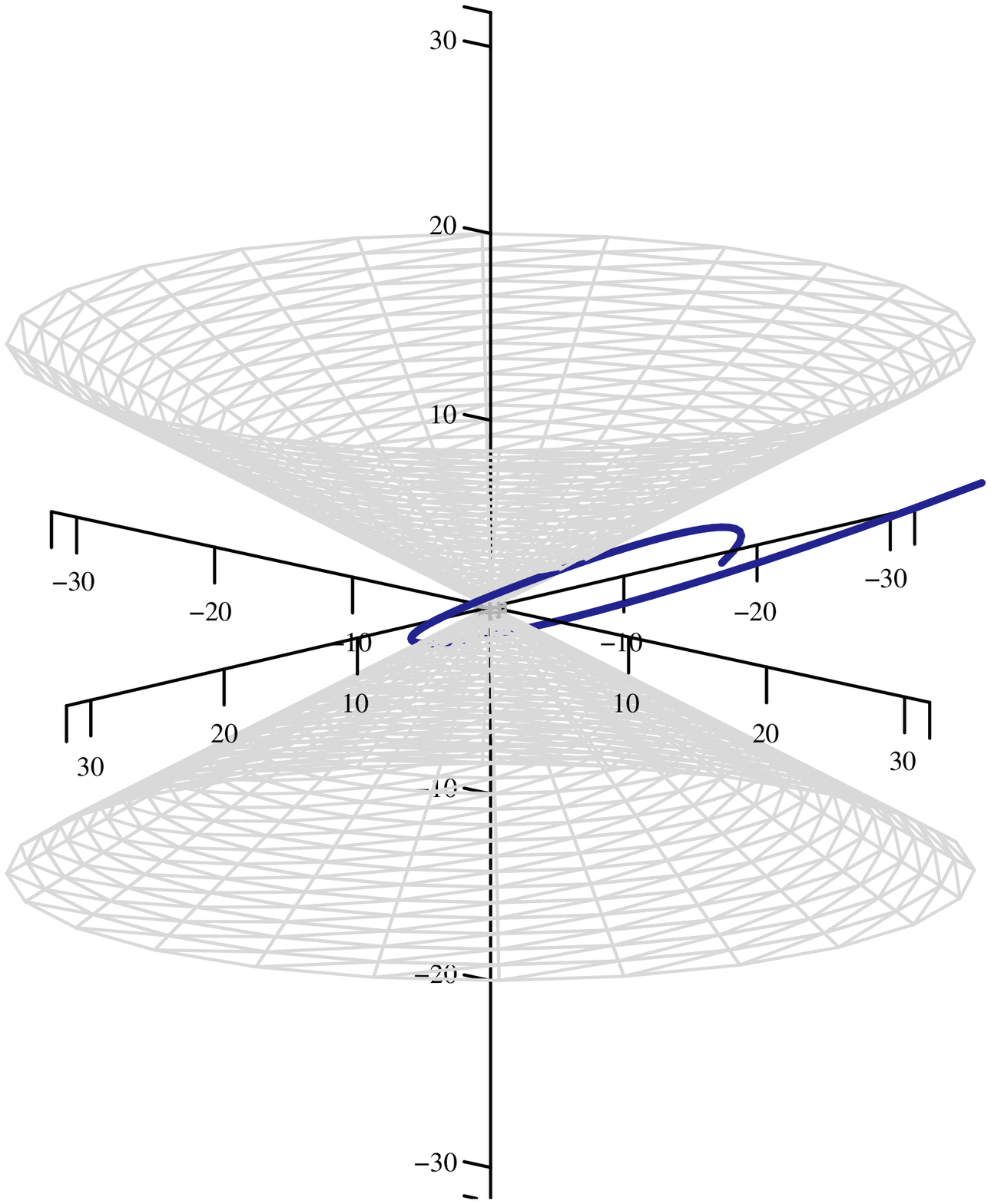}}
\caption{Flyby orbits for $\la=10^{-5}$, $\al=0.4$, and $\bar K=3$. For $\Lambda=0$ the parameters of (c) would have been located in region (Vb) and the parameters of (d) in region (IIIb). Light grey cones correspond to extremal $\theta$ and dark grey spheres to horizons.}
\label{Fig:flybyorbits}
\end{figure}

An important feature of geodesics in stationary axisymmetric space-times is motion of the nodes where the orbit of a test-particle or light intercepts the equatorial plane. This motion is caused by the $g_{0i}$ components of the space-time metric and known as the Lense-Thirring effect. In the weak field regime it becomes visible by a precession of the orbital plane, cp.~Fig.~\ref{Fig:boundorbit} for an obvious example. This orbital precession has been confirmed within an accuracy of about 10\% by the LAGEOS (Laser Geodynamics Satellite) mission \cite{Ciufolini07}, \footnote{Another method to observe the influence of the gravitomagnetic components $g_{0i}$ is through the precession of gyroscopes known also as Schiff effect. Such a measurement has been carried through by Gravity Probe B \cite{Everittetal09}. While the Lense-Thirring effect is an orbital effect involving the motion of the whole orbit thus constituting a global measurement, the Schiff effect is a local effect showing the dragging of local inertial frames due to the existence of the $g_{0i}$ components.}.

\begin{figure}
\includegraphics[width=0.3\textwidth]{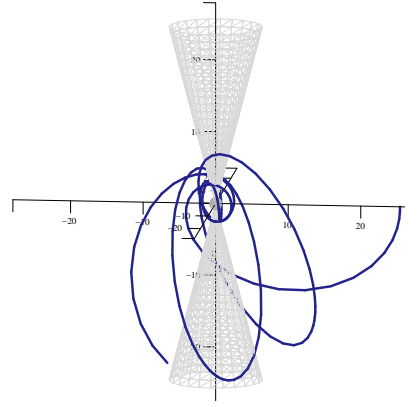}
\caption{Precession of the orbital plane for a bound orbit with $\la=10^{-5}$, $\al=0.4$, $\bar K=4.5$, $\bar L_z=-0.5$, and $E^2=0.96$ (region (IIIb)). As in Fig.~\ref{Fig:flybyorbits}, the cones and spheres correspond to extremal $\theta$ and horizons.}
\label{Fig:boundorbit}
\end{figure}

Let us also discuss some exceptional orbits related to multiple zeros of $R$, i.e. spherical orbits with constant $r$ and orbits asymptotically approaching a constant $r$. There are two types of spherical orbits: stable and unstable. Stable spherical orbits with $\r(\gamma) \equiv \r_0$ occur if radial coordinates adjacent to $\r_0$ are not allowed due to $R(\r)<0$, which happens if $\r_0$ is a maximum of $R$. Unstable spherical orbits with $\r(\gamma) \equiv \r_0$ are trajectories where radial coordinates $\r$ in the neighbourhood of $\r_0$ with $\r<\r_0$ or $\r>\r_0$ are allowed. Therefore, these orbits are related to a minimum or to an inflection point of $R$. If $\r_0$ is an inflection point, an asymptotic approach to $\r_0$ is only possible from  one side of $\r_0$ whereas this is possible from both sides if $\r_0$ is a minimum of $R$. Asymptotic orbits can also be devided into two types: unbound and bound. The latter case corresponds to orbits which approach for both $t \to \infty$ and $t \to -\infty$ a spherical orbit. Bound asymptotic orbits are also known as homoclinic orbits. If the asymptotic orbit is unbound it reaches $\r=\infty$ for either $t \to \infty$ or $t \to -\infty$.

For asymptotic bound or unbound trajectories corresponding to an unstable spherical orbits the equations of motion simplify considerably. In this case the equation for $\r(\gamma)$ as well as the $\r$ dependent integrals in the $\phi$ and $t$ equations are of elliptic type and can be solved in terms of Weierstrass elliptic functions, see \eqref{requation_elliptic} and \eqref{phirequation_elliptic}. Note that these solutions are not limited to the case of equatorial circular orbits but are valid for all types of asymptotic orbits and, thus, generalize the analytical solutions for homoclinic orbits in \cite{LevinPerez-Giz09} not only to Kerr-de Sitter space-time but also to arbitrary inclinations. 

From all spherical orbits the Last Stable Spherical Orbit (LSSO), in particular, the Innermost Stable Circular Orbit (ISCO) in the equatorial plane are of importance as they represent the transition from stable orbits to those which fall through the event horizon. The corresponding multiple zero of $R$ appears at the boundaries of the different regions of $r$ motion, cp.~Fig.~\ref{Fig:r_EL}. Because necessarily $\bar K = \chi^2 E^2(\al-\D)^2$ for equatorial orbits, from this we can determine the LSSO for given $\al$, $\la$, $\bar K$ and the ISCO for given $\al$, $\la$ by solving first
\begin{equation} \label{LSSO}
R(\r)=0\,, \quad \frac{dR}{d\r}(\r)=0\,, \quad \text{and} \quad \frac{d^2R}{d\r^2}(\r)=0
\end{equation}
for $\r \geq \r_{h}$ with the event horizon $\r_{h}$. The solutions are limiting cases of the LSSO or ISCO and are given by the corner points on the borders of region (III) of the $\r$ motion (as corners on other boundaries correspond to $\r<\r_{h}$). From the results of \eqref{LSSO} we search for the smallest possible double zero $\r$ which is a maximum. In the case of the ISCO in the equatorial plane we are now done. For the LSSO, we have to check in addition whether the corresponding values of $E^2(\r)$ and $\bar L_z(\r)$ (given by \eqref{doublezerosr}) are located in an allowed region of the $\theta$ motion. If this is the case, we found the LSSO. If not, we can determine the LSSO as the intersection point of the boundary of region (III) with a boundary of an allowed $\theta$ region. Note that it is not possible to determine an LSSO (for given $\al$, $\la$, and $\bar K$) if there is no spherical orbit at all outside the event horizon which happens if no boundary of the $\r$ motion is located in an allowed region of the $\theta$ motion. As an example, this is the case for $\la=10^{-5}$, $\al=0.1$, and $\bar K=0.1$. Also, the LSSO is identical with the ISCO if it is given as an intersection point with the boundary of region (b) of the $\theta$ motion. For examples of spherical orbits see Figs.~\ref{Fig:isso} and \ref{Fig:asymptotic}. Note that within the event horizon there may be additional stable spherical orbits.

\begin{figure}
\subfigure[$\, \al=0.1$, $\la=10^{-5}$, $\bar K=3$]{
\includegraphics[width=0.23 \textwidth]{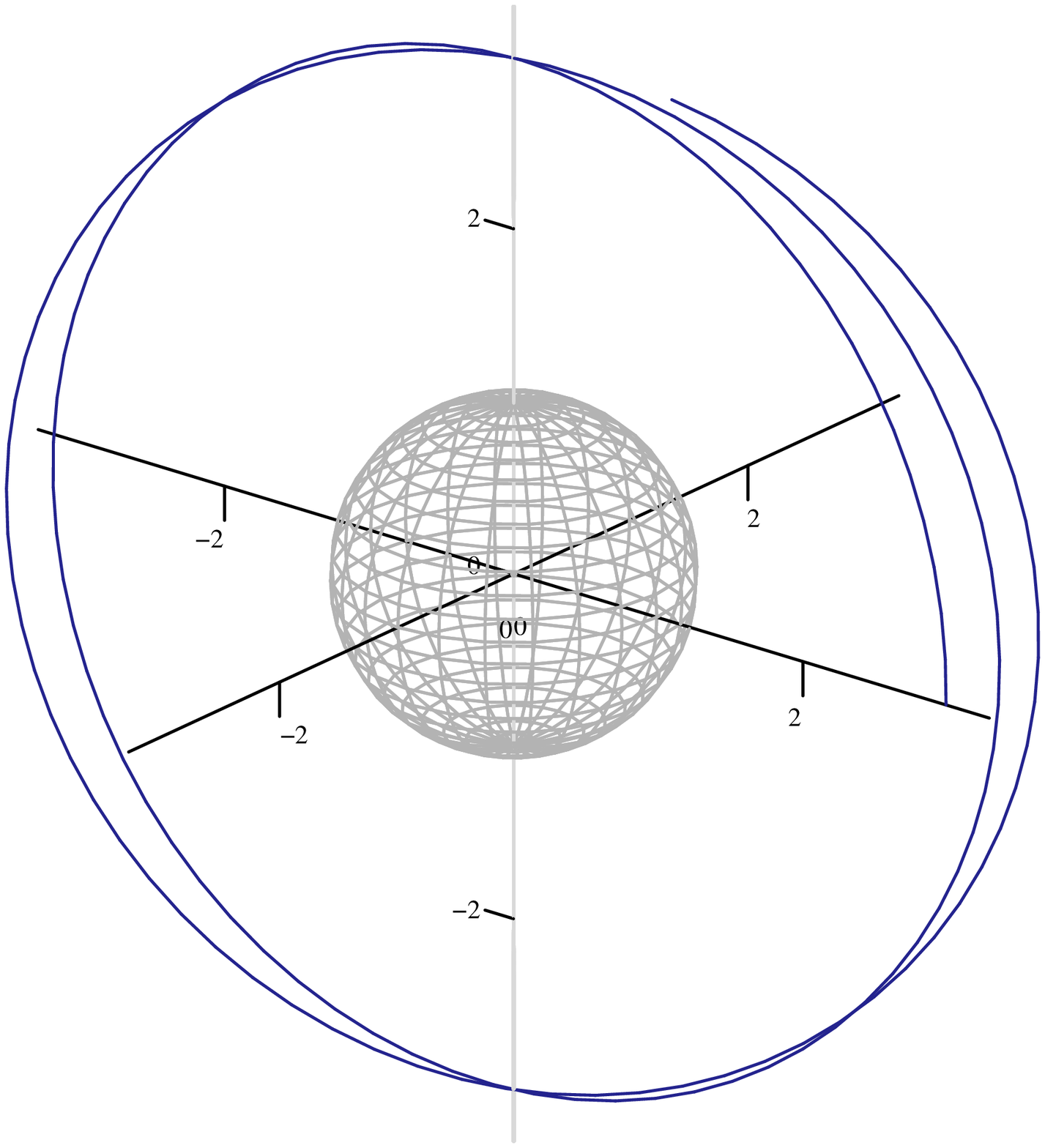}}
\subfigure[$\, \al=0.4$, $\la=10^{-5}$, $\bar K=3$]{
\includegraphics[width=0.23 \textwidth]{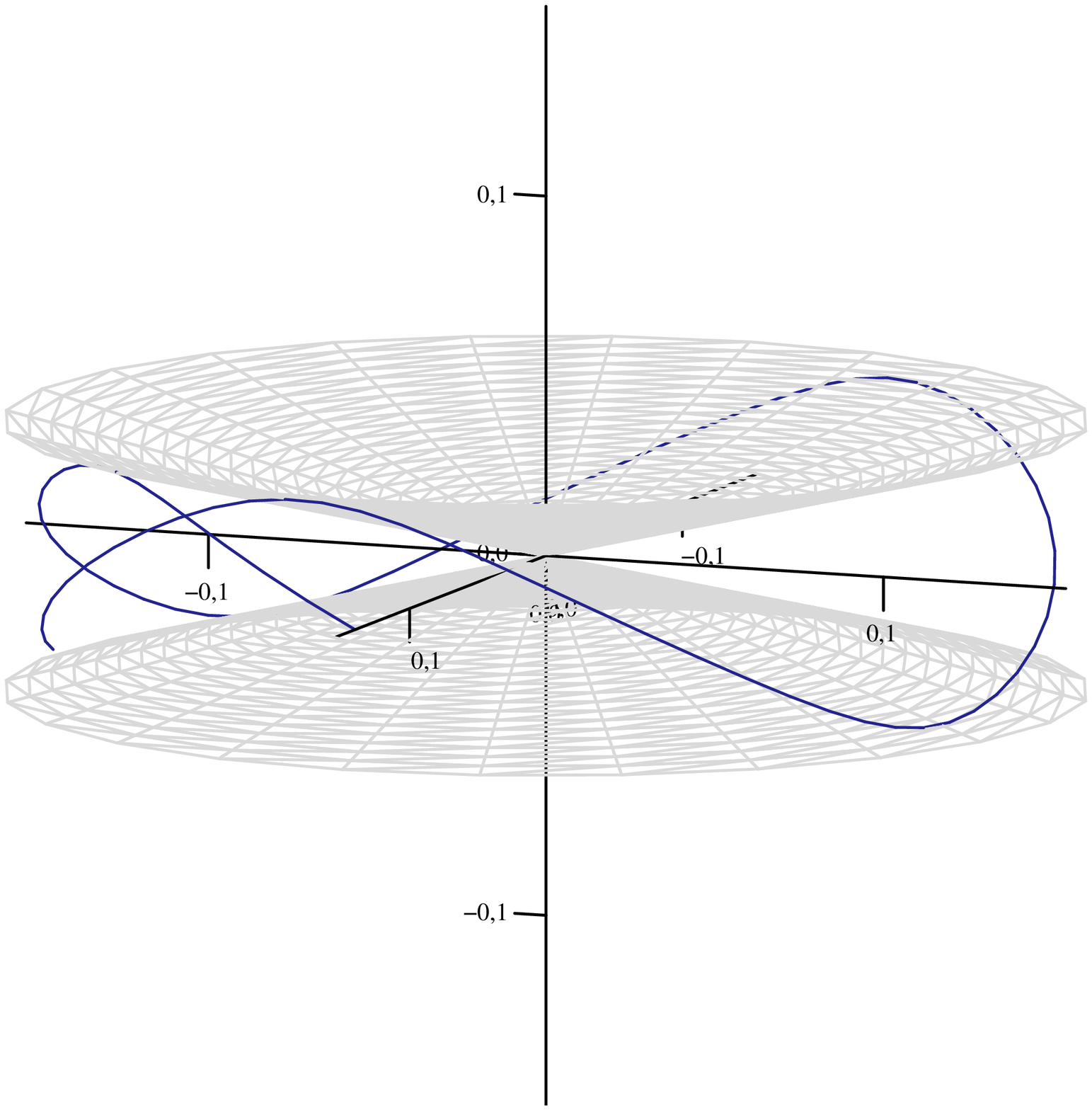}}
\caption{(a): Last Stable Spherical orbit at $\r(\gamma)\equiv 2.98982782315$. The corresponding parameter values are (approximately) $E^2=0.88830740016$ and $\bar L_z = 0.00154804028$ (lower left corner on boundary of region (IIIb)). (b): Stable spherical orbit at $\r(\gamma)\equiv0.14$ within the inner Cauchy horizon. Here $E^2=124.76756414675$ and $\bar L_z=5.879796033955$ (region (IIb)).}
\label{Fig:isso}
\end{figure}

\begin{figure}
\includegraphics[width=0.23\textwidth]{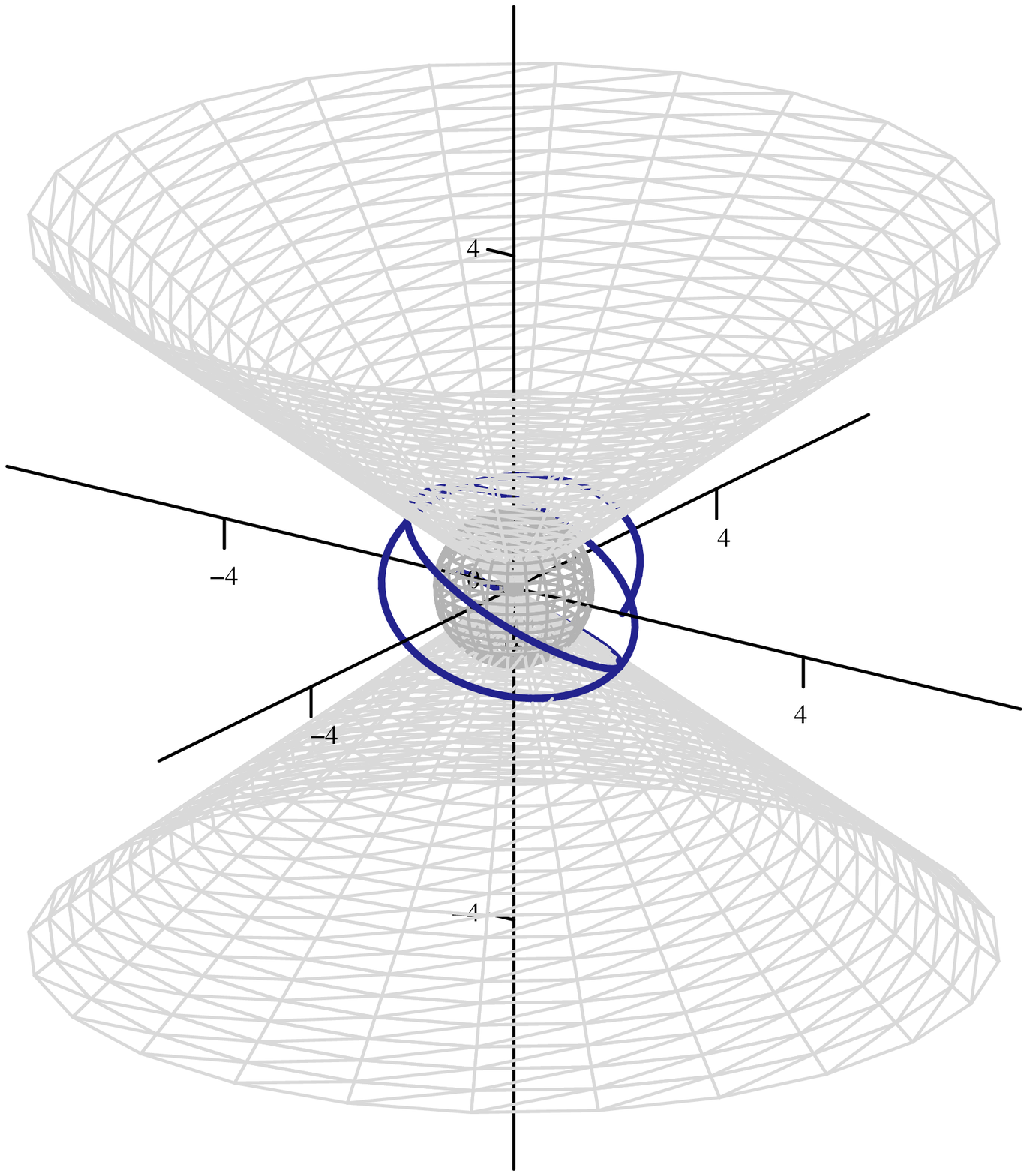}
\includegraphics[width=0.23\textwidth]{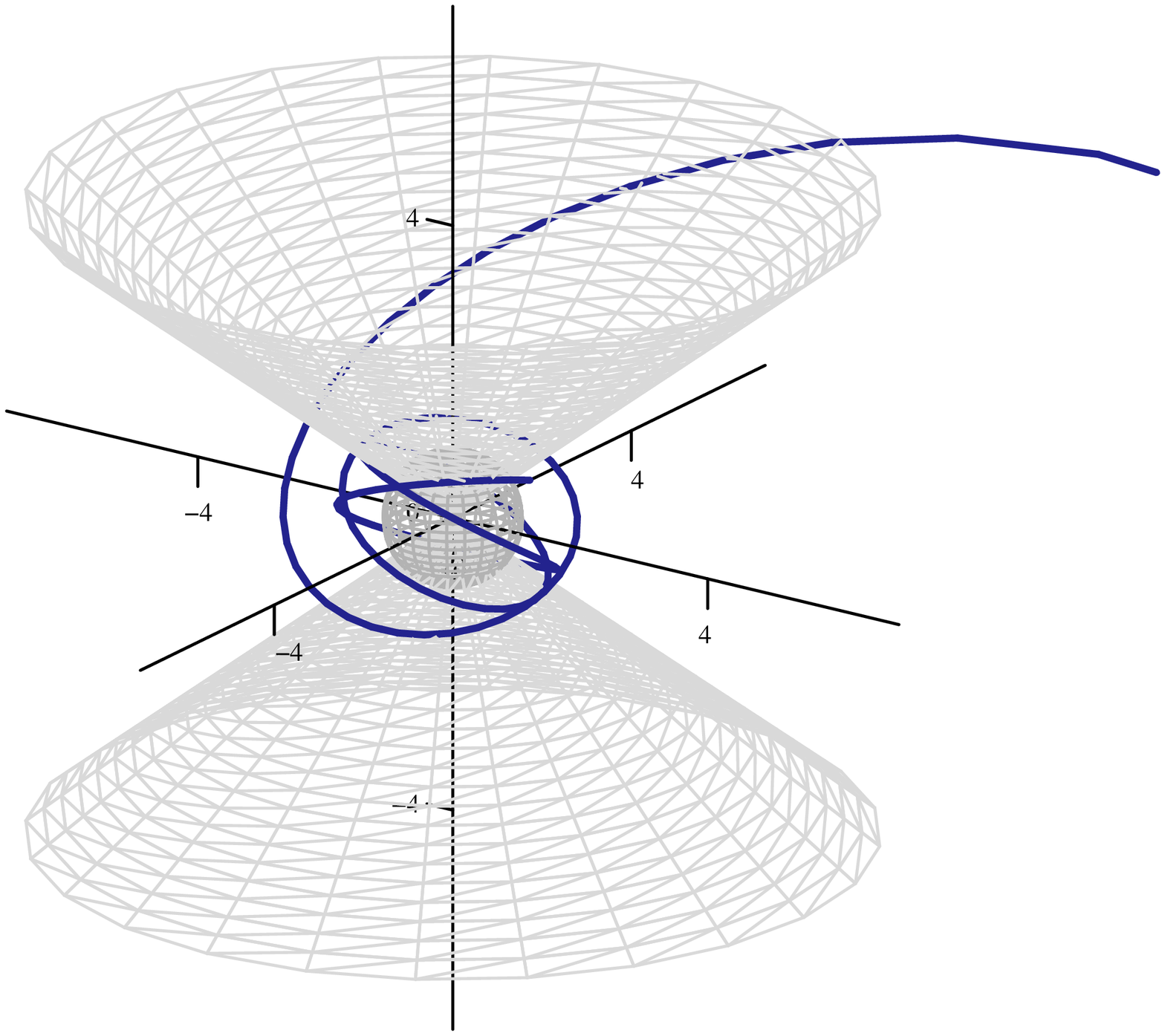}
\caption{Unstable spherical orbit with $\r(\gamma)\equiv1.5$ (left) and asymptotic approach (right). Here $\al=0.3$, $\la=10^{-5}$, $\bar K=2$, $E=0.9720311146$ and $\bar L_z=2.5677957370$ (boundary of region (IIIb)).}
\label{Fig:asymptotic}
\end{figure}

\section{Analytic expressions for observables}
For the understanding of characteristic features of space-times by measurements of geodesics in that space-time it is crucial to identify certain theoretical quantities of observables. For flyby orbits, this can be the deflection angle of the geodesic whereas for bound orbits it is of interest to determine the orbital frequencies as well as the periastron shift and the Lense-Thirring effect. 

Let us first consider flyby orbits. The deflection angle of such an orbit depends on the two values $\gamma_\infty^\pm$ of the normalized Mino time for which $\r(\gamma_\infty^\pm) = \infty$. These are given by  
\begin{align}
\gamma_{\infty}^\pm = \frac{1}{\sqrt{c_{5}}} \left( \int_{u_0}^0 \frac{u du}{\sqrt{R_u}} - \gamma_0\right)
\end{align}
for the two branches of $\sqrt{R_u}$. Therefore, we can calculate the values of $\theta$ and $\phi$ for $\r \to \infty$ that are taken by this flyby orbit by $\theta^\pm = \theta(\gamma_\infty^\pm)$ and $\phi^\pm = \phi(\gamma_\infty^\pm)$. The deflection angles are then given by $\Delta \theta = \theta^+ - \theta^-$ and $\Delta \phi = \phi^+ - \phi^-$.

For bound orbits we can identify three orbital frequencies $\Omega_r$, $\Omega_\theta$ and $\Omega_\phi$ associated with the coordinates $r$, $\theta$ and $\phi$. The precessions of the orbital ellipse, which in the weak field regime can be identified with the periastron shift, and the orbital plane, which in the weak field regime can be identified with the Lense-Thirring effect, are induced by mismatches of these orbital frequencies. More precisely, the orbital ellipse precesses at $\Omega_\phi-\Omega_r$ and the orbital plane at $\Omega_\phi-\Omega_\theta$. 

Let us consider the orbital frequency $\Omega_r$. For bound orbits the coordinate $\r$ is contained in an intervall $\r_p \leq \r \leq \r_a$ with the peri- and apoapsis distances $\r_p$ and $\r_a$. The orbital period $\omega_{\r}$ defined by $\r(\gamma + \omega_{\r}) = \r(\gamma)$ is then given by a complete revolution from $\r_p$ to $\r_a$ and back (with reversed sign of the square root) to $\r_p$,
\begin{align} \label{omega_r}
\omega_{\r} = 2 \int_{\r_p}^{\r_a} \frac{d\r}{\sqrt{R}} \,.
\end{align}
The orbital frequency of the $r$ motion with respect to $\gamma$ is then given by $\frac{2\pi}{\omega_{\r}}$. For the calculation of $\Omega_r$, which represents the orbital frequency with respect to $t$, we need in addition the average rate $\Gamma$ at which $t$ accumulates with $\gamma$. This will be determined below.

For the calculation of the orbital frequency $\Omega_\theta$ we have to determine the orbital period $\omega_\theta$ such that $\theta(\gamma+\omega_\theta) = \theta(\gamma)$. The $\theta$ motion is likewise bounded by $\theta_{\min} \leq \theta \leq \theta_{\max}$ for two real zeros $\theta_{\rm min}, \theta_{\rm max} \in (0,\pi)$ of $\Theta$ and, therefore,  
\begin{align}
\omega_\theta = 2 \int_{\theta_{\rm min}}^{\theta_{\rm max}} \frac{d\theta}{\sqrt{\Theta}} \,.
\end{align}
Again, the orbital frequency of the $\theta$ motion with respect to $\gamma$ is given by $\frac{2\pi}{\omega_\theta}$. 

The orbital periods of the remaining coordinates $t$ and $\phi$ has to be treated somewhat differently because they depend on both $\r$ and $\theta$. The solutions $t(\gamma)$ and $\phi(\gamma)$ consist of two different parts, one which represents the average rates $\Gamma$ and $Y_\phi$ at which $t$ and $\phi$ accumulate with $\gamma$ and one which represents oscillations around it with periods $\omega_{\r}$ and $\omega_\theta$ \cite{DrascoHughes04, FujitaHikada09}. The periods $\Gamma$ and $Y_\phi$ can be calculated by \cite{FujitaHikada09}
\begin{align}
\Gamma & = \frac{2}{\omega_{\r}} \int_{\r_p}^{\r_a} \frac{(\r^2+\al^2) \mathbbm{P}}{\Dr \sqrt{R}} d\r - \frac{2}{\omega_\theta} \int_{\theta_{\rm min}}^{\theta_{\rm max}} \frac{\al \mathbbm{T} d\theta}{\Dt \sqrt{\Theta}} \,, \label{Y_t}\\
Y_\phi & = \frac{2}{\omega_{\r}} \int_{\r_p}^{\r_a} \frac{\al \mathbbm{P}}{\Dr \sqrt{R}} dr - \frac{2}{\omega_\theta} \int_{\theta_{\rm min}}^{\theta_{\rm max}} \frac{\mathbbm{T} d\theta}{\Dt \sin^2\theta \sqrt{\Theta}} \,. \label{Y_phi}
\end{align}
The orbital frequencies $\Omega_r$, $\Omega_\theta$, and $\Omega_\phi$ are then given by
\begin{equation}
\Omega_r = \frac{2\pi}{\omega_{\r}}\, \frac{1}{\Gamma}\,, \quad \Omega_\theta = \frac{2\pi}{\omega_\theta}\, \frac{1}{\Gamma} \,, \quad \Omega_\phi = \frac{Y_\phi}{\Gamma} \,.
\end{equation}

If the integral expressions for $\omega_r$ and the $\r$ dependent parts of $\Gamma$ and $Y_\phi$ degenerate to elliptic or elementary type, i.e. if we consider light or $R$ possesses multiple zeros, we can find analytical expressions for \eqref{omega_r}, \eqref{Y_t}, and \eqref{Y_phi} with the techniques presented in \cite{FujitaHikada09}. If $R$ has only simple zeros and $\delta=1$, the integral $\omega_r$ is an entry of the fundamental period matrix $\omega$ which enters in the definition of the period lattice $\Gamma$ of the holomorphic differentials $d\vec z$, $\Gamma=\{\omega v+\omega'v\,|\, v,v' \in \mathbbm{Z}^2\}$. The more complicated integrals involving $R$ in \eqref{Y_t} and \eqref{Y_phi} can be rewritten in terms of periods of the differentials of second kind $d\vec r$ and of third kind $dP(x_1,x_2)$ by a decomposition in partial fractions. In this way, expressions for $\omega_r$, $\Gamma$ and $Y_\phi$ which are totally analogous to the elliptic case can be obtained. 

It follows that the periastron shift is given by
\begin{equation}
\Delta_{\rm periastron} = \Omega_\phi - \Omega_r = \left( Y_\phi - \frac{2\pi}{\omega_{\r}} \right) \frac{1}{\Gamma} 
\end{equation}
and the Lense-Thirring effect by
\begin{equation}
\Delta_{\rm Lense-Thirring} = \Omega_\phi - \Omega_\theta = \left( Y_\phi - \frac{2\pi}{\omega_{\theta}} \right) \frac{1}{\Gamma} \,.
\end{equation}
Another way to access information encoded in the orbits is through a frequency decomposision of the whole orbit \cite{DrascoHughes04}. This will be analyzed elsewhere.

\section{Summary and outlook}
In this paper we derived the analytical solution for both timelike and lightlike geodesic motion in Kerr-(anti-)de Sitter space-time. The analytical expressions for the orbits $(r,\theta,\phi,t)$ are given by elliptic Weierstrass and hyperellliptic Kleinian functions in terms of the normalized Mino time. We also presented a method for solving differential equations of hyperelliptic type and third kind, and applied it to the equations of motion for $\phi$ and $t$. We classified possible types of geodesic motion by an analysis of the zeros of the polynomials underlying the $\theta$ and $\r$ motion and discussed the influence of a non-vanishing cosmological constant on the orbit types. Some particular interesting orbits not present in Kerr space-time were shown and a systematic approach for determining the last stable spherical and circular orbits was presented. In addition, we derived the analytic expressions for observables connected with geodesic motion in Kerr-de Sitter space-time, namely the deflection angle for escape orbits as well as the orbital frequencies, the periastron shift, and Lense-Thirring effect of bound orbits.

The results of this paper can be viewed as the starting point for the analysis of several features of geodesics in Kerr-de Sitter space-time not treated in this publication. Although mathematical analogous to the case of slow Kerr-de Sitter a complete discussion of orbits in fast and extreme Kerr-de Sitter space-times may lead to special features and should be carried through. Also, it would be interesting to study bound geodesics crossing $\r=0$ (and mayby also the Cauchy horizon for positive $\r$) in general and, in particular, their causal structure. In this context the analysis of closed timelike trajectories is also of interest. In addition, we not yet considered geodesics lying entirely on the axis $\theta=0,\pi$ or even crossing it. Until now, we only considered the Boyer-Lindquist form of the Kerr-de Sitter metric which is not a good choice for considering geodesics which fall through a horizon. Therefore, for a future publication it would be interesting to use a coordinate-singularity free version of the metric.

The methods for obtaining analytical solutions of geodesic equations presented in this paper are not only limited to Kerr-de Sitter space-time. Indeed, they had already been used to solve the geodesic equation in Schwarzschild-de Sitter \cite{HackmannLaemmerzahl08,HackmannLaemmerzahl08b} and Reissner-Nordstr\"om-de Sitter \cite{Hackmannetal08} space-times. Also the geodesic equations in higher-dimensional static spherically symmetric  space-times \cite{Hackmannetal08} were solved by these methods. The same type of differential equations is also present in the Pleba\'{n}ski-Demia\'{n}ski space-time without acceleration, which is the most general space-time with separable Hamilton-Jacobi equation. The analytical solution of the geodesic equation in this space-time is given in \cite{Hackmannetal09} but will be elaborated in an upcoming publication. It will also be interesting to apply the presented methods to higher dimensional stationary axially symmetric space-times like the Myers-Perry solutions. 

The same structure of equations we solved in this paper is also present in the geodesic equation of the effective one-body formalism of the relativistic two-body problem. The effective metric in this formalism can be described as a perturbed Schwarzschild or Kerr metric, where the pertubation is given in powers of the radial coordinate $r$ \cite{Damour01,Damouretal08,Damouretal08b}. Therefore, we expect that the polynomial appearing in the resulting equations of motion will have a higher degree than the corresponding polynomial in the Schwarzschild or Kerr case and, thus, that it is necessary to generalize the elliptic functions used in these cases to hyperelliptic functions used in this paper. A similar situation can be found in the expressions of axisymmetric gravitational multipole space-times. For example, some types of geodesics in Erez-Rosen space-time, which reduces to the Schwarzschild case if the quadrupole moment is neglected, were already solved analytically \cite{Quevedo90, Quevedo89}. We expect that the methods presented in this paper will be helpful to solve geodesics in space-times with multipoles.

Analytic solutions are the starting point for approximation methods for the description of real stellar, planetary, comet, asteroid, or satellite trajectories
(see e.g. \cite{Hagihara70}). In particular, it is possible to derive post-Kerr, post-Schwarzschild, or post-Newton series expansions of analytical solutions.
Due to the, in principle, arbitrary high accuracy of analytic solutions of the geodesic equation they can also serve as test beds for numerical codes for the dynamics of binary systems in the extreme stellar mass ratio case (extreme mass ratio inspirals, EMRIs) and also for the calculation of corresponding gravitational wave templates. For the case of Kerr space-time with vanishing cosmological constant it has already been shown that gravitational waves from EMRIs can be computed more accurately by using analytical solutions than by numerical integration \cite{FujitaHikada09}. 

Due to the high precission, the analytical expressions for observables in Kerr-de Sitter space-time may be used for comparisons with observations where the influence of the cosmological constant might play a role. This could be the case for stars moving around the galactic center black hole or binary systems with extreme mass ratios where one body serves as test-particle. For example, quasar QJ287 (cp. \cite{Valtonenetal08}) could be a candidate for observing the effects of a non-vanishing cosmological constant. In this context it would also be interesting for a future publication to derive post-Kerr, post-Schwarzschild, or post-Newton expressions for observables.

\begin{acknowledgments}
We are grateful to H. Dullin, W. Fischer, and P. Richter for helpful discussions. V.K. thanks the German Academic Exchange Service DAAD and E.H. the German Research Foundation DFG for financial support.
\end{acknowledgments}


\appendix

\section{Integration of elliptic integrals of the third kind} \label{app:elliptic}
In this appendix we will demonstrate the details of the integration method of the elliptic integrals of third kind which appear in the $\theta$ dependent part of the $\phi$ and $t$ motion \eqref{phi_int}, \eqref{t_int} and in the special case of the $r$ motion where $R$ has a double or triple zero. We will explain the procedure for the example of the integral $I_\theta$ in \eqref{phi_int}. As this integral is only elliptic if $\Theta$ has only simple zeros we assume in the following that this is the case.

Before we demonstrate the calculational steps, we will summarize them for convenience:
\begin{enumerate}
\item Cast the expression under the square root in the standard Weierstrass form $4y^3-g_2y-g_3$ for some constants $g_2,g_3$, the so-called Weierstrass invariants.
\item Decompose the integrand (without the square root) in partial fractions.
\item Substitute $y=\wp(v)$.
\item For every partial fraction, rewrite the integrand in terms of the Weiertsrass $\wp$ (double pole) or $\zeta$ function (simple pole).
\item Integrate the Weierstrass $\wp$ and $\zeta$ functions and assemble all parts.
\end{enumerate}

Let us start from eq. \eqref{I_theta}
\begin{align*}
I_\theta & = \frac{1}{2} \int_{\nu_0}^\nu \frac{\al-\D-\al \nu'}{\Delta_{\nu'} (1-\nu') \sqrt{\nu' \Theta_{\nu'}}} d\nu' \,,
\end{align*}
where we assumed that the original integration path was fully contained in $(0,\frac{\pi}{2}]$. With the substitutions $\nu=\xi^{-1}$ and $\xi=\frac{1}{a_{3}} \left( 4y-\frac{a_{2}}{3}\right)$, as in subsection \ref{thetamotion}, we obtain
\begin{align}
I_\theta & = \frac{1}{2} \int_{\xi_0}^\xi \frac{(\al-\D)\xi'-\al}{\Delta_{\xi'} (\xi'-1) \sqrt{\Theta_{\xi'}}} \xi' d\xi' \nonumber \\
& = \frac{|a_{3}|}{2 a_{3}} \int_{y_0}^y \frac{ ((\al-\D) (4y'-\frac{a_{2}}{3}) - \al a_{3})(4y'-\frac{a_{2}}{3}) dy'}{ a_{3} \Delta_{y'} (4y'-\frac{a_{2}}{3}-a_{3}) \sqrt{4y'^3-g_2y'-g_3}} \,, \label{phitheta_y}
\end{align} 
where $\Theta_\xi$ is defined in \eqref{Theta_xi}, $g_2, g_3$ are defined as in \eqref{def_g2g3}, $\Delta_{\xi} = \xi+\al^2 \la$, and $\Delta_y = d_1 y+d_2 = \frac{4}{a_{3}}y - \frac{a_{2}}{3 a_{3}} + \al^2 \la$. Now we simplify the integrand in \eqref{phitheta_y} (without the square root) by a partial fraction decomposition
\begin{multline} \label{partial fraction}
\frac{((\al-\D) (4y-\frac{a_{2}}{3}) - \al a_{3}) (4y-\frac{a_{2}}{3})}{a_{3} \Delta_{y} (4y-\frac{a_{2}}{3}-a_{3})} \\ 
= (\al - \D) - \frac{a_{3}}{4 \chi} \left[ \frac{\al^3 \la(\chi-\al \la \D)}{y + \frac{d_2}{d_1}} + \frac{\D}{y-d_3} \right] \,,
\end{multline}
where $d_3=\frac{a_{2}}{12} + \frac{a_{3}}{4}$. With the substitution $y = \wp(v)$ we can get rid of the square root in \eqref{phitheta_y} as $\wp'(v) = \pm \sqrt{4\wp^3(v)-g_2\wp(v)-g_3}$ where the sign has to be chosen according to the sign of $\wp'$ and the branch of the square root. The function $\wp'(v)$ is negative for  $v \in [0,\omega_1]$ and positive for $v \in [\omega_2,\omega_2+\omega_1]$ where $2 \omega_1 \in \mathbbm{R}$ and $2 \omega_2 \in \mathbbm{C}$ are the fundamental periods of $\wp$. As $\theta=\frac{\pi}{2}$ corresponds to $y=\infty$ we will have $v \in [0,\omega_1]$ in most cases. 

Altogether, the integral $I_\theta$ now reads
\begin{multline}
I_\theta = \frac{|a_{3}|}{2 a_{3}} \bigg\{ (\al-\D) \int_{v_0}^v dv' - \frac{a_{3}}{4 \chi} \bigg[ \al^3 \la \\
\int_{v_0}^v \frac{(\chi-\al \la \D) dv'}{\wp(v') + \frac{d_2}{d_1}} + \int_{v_0}^v \frac{\D dv'}{\wp(v')-d_3} \bigg] \bigg\} \,.
\end{multline}
The second and third integral are of third kind because $f_1(v) = \left( \wp(v) + \frac{d_2}{d_1} \right)^{-1}$ and $f_2(v) = ( \wp(v)-d_3)^{-1}$ have simple poles. We will rewrite now $f_1$ and $f_2$ in terms the Weierstrass $\zeta$-function, which has a simple zero in $0$. The reason is, that $\zeta$ can easily be integrated in terms of the Weierstrass $\sigma$-function
\begin{align}
\int_{v_0}^v \zeta(v') dv' = \log \sigma(v) - \log \sigma(v_0)\,.
\end{align}
We only demonstrate the procedure for $f_1$ which is totally analogous to the procedure for $f_2$. The function $f_1$ has two simple poles $v_1$ and $v_2$ lying in the fundamental domain $\{ 2a\omega_1+2b\omega_2 \,|\, a,b \in [0,1)\}$, where $\omega_1$ and $\omega_2$ as above, with $\wp(v_1) = - \frac{d_2}{d_1} = \wp(v_2)$. An expansion of $f_1$ and $\wp(v)+\frac{d_2}{d_1}$ in neighbourhoods of $v_i$ yields
\begin{align}
f_1(v) & = a_{-1,i} (v-v_i)^{-1} + \text{ holomorphic part} \,,  \\
\wp(v)+\frac{d_2}{d_1} & = \wp'(v_i) (v-v_i) + \text{higher order terms}
\end{align}
for some constants $a_{-1,i}$. Now a comparison of coefficients gives
\begin{align}
1 & = f_1(v) \left(\wp(v)+\frac{d_2}{d_1}\right) \nonumber \\
& = a_{-1,i} \wp'(v_i) + \text{ higher order terms} \,, \nonumber \\
\text{that is } \quad a_{-1,i} &= ( \wp'(v_i) )^{-1} \,.
\end{align}
It follows that the function $f_1(v) - \sum_{i=1}^2 \frac{\zeta(v-v_i)}{\wp'(v_i)}$ is an elliptic function without poles and, therefore, equal to a constant $A$ which can be determined by $0=f(0)$. This yields $A = - \sum_{i=1}^2 \frac{\zeta(-v_i)}{\wp'(v_i)}$ and
\begin{align}
f_1(v) & = \sum_{i=1}^2 \frac{\zeta(v-v_i) + \zeta(v_i)}{\wp'(v_i)} \,.
\end{align}
Note that $\wp'(v_2) = \wp'(2\omega_j-v_1) = \wp'(-v_1) = - \wp'(v_1)$. The expression for $\wp'(v_1)$ can be determined by the differential equation $\wp'(x) = \pm \sqrt{4\wp(x)^3-g_2\wp(x)-g_3}$, where again the sign of the square root has to be chosen according to the sign of $\wp'$. 
\comment{
This is, $\wp'$ is negative on $[0,\omega_1]$ and positive on $[\omega_2,\omega_2+\omega_1]$. In the interval $[\omega_1,\omega_1+\omega_2]$, $\wp'$ is positively imaginary, on $[0,\omega_2]$, $\wp'$ is negatively imaginary.
}
An intergation of $f_1$ yields
\begin{multline}
\int_{v_0}^v f_1(v') dv' = \sum_{i=1}^2 \frac{1}{\wp'(v_i)} \bigg[ \zeta(v_i) (v-v_0) + \log(\sigma(v-v_i)) \\
 - \log(\sigma(v_0-v_i)) \bigg] \,.
\end{multline}
In the same way we can integrate $f_2$, where we assume that $v_3$ and $v_4$ are the simple poles of $f_2$ in the fundamental domain with $\wp(v_3) = d_3 = \wp(v_4)$. Summarized, $I_\theta$ is given by (cp.~\eqref{sol I_theta})
\begin{multline*}
I_\theta = \frac{|a_{3}|}{2 a_{3}} \bigg\{ (\al-\D) (v-v_0) \\
- \sum_{i=1}^4 \frac{a_{3}}{4 \chi \wp'(v_i)} \left( \zeta(v_i) (v-v_0) + \frac{\log(\sigma(v-v_i))}{\log(\sigma(v_0-v_i))} + 2\pi i k_i \right) \\
\cdot \big( \al^3 \la (\chi-\al \la \D) (\delta_{i1}+\delta_{i2}) + \D (\delta_{i3}+\delta_{i4}) \big) \bigg\} \,,
\end{multline*}
where $v = v(\gamma) = 2\gamma - \gamma_{\theta,\rm in}$ by \eqref{sol_y} and $v_0 = v(\gamma_0)$. The integers $k_i$ correspond to different branches of $\log$.

\section{Integration of hyperelliptic integrals of the third kind} \label{app:hyperelliptic}
In this appendix we will demonstrate the details of the integration method of the hyperelliptic integrals of third kind which appear in the $\r$ dependent part of the $\phi$ and $t$ motion \eqref{phi_int}, \eqref{t_int}. We will explain the procedure for the example of the integral $I_r$ in \eqref{phi_int}. As this integral is only hyperelliptic if we consider timelike geodesics, i.e. $\delta=1$ and $R$ has only simple zeros we assume in the following that this is the case.

Before we demonstrate the solution steps, we will summarize them for convinience:
\begin{enumerate}
\item Cast the expression under the square root in the standard form $u^5 + \sum_{i=0}^4 c_i u^i$ for some constants $c_i$.
\item Decompose the integrand (without the square root) in partial fractions.
\item If integrals of first or second kind are present, rewrite them as functions of $\gamma$.
\item Rewrite the integrals of third kind in terms of the canonical integral of third kind $dP(x_1,x_2)$.
\item Rewrite the canonical integrals of third kind in terms of the Kleinian sigma functions, such reducing them to functions of integrals of first kind.
\item Express the integrals of first kind in terms of $\gamma$ and assemble all parts.
\end{enumerate}

Let us start with eq. \eqref{I_r} 
\begin{equation*}
I_r = \int_{\r_0}^{\r} \frac{\al \left( \r'^2+\al^2 - \al \D \right) d\r'}{\Delta_{\r'} \sqrt{R}} \,,
\end{equation*}
which can analogously to section \ref{rmotion} be transformed to the standard form by $\r= \pm 1/u+\r_R$ with a zero $\r_R$ of $R$ and get
\begin{align}
I_r & = - \al \int_{u_0}^u \frac{(\pm \frac{1}{u}+\r_R)^2+ \al(\al-\D)}{\Delta_{\r=\pm 1/u + \r_R} \sqrt{u^{-6} c_{5} R_u}} \frac{du}{u^2} \nonumber \\
& = - \al \int_{u_0}^u \frac{[\r_R^2+\al(\al-\D)] u^2 \pm 2\r_Ru + 1}{\sqrt{c_{5}} \quad \Delta_u \sqrt{R_u}} |u^3| du \,,
\end{align}
where $R_u$ and $c_{5}$ are defined as in section \ref{rmotion} and $\frac{1}{u^4} \Delta_u = \Delta_{\r=\frac{1}{u}+\r_R}$, i.e.
\begin{multline}
\Delta_u = ( u^2(1-\la \r_R^2) \mp 2\r_R\la u-\la )(u^2(\r_R^2+\al^2)\pm 2\r_Ru+1) \\
\mp u^3-\r_Ru^4\,,
\end{multline}
which is a polynomial of degree 4 in $u$. Note that for geodesic motion, the coordinate $\r$ is always contained in an interval bounded by two adjacent real zeros of the polynomial $R$ or by a real zero and infinity. This implies that $u=\pm (\r-\r_R)^{-1}$ for a real zero $\r_R$ of $R$ does not change sign on the integration path and, therefore, we can neglect the absolute value of $u$ appearing in the integrand if we multiply the hole integral with $\text{sign}(u_0) = \frac{u_0}{|u_0|}$. Consequently 
\begin{align}
\frac{I_r}{\al} & = - \frac{|u_0|}{u_0} \int_{u_0}^u \frac{[\r_R^2+\al(\al-\D)] u^2 \pm 2\r_Ru + 1}{\sqrt{c_{5}} \quad \Delta_u \sqrt{R_u}} u^3 du \,. \label{I_r(u)}
\end{align} 

The expression for the integrand in \eqref{I_r(u)} can be simplified by a partial fraction decomposition
\begin{multline}
- \frac{\sqrt{c_{5}} |u_0|}{\al u_0} I_r = C_1 \int_{u_0}^u \frac{u du}{\sqrt{R_u}} + C_0 \int_{u_0}^u \frac{du}{\sqrt{R_u}} \\
+ \sum_{i=1}^4 C_{2,i} \int_{u_0}^u \frac{du}{(u-u_i) \sqrt{R_u}} \,, \label{r_partialfractions}
\end{multline}
where $u_i$, $1 \leq i \leq 4$ denote the zeros of $\Delta_u$ and $C_0, C_1, C_{2,i}$ are quite complicated expressions dependent on the parameters and the zero $\r_R$ of $R$ which may be calculated by a Computer Algebra System. 

The first two integrals in this expression are of first kind and can be solved analogous to section \ref{rmotion} \eqref{R_u}, i.e. 
\begin{align} 
\int_{u_0}^u \frac{u du}{\sqrt{R_u}} & = \sqrt{c_{5}} (\gamma - \gamma_{0}) \,, \label{sol_firstkind1}\\
\int_{u_0}^u \frac{du}{\sqrt{R_u}} & = \int_{u_0}^\infty \frac{du}{\sqrt{\tilde R_u}} + \int_\infty^u \frac{du}{\sqrt{\tilde R_u}} \nonumber \\
& = - f(\sqrt{c_{5}} \gamma_0 - \gamma_{\r,\rm in}) + f(\sqrt{c_{5}} \gamma - \gamma_{\r,\rm in}) \,. \label{sol_firstkind2}
\end{align}
where again $\gamma_{\r,\rm in} = \sqrt{c_{5}} \gamma_0 + \int_{u_0}^\infty \frac{u' du'}{\sqrt{\tilde R_{u'}}}$ with $u_0 = \pm \left(\r_0-\r_R \right)^{-1}$ only depends on the initial values $\gamma_0$ and $u_0$, and $f$ describes the $\theta$-divisor, i.e. $\sigma\left( (f(z),z)^t \right) = 0$. 

The four integrals in \eqref{r_partialfractions} containing $(u-u_i)^{-1}$ are in general of third kind and can be expressed in terms of the canonical integral of third kind $\int dP(x_1,x_2)$. The most simple construction of a differential of third kind  
\begin{equation}
dP(x_1,x_2) = \left( \frac{y+y_1}{x-x_1} - \frac{y+y_2}{x-x_2} \right) \frac{dx}{2y}  \label{def_dP}
\end{equation}
has simple poles in the points $(x_1,y_1)$ and $(x_2,y_2)$ of the Riemann surface of $y^2=g(x)$, $g$ a polynomial, with residual $+1$ and $-1$, respectively (cp. \cite{BuchstaberEnolskiiLeykin97, Baker07}). In particular, we get
\begin{align}
\int_{u_0}^u \frac{du}{(u-u_i) \sqrt{R_u}} = \frac{1}{+\sqrt{R_{u_i}}} \int_{u_0}^u dP(u_i^+,u_i^-) \,,
\end{align}
where $u_i^+ = (u_i, \sqrt{R_{u_i}})$ is the pole $u_i$ located on the positive branch of the square root and $u_i^- = (u_i, - \sqrt{R_{u_i}})$ is the pole $u_i$ located on the negative branch of the square root. Based on Riemann's vanishing theorem (see e.g. \cite{BuchstaberEnolskiiLeykin97}) the canonical differential of third kind can be expressed in terms of Kleinian $\sigma$ functions by
\begin{multline}
\int_{u_0}^u dP(u_i^+,u_i^-) = \frac{1}{2} \log \frac{\sigma( \int_{\infty}^u d\vec z - 2 \int_{\infty}^{u_i^+} d\vec z )}{\sigma( \int_{\infty}^u d\vec z - 2 \int_{\infty}^{u_i^-} d\vec z ) } \\ 
- \frac{1}{2} \log \frac{\sigma( \int_{\infty}^{u_0} d\vec z - 2 \int_{\infty}^{u_i^+} d\vec z )}{\sigma( \int_{\infty}^{u_0} d\vec z - 2 \int_{\infty}^{u_i^-} d\vec z )} - \left( \int_{u_0}^u d\vec z \right)^t \left( \int_{u_i^-}^{u_i^+} d\vec r\right) \,, \label{sol_dP}
\end{multline} 
where $d\vec z$ is the vector of canonical differentials of the first kind and $d\vec r$ the vector of canonical differentials of the second kind
\begin{align}
dz_i & = \frac{u^{(i-1)} du}{\sqrt{R_u}}\,, \quad i=1,2,  \label{def_holomorphic} \\
dr_i & = \sum_{k=i}^{5-i} (k+1-i) \frac{c_{k+1+i}}{c_{5}} \frac{u^k du}{4 \sqrt{R_u}} \,, \quad i=1,2.
\end{align}
Finally, we rewrite \eqref{sol_dP} in terms of the affine parameter $\gamma$. By \eqref{sol_firstkind1} and \eqref{sol_firstkind2} we can express $\int_{u_0}^u d\vec z$ as well as the arguments of the $\sigma$ functions $\int_{\infty}^u d\vec z = \int_{u_0}^u d\vec z - \int_{u_0}^\infty d\vec z$ as functions of $\gamma$. If we define $w = w(\gamma) = \sqrt{c_{5}} \gamma - \gamma_{\r,\rm in}$ and $w_0 =w(\gamma_0)$ the integral $I_r$ is given by (cp.~\eqref{sol I_r})
\begin{multline*}
I_r = - \frac{\al u_0}{\sqrt{c_{5}} |u_0|} \bigg\{ C_1 (w-w_0)  + C_0 ( f(w)-f(w_0) ) \\
+ \sum_{i=1}^4 \frac{C_{2,i}}{\sqrt{R_{u_i}}} \bigg[ \frac{1}{2} \log \frac{\sigma(W^+(w))}{\sigma( W^-(w) )} - \frac{1}{2} \log \frac{\sigma( W^+(w_0))}{\sigma( W^-(w_0) )} \\
- \big( f(w)-f(w_0), w-w_0 \big) \left( \int_{u_i^-}^{u_i^+} d\vec r\right) \bigg] \bigg\} \,,
\end{multline*}
where $W^+(w):=(f(w),w)^t - 2 \int_\infty^{u_i^+} d\vec z$ and $W^-(w) = (f(w),w)^t - 2 \int_\infty^{u_i^-} d\vec z$.


\bibliographystyle{apsrev}
\bibliography{Kerr-deSitter}

\end{document}